%% file: main_arxiv.tex
\documentclass[11pt]{article}

\usepackage{amssymb,amsmath,amsthm,amsfonts}

\usepackage[T1]{fontenc}
\usepackage[tt=false, type1=true]{libertine}
\usepackage[varqu]{zi4}
\usepackage[libertine]{newtxmath}

\usepackage[margin=1in]{geometry}
\usepackage{enumitem}
\setlist{topsep=2pt,leftmargin=*}

\usepackage{hyperref}
\hypersetup{
    colorlinks,
    allcolors=myblue
}

\usepackage[sortcites,sorting=nyt,style=alphabetic,maxbibnames=8]{biblatex}
\usepackage[ruled]{algorithm2e} 
    
    \SetAlFnt{\small}
    \SetAlCapFnt{\small}
    \SetAlCapNameFnt{\small}
    \SetAlCapHSkip{0pt} 
    \IncMargin{-\parindent}

\input{commands}

\newtheorem{theorem}{Theorem}[section]
\newtheorem{corollary}[theorem]{Corollary}
\newtheorem{lemma}[theorem]{Lemma}

\newtheorem{remark}[theorem]{Remark}

\theoremstyle{definition}

\newtheorem{definition}{Definition}[section]

\newcommand{\citet}[1]{\textcite{#1}}

\title{Beyond Worst-Case Online Allocation via Dynamic Max-min Fairness}

\author{
    Giannis Fikioris
    \thanks{Supported in part by the Department of Defense (DoD) through the National Defense Science \& Engineering Graduate (NDSEG) Fellowship, the Onassis Foundation -- Scholarship ID: F ZS 068-1/2022-2023, the Google PhD Fellowship, and AFOSR grant FA9550-23-1-0068.}\\
    Cornell University\thanks{Part of this work was done while Giannis Fikioris was visiting Google as a Student Researcher in the Market Algorithms team.}\\
    \texttt{gfikioris@cs.cornell.edu}
    \and
    Siddhartha Banerjee
    \thanks{Supported in part by AFOSR grant FA9550-23-1-0068, ARO MURI grant W911NF-19-1-0217, NSF grants ECCS-1847393 and CNS-195599, and the Simons Institute for the Theory of Computing.}\\
    Cornell University\\
    \texttt{sbanerjee@cornell.edu}
    \and
    \'Eva Tardos
    \thanks{Supported in part by AFOSR grants FA9550-19-1-0183 and FA9550-23-1-0068.}\\
    Cornell University\\
    \texttt{eva.tardos@cornell.edu}
}

\date{\vspace{-20pt}}

\begin{document}

\let\oldabovedisplayskip\abovedisplayskip
    \setlength{\abovedisplayskip}{\oldabovedisplayskip-3pt}
\let\oldbelowdisplayskip\belowdisplayskip
    \setlength{\belowdisplayskip}{\oldbelowdisplayskip-3pt}
\let\oldtextfloatsep\textfloatsep
    \setlength{\textfloatsep}{\oldtextfloatsep-10pt}

\maketitle{}

\begin{abstract}
    \input{body/0.abstract}
\end{abstract}

\input{body/1.intro}

\input{body/2.related_work}

\input{body/3.definitions}

\input{body/4.ideal}

\input{body/5.single}
\input{body/6.mult}

\input{body/7.imposs}

\input{body/8.conclusion}

\printbibliography{}

\appendix
\input{appendixes/5.single}
\input{appendixes/6.mult}
\input{appendixes/7.imposs}

\end{document}

%% file: commands.tex
\usepackage{enumitem,xspace,nicefrac,tcolorbox,xifthen,tikz,subcaption,enumitem,physics,suffix,wrapfig,titlesec}
\usetikzlibrary{shapes, arrows, positioning}

\usepackage[bb=dsserif]{mathalpha}

\usepackage{xcolor}
    \definecolor{myred}{HTML}{ea4335}
    \definecolor{mygreen}{HTML}{41a756}
    \definecolor{myblue}{HTML}{4285f4}

\usepackage[capitalize]{cleveref}

\usepackage[showdeletions]{color-edits} 
\definecolor{@gray}{HTML}{edc3c5}
\addauthor[Eva]{et}{orange}
\addauthor[Giannis]{gf}{teal}
\addauthor[Sid]{sb}{violet}
\addauthor[\textbf{TODO}]{todo}{purple}

\renewcommand{\paragraph}[1]{\smallskip\noindent\textbf{#1.}}

\let\a\alpha
\let\b\beta
\let\g\gamma
\let\e\varepsilon
\let\d\delta
\let\l\lambda

\newcommand{\N}{\mathbb{N}}
\newcommand{\R}{\mathbb{R}}

\DeclareMathOperator*{\argmin}{\arg\min}
\DeclareMathOperator*{\E}{\mathbb{E}}
\DeclareMathOperator*{\PR}{\mathbb{P}}

\newcommand{\Ex}[2][]{\E_{#1}\left[#2\right]}
\newcommand{\ExC}[3][]{\E_{#1}\left[#2\phantom{\Big|}\middle|\phantom{\Big|}#3\right]}
\renewcommand{\Pr}[2][]{\PR_{#1}\left[#2\right]}

\newcommand{\One}[1]{\mathbb{1}\left[#1\right]}
\usepackage{xifthen}
\newcommand{\Line}[4]{%
    #1&%
    \;\ifthenelse{\isempty{#2}}{\phantom{=}}{#2}\;%
    #3%
    \ifthenelse{\isempty{#4}}{}{&&\qquad\left(#4\right)}%
}

\newcommand{\kmax}{k_{\max}}
\newcommand{\vstar}{v^{\star}}
\newcommand{\rstar}{\rho^{\star}}
\newcommand{\req}{\mathtt{Req}}

\newcommand{\blk}{\mathtt{Blk}}
\newcommand{\Blk}{\blk}
\newcommand{\calS}{\mathcal{S}}
\newcommand{\calF}{\mathcal{F}}

\newcommand{\dmmf}{\textsf{DMMF}\xspace}

%% file: body/0.abstract.tex
We study the allocation of shared resources over multiple rounds among competing agents, via the \emph{dynamic max-min fair} (DMMF) mechanism: the good in each round is allocated to the requesting agent with the least number of allocations received to date.
We show that in large markets when an agent has i.i.d. values across rounds, under mild distributional assumptions (e.g., bounded PDF function), the DMMF mechanism allows each agent to realize a $1 - o(1)$ fraction of her ideal utility -- her highest achievable utility given her nominal share of resources.
This guarantee holds under arbitrary behavior by other agents and is achieved by characterizing the agent's utility under a rich space of strategies, wherein an agent can tune how aggressive to be in requesting the item.
Our techniques also allow us to handle settings where an agent's values are correlated across rounds, thereby allowing an adversary to predict and block her future values.
By tuning the aggressiveness, an agent can guarantee $\Omega(\gamma)$ fraction of her ideal utility, where $\gamma\in [0, 1]$ is a parameter that quantifies dependence across rounds (with $\gamma = 1$ indicating full independence and lower values indicating more correlation).
Finally, we extend our efficiency results to the case of reusable resources, where an agent might need to hold the item over multiple rounds to receive utility.
Our results subsume previous guarantees obtained using a more complicated mechanism proving a half ideal utility guarantee under i.i.d. values sampled from worst-case distributions.

%% file: body/1.intro.tex
\section{Introduction}
\label{sec:intro}

Consider a principal (e.g., a system administrator) with some public resource repeatedly shared between a set of self-interested agents over multiple rounds.
The public resource could be a scientific instrument (e.g., mass spectrometer, gene sequencer, microscope, telescope, etc.), a computing cluster, donations for food banks, or something more commonplace like the nice big conference room.
The principal in charge of allocating these resources has to face up to multiple challenges:
\begin{itemize}
    \item \emph{Resource constraints:}
    A microscope can be used for one experiment at a time; a telescope can focus on one point in the sky; the conference room cannot accommodate two meetings simultaneously.

    \item \emph{Stochastic and uncertain demands:}
    Astronomical events like supernovae are unpredictable, appearing randomly and independently of past events.
    Conversely, higher than average traffic at a food bank might indicate higher demand the next day, making demands stochastic but not independent.
    Finally, these stochastic demands are often far from the worst-case examples assumed in theory.
    
    \item \emph{Inefficiency:} 
    If, for the sake of fairness, the nice conference room cannot be booked by the same person twice a week, it would be frustrating to look at its empty chairs while presenting from a nearby room with a faulty projector.

    \item \emph{Unfair use:}
    At any point, it should not be that a team's experiment has been rejected from the computing cluster $9$ times out of $10$ so far, while another team's has been accepted $19$ times out of $20$.
    Situations like this must be avoided even when one team tries to game the allocation rules.

    \item \emph{Simple non-monetary mechanism:}
    The researcher of a research lab cannot be charged real money to use the lab's resources.
    The allocation rule has to be understandable and depend only on the agents' preferences, which are expressed in non-monetary value and thus cannot be compared.
\end{itemize}

We model this allocation problem over $n$ agents, where agent $i$ has a private value for being allocated in each of the $T$ rounds.
The principal wants to be fair and cannot use money, so we cannot use auction-based mechanisms like VCG, which are truthful and efficient.
This leaves the principal with more ad-hoc approaches, which may be inefficient or unfair.
For example, priority-based policies work well for those with high priority but are unfair to others; round-robin or uniformly random allocation policies are fair but inefficient as they may allocate to an agent when she has no need.

An allocating mechanism needs to handle agents with different priorities.
Agent $i$'s \textit{fair share} $\a_i$ (with $\sum_i \a_i = 1$) quantifies that `priority' in the system.
For example, given $n$ research labs sharing an instrument, equalizing the fair shares ($\a_i = 1/n$) means the mechanism should favor all agents equally.
In general, however, research labs might have different sizes and require unequal use of the instrument; the (exogenously specified) fair shares then direct the mechanism to treat different labs unequally, proportional to their fair shares.

The agents' fair shares subsequently define our fairness metric.
For any agent $i$, her \emph{ideal utility} $\vstar_i(\a_i)$ is the maximum expected per-round utility she can get if there is no competition but constrained to receive the resource only on an $\a_i$ fraction of the rounds.
Consequently, a mechanism tries to maximize the fraction of ideal utility an agent can get.
Past work has focused on \textit{robust} utility guarantees, i.e., guarantees that hold under adversarial behavior by other agents.
This focus has been due to the uncertainty in how other agents behave: in practice, agents are not perfect optimizers and might not best respond (e.g., follow an equilibrium) as most theoretical results assume.
\citet{DBLP:conf/sigecom/GorokhBI21} first defined ideal utility, which subsequently became the standard fairness metric for this setting.
They show that under a repeated first-price auction with artificial currencies, each agent can realize at least half of her ideal utility robustly.
\citet{DBLP:conf/sigecom/BanerjeeFT23} generalize this result to the allocation of reusable resources, i.e., those that agents might use for multiple rounds, and prove that in such settings, this guarantee is tight.

\subsection{Our mechanism and results}
We study a simple and widely used mechanism, \emph{Dynamic Max-Min Fair} (\dmmf), rather than pseudo-markets with artificial currencies.
The \dmmf mechanism (\cref{algo:algo}) is a dynamic version of the classical max-min fair policy: \emph{in any round, given the agents requesting for the resource in that round, \dmmf allocates to the agent with the least number of past allocations, normalized by her fair share}.
We focus on the \dmmf mechanism because of its advantages over pseudo-market mechanisms.
Instead of declaring a real-valued bid in every round, agents in \dmmf declare only if they want the item in that round, making it easier to implement in practice.
In addition, \dmmf is work-conserving: if anyone wants the item, then someone gets it; in contrast, in pseudo-markets, agents can run out of budget early.

When agent $i$'s values are i.i.d. across rounds, we prove that she can guarantee at least $\frac{1}{2-\a_i} \vstar_i(\a_i) T$ total utility, irrespective of how others bid and her distribution.
This offers a slight improvement over the bounds of \cite{DBLP:conf/sigecom/GorokhBI21,DBLP:conf/sigecom/BanerjeeFT23}, who focus only on such worst-case guarantees, i.e., under arbitrary value distributions.
This $\frac{1}{2}$ guarantee is close to tight: no mechanism can guarantee every agent more than $1 - \frac{1}{e} \approx 0.63$ fraction of her ideal utility for every value distribution (see \cref{thm:imp:ind})\footnote{For reusable resource the half ideal utility guarantee is tight \cite{DBLP:conf/sigecom/BanerjeeFT23}.}.

While the above guarantee is significant, given the upper bound, focusing on worst-case guarantees may be too pessimistic for real-world applications.
Motivated by this, we focus on improved utility guarantees that depend on the agent's value distribution.
To the best of our knowledge, this is the first time such fine-grained analysis has been applied to similar settings.
This improvement starts by noticing that since $\vstar_i(\cdot)$ only depends on agent $i$'s value distribution, $\vstar_i(\b)$ can be computed for any $\b\in[0,1]$.
This motivates a rich space of strategies we refer to as \textbf{$\b$-aggressive strategies} (\cref{def:b-aggresive}), where an agent with fair share $\a_i$ acts as if her share is $\b$.
By setting $\b > \a_i$ (resp. $\b < \a_i$), the agent can become more (resp. less) assertive than she is entailed; however, her fair share $\a_i$ still determines her normalization in \dmmf as well as her benchmark $\vstar_i(\a_i)$, which is the maximum utility she should hope for.
We get the following results, which exhibit the generality of \dmmf and $\b$-aggressive policies.

\paragraph{Main focus: Fine-grained robust guarantees}
Our main focus is on agents with values i.i.d. across rounds.
We show (\cref{thm:single:guar}) that under \dmmf an agent $i$ with a $\b$-aggressive strategy, no matter how other agents behave, is guaranteed utility \emph{up to any round $\tau$} satisfying\footnote{For $m \in \N$ we denote $[m] = \{1, 2, \ldots, m\}$.}
\begin{equation}\label{eq:intro:general_bound}
    \sum_{t \in [\tau]} \Ex{\big.U_i[t]}
    \ge
    \frac{\a_i}{\a_i + \b - \a_i\b}
    \vstar_i(\b) \tau
    -
    \order{1}
\end{equation}
where $U_i[t]$ is agent $i$'s utility in round $t$.
This bound is not directly useful as it is expressed in terms of $\vstar_i(\b)$ and not the ideal utility $\vstar_i(\a_i)$.
However, it gives us a parameterized set of guarantees for all $\b$, which hold without the \dmmf mechanism having any knowledge about the agents (e.g., their distributions) and can be optimized by the agent to get setting-specific results.
We showcase this with multiple examples that exhibit the richness of $\b$-aggressive policies.
\begin{itemize}
    \item \textbf{Uniform distribution}:
    A demonstration of \cref{eq:intro:general_bound} is when agent $i$'s value is uniformly distributed and $\a_i$ is small ($\leq 1/8$): setting $\b = \sqrt{2\a_i}$ guarantees a $1 - \sqrt{2\a_i}$ fraction of her total ideal utility $\vstar_i(\a_i) \tau$.
    Note that the agent requests much more aggressively than her fair share suggests (since $\sqrt{2\a_i}\gg \a_i$).
    However, by doing so, she improves her guaranteed fraction of ideal utility from $1/2$ to $1 - \sqrt{2\a_i}$, which is \textbf{close to optimal} as $\a_i \to 0$ (i.e., in a `large market' setting). 
    
    \item \textbf{Bounded distributions} (\cref{cor:indep:density}):
    More generally, when an agent's value distribution pdf is lower and upper bounded by $\l_1, \l_2$, respectively, then choosing $\b = \sqrt{\nicefrac{2\a_i\l_1}{\l_2}}$ guarantees a $1 - \sqrt{\nicefrac{2\a_i\l_2}{\l_1}}$ fraction of her ideal utility.
    Similar to before, in large market settings, this is near optimal.

    \item \textbf{Bernoulli distributions} (\cref{cor:indep:bernoulli}): If the agent's value in each round is Bernoulli distributed (i.e., ON/OFF, or more generally, high/low), with probability of ON states being $p$, then by using a $p$-aggressive policy, the agent can realize a $\frac{\max\{\a_i, p\}}{\a_i + p - \a_i p}$ fraction of her ideal utility.
    In particular, if either $p \gg \a_i$ or $\a_i \gg p$ (i.e., the probability of ON is higher/lower than her fair share), then the above guarantee approaches $\vstar(\a_i) \tau - \order{1}$.
    This result shows the brittleness of worst-case results, which always concern Bernoulli distributions.
    For example, \cref{thm:imp:ind}, which shows no mechanism can guarantee more than $1 - 1/e$ fraction of ideal utility, uses Bernoulli distributions with exactly $p = \a_i$.

    \item \textbf{Worst-case distributions} (\cref{cor:indep:worst}):
    If we want results agnostic of the agent's distribution, then we can still set $\b = \a_i$ in~\cref{eq:intro:general_bound}, to realize a $\frac{1}{2-\a_i}$ fraction of $\vstar_i(\a_i) \tau$.
    This shows an agent can always guarantee a constant fraction of her ideal utility under \dmmf.
    On the negative side, this is a much smaller fraction of ideal utility than our distribution-aware guarantees.
    In addition, if every agent $i$ uses an $\a_i$-aggressive strategy (as suggested by the existing worst-case robustness guarantees), then a constant fraction of rounds are uncontested.
    This highlights the importance of our distribution-specific results and justifies the use of \dmmf in large markets.
\end{itemize}

Even in the worst-case setting (which was the sole focus of previous work), our mechanism improves in two critical ways: it gives \emph{anytime} guarantees (i.e., holding in every round $\tau$, not just the final round $T$) and has a better additive error ($\order{1}$ instead of $\order*{\sqrt T}$).
While these might seem negligible at first, they are essential in practice since they allow \dmmf to handle \emph{dynamic settings/populations}: at any round, the mechanism can be reset (e.g., to accommodate new agents) with minimal additional loss (see \cref{rem:single:dynamic_pop} for details).
In contrast, to minimize losses in pseudo-market mechanisms, the principal can only restart after a pre-determined final round, with higher additive loss.
This last point, along with the simplicity of \dmmf, motivates the use of our mechanism besides the distribution-aware utility guarantees, which we believe might also be achievable for pseudo-market mechanisms.

\paragraph{Price of Anarchy}
In symmetric settings (i.e., equal shares and identical value distributions across agents), robust guarantees on the ideal utility also imply bounds on the Price of Anarchy -- the inverse of the fraction of the optimal welfare achieved under any equilibrium. 
In particular, if every agent can robustly guarantee an $c \le 1$ fraction of her ideal utility, then the Price of Anarchy is at most $1/c$.
This means that in any Nash equilibrium, \dmmf is fair and efficient in large markets under mild assumptions on the value distributions (\cref{cor:single:welfare}). 
    
\paragraph{Correlated rounds}
We also extend our guarantees to handle time-dependent models.
We consider settings where agent $i$'s values are drawn via a hidden Markov model, with the underlying Markov chain over latent states parameterized by a \emph{decorrelation parameter} $\g_i \in (0,1]$.
We define this formally in~\cref{sec:setting}, but at a high level, $\g_i = 1$ indicates independent values, while as $\g_i$ decreases, states become more correlated across rounds. 
Previous work~\cite{DBLP:conf/sigecom/GorokhBI21,DBLP:conf/sigecom/BanerjeeFT23,gorokh2017,guo2010} in this setting focuses exclusively on i.i.d. values, and to the best of our knowledge, we are the first to consider values that can be correlated across rounds.
We show that under the same \dmmf mechanism, by using a $\b$-aggressive policy with appropriately chosen $\b \le \a_i$, agent $i$ can always guarantee a $\Omega(\g_i)$ fraction of her ideal utility (see \cref{cor:single:big_gamma,cor:single:small_gamma} for the detailed bounds).
We also provide a utility guarantee in \cref{thm:single:state_independent} that is irrespective of $\g_i$ and instead depends on the similarity of the agent's values in different states.

\paragraph{Multi-Round Demands}
In \cref{sec:mult}, we consider the setting of~\cite{DBLP:conf/sigecom/BanerjeeFT23} with multi-round demands (agents need the resource for intervals of rounds) with independent value distribution.
We recover our per-agent utility guarantees in~\cref{eq:intro:general_bound} (up to $\order*{\sqrt T}$ terms), with two caveats:
$(i)$ \dmmf needs to know the horizon $T$ and use a parameter $r$ to limit how many multi-round demands agents get, and
$(ii)$ the parameter $r$ is tuned according to $\b$ (see \cref{thm:mult:guar}). 
This means that to get distribution-aware guarantees, we need to have symmetric agents, and the principal needs to know the distribution and tune $r$ accordingly.
Setting $r$ independent of the value distributions also lets us recover the $\frac{1}{2}$ robustness guarantee under \dmmf, matching the guarantee for pseudo-markets shown in \cite{DBLP:conf/sigecom/BanerjeeFT23} (who further showed this is tight for \emph{any} mechanism under worst-case distributions).
    
\paragraph{Impossibility results}
Finally, in \cref{sec:imp}, we study impossibility results.
For single-round demands, we first show that no mechanism can guarantee every agent more than a $1 - \frac{1}{e}$ fraction of her ideal utility (\cref{thm:imp:ind}).
For the \dmmf mechanism, we show that agent $i$ cannot guarantee more than a $\frac{\g_i}{1+\g_i}$ fraction of her ideal utility, making our guarantee for dependant values tight up to a constant factor (\cref{cor:imp:markov}).
For multi-round demands, we show that the use of the parameter $r$ in \dmmf, as well as the need to tune it according to the aggressiveness level $\b$, is necessary: our bound in \cref{thm:mult:guar} is tight for every $r$ and $\b$, up to additive $\order*{\sqrt T}$ terms (\cref{thm:mult:imp}).
We believe that studying the limitations of the \dmmf mechanism is important, as it is often employed in practice \cite{DBLP:conf/osdi/VuppalapatiF0CK23}.

\subsection{Challenges and Technical Novelty}

First, we establish an invariant property (i.e., a `conservation law') that holds under \dmmf on \emph{any} sample path.
Our guarantees stem from such a property given in \cref{lem:single:dmmf}, where we show that for any agent $i$ and interval $[1, t]$, the ratio of the number of rounds that $i$ wins to the number of rounds in which she is `blocked,' is at least $\alpha_i/(1-\alpha_i)$ (i.e., the ratio of $i$'s fair share to that of everyone else).
While this should seem natural under \dmmf, formalizing it is subtle; it is critical to define `blocked' in a counterfactual sense: an agent is blocked in round $t$, if she would not receive the item conditioned on requesting, regardless of if she requested or not.

The above invariant implies that agent $i$ cannot win significantly fewer rounds relative to the number of her requests and fair share.
In the case of i.i.d. values, this share of winning rounds directly implies utility guarantees using a $\b$-aggressive strategy.
However, increasing $\b$ implies winning more rounds but being blocked for more rounds.
In \cref{ssec:single:indep}, we show how to pick $\b$ to get the right trade-off between the above two competing objectives.

When agent $i$'s values are correlated across rounds, other agents might be able to infer information about her values from past requests, being able to block her high-valued requests.
For a fixed decorrelation parameter $\g_i$, agent $i$ can use the $\a_i$-aggressive strategy to get some utility guarantees depending on $\g_i$; however, this guarantee turns out to be a $\Theta(\g_i^2)$ fraction of her ideal utility.
This is the outcome of information leakage when $\g_i < 1$: the other agents can now block agent $i$'s high-valued rounds \emph{without} blocking her low-valued rounds.
This has a twofold effect: agent $i$ loses utility \emph{and} wins rounds worth less, leading to quadratic dependence on $\g_i$.
In \cref{ssec:single:correl}, we show that agent $i$ can control this leakage by following a $\b$-aggressive strategy for $\b < \a_i$ and get $\Omega(\g_i)$ utility bounds.
This use of $\b$-aggressive strategies is another evidence of their importance.
In fact, we believe that they might be of independent interest outside of this paper.

%% file: body/2.related_work.tex
\subsection{Related Literature}

Our work combines three separate considerations: dynamics (the `online' aspect), competing objectives (overall efficiency vs. individual guarantees), and competition (private values and robust guarantees).
These topics, and combinations thereof, have a long history of study

The closest work to ours is that of~\cite{DBLP:conf/sigecom/GorokhBI21,DBLP:conf/sigecom/BanerjeeFT23} who study robust guarantees of pseudo-market mechanisms when agents have i.i.d. valuations.
They only focus on guarantees for worst-case distributions and prove that an agent can always guarantee a $1/2$ fraction of her ideal utility; \cite{DBLP:conf/sigecom/GorokhBI21} study only single-round demands while \cite{DBLP:conf/sigecom/BanerjeeFT23} extend this guarantee to multi-round demands which they show is tight for any mechanism.
These works themselves build on earlier work on using pseudo-markets for online allocation~\cite{cavallo2014incentive,jackson,guo2010,gorokh2017}.
An important distinction, though, is that these earlier works study equilibrium outcomes in Bayesian settings instead of robust outcomes; moreover, all these works require agents to have i.i.d. valuations.
More recently, \cite{elokda2023self} study pseudo-markets where agent's values can be correlated across time; they prove the existence of a static equilibrium, and empirically study its performance in a few practical examples.
The interest in pseudo-markets stems from successful real-world deployments for course allocation~\cite{budish2017course}, food banks~\cite{walsh2014allocation,prendergast2022allocation} and cloud computing~\cite{dawson2013reserving,vasudevan2016customizable}.
Our results suggest that simpler mechanisms may suffice in some of these settings, and also that strong performance guarantees are achievable under more relaxed models of agent utility.

Next, we analyze work that is related to this paper but is further afield.

First, we have the study of fairness-efficiency tradeoffs in online allocation of divisible resources, starting with~\citet{kelly1998rate}, and 
with a recent renewal of interest due to data-centers~\cite{bonald2001impact,bonald2006queueing} and cloud computing~\cite{DBLP:conf/nsdi/GhodsiZHKSS10, DBLP:conf/osdi/ShueFS12, DBLP:conf/osdi/GrandlKRAK16,DBLP:journals/teco/ParkesP015}.
In particular, \cite{DBLP:journals/teco/ParkesP015} study the static allocation of multiple divisible resources and show that low social welfare is needed where requiring properties strategy-proofness, envy-freeness, or sharing incentives; they also study envy-freeness for indivisible resources.
Most of these works consider static settings or assume requests of equal importance.
Recent work has focused on extending these results to time-varying utilities~\cite{DBLP:conf/sigmetrics/FreemanZCL18,DBLP:journals/corr/FikiorisAT21,DBLP:conf/osdi/VuppalapatiF0CK23} or dynamic arrival of agents \cite{DBLP:conf/atal/KashPS13,DBLP:journals/pomacs/ImMMP20,DBLP:journals/mor/VardiPF22}, analyzing the incentive and efficiency properties of simple \dmmf-like policies for divisible resource allocation.
While the need for rationing in our setting makes it incompatible with these works, our idea of proving robust guarantees against per-agent benchmarks may also prove useful in those models.

Another stream of work has looked at efficiency-fairness tradeoffs in online allocation problems, but without incentive constraints (i.e., where values are known to the principal at the start of a round).
In this context, prior work has considered both stochastic~\cite{sinclair2022sequential,banerjee2023online,yin2022optimal} as well as adversarial~\cite{freeman2017fair,banerjee2022online,kawase2022online,barman2022universal,gkatzelis2023} utility models, and characterized various tradeoffs between different objectives.
This is also related to a much larger body of work looking at similar questions in static (i.e., one-shot) settings. A particularly relevant example of this for us is recent work by~\citet{DBLP:conf/sigecom/BabaioffEF21,DBLP:conf/wine/BabaioffEF22}, who study the one-shot allocation of indivisible goods among agents with exogenous shares, and arbitrary values, and prove that every agent can guarantee a constant fraction of her \textit{Any Price Share} (APS), a quantity that could be described as a worst-case ideal utility. While their setting and results are, for the moment, incomparable to the one we consider, understanding how they are related is an important future direction.

Finally, considering incentives in online allocation places our work in the literature on dynamic mechanisms with state.
The challenge here arises from factors that couple decisions across time; e.g.,
incomplete information and learning~\cite{Iyer2014,devanur2015perfect,nekipelov2015econometrics,DBLP:conf/nips/BranzeiHPSW24,DBLP:conf/aaai/TamuzVZ18},
constraints like budget limits~\cite{gummadi2012repeated,leme2012sequential,balseiro2015repeated}, 
queues in a network allocation~\cite{DBLP:conf/sigecom/GaitondeT20,DBLP:conf/sigecom/GaitondeT21}, 
stochastic fluctuations~\cite{gershkov2009dynamic,bergemann2010dynamic},
etc. 
Analyzing equilibria in such repeated settings can be difficult, and our work contributes to the small but growing set of techniques available for this; in particular, our idea of obtaining results that interpolate between constant factors in worst-case settings to near-optimal efficiency in more realistic settings is reminiscent of the work on Price of Anarchy in large markets~\cite{feldman2016price}.

%% file: body/3.definitions.tex
\section{Online Public Resource Allocation -- Single Round Demands}
\label{sec:setting}

This section formalizes our model, mechanism, and benchmark when agents have single-round demands.
We develop our results in the general setting where an agent's value in different rounds can be correlated.
The results for independent demands follow as special cases.

\subsection{Agent Model}

The principal has a single resource in each of $T$ rounds to allocate between $n$ agents.
Agents have \emph{single-round demands}: in every round $t\in [T]$, each agent $i\in [n]$ samples a random \emph{value} $V_{i}[t] \in \R_+$: if she is allocated the item on round $t$ then she gets $V_{i}[t]$ utility in that round.
In \cref{sec:mult}, we extend this to \emph{multi-round} demands, where an agent may need the resource for more than one contiguous round to satisfy her demand.

We assume that each agent's value is drawn from an underlying distribution, which is \emph{independent across agents} but possibly \emph{dependent across rounds}.
The dependence follows a Hidden Markov Model: agent $i$ in round $t$ has some latent state $S_i[t]$ drawn from some discrete and finite state-space $\mathcal S_i$.
In round $t$, agent $i$'s latent state $S_i[t]$ is generated based on $S_i[t-1]$ via a time-invariant Markov chain with transition kernel $p_i(\cdot,\cdot)$ (i.e., the transition probability from state $s'$ to $s$ is given by $p_i(s',s)$).
We assume this Markov chain is ergodic (i.e., irreducible and aperiodic), thereby converging to a unique stationary distribution $\pi_i$ from any starting state.
We assume the state in the first round $S_i[1]$ is sampled from the stationary distribution $\pi_i$\footnote{Assuming that $S_i[1]$ is arbitrarily picked only incurs a constant additive error as long as the mixing time is independent of $T$.}.

Each state $s_i \in \mathcal S_i$ is associated with a value distribution $\mathcal F_{i,s_i}$; given state $S_i[t]$ of agent $i$ in round $t$, $V_i[t]$ is sampled from $\mathcal F_{i,S_i[t]}$.
This sampling process is independent of other rounds: formally, \textit{conditioned on the value of $S_i[t]$}, the value $V_i[t]$ is independent of $V_1[1], \ldots, V_i[t-1]$ and $S_1[1], \ldots, S_i[t-1]$.
We define $\mathcal F_i$ as the induced value distribution when state $S_i[t]$ is sampled from $\pi_i$, the stationary distribution of agent $i$'s Markov chain; formally, $\mathcal F_i = \sum_s \pi_i(s) \mathcal F_{i,s}$.

Finally, our main robustness guarantees depend on a parameter that quantifies the amount of decorrelation across rounds.
For agent $i$, we define $\g_i \in [0, 1]$ as $\g_i = \min_{s',s\in \calS_i}\big(\nicefrac{p_i(s',s)}{\pi_i(s)}\big)$.
Thus $\g_i = 1$ corresponds to the i.i.d. value distribution case, since $p_i(s', s) = \pi_i(s)$ and therefore in each round $t$ the state $S_i[t]$ is sampled from $\pi_i$, regardless of previous rounds.
In contrast, if $\g_i = 0$, then some states are perfectly correlated with the following ones.

\subsection{The Dynamic Max-Min Fair Mechanism}
\label{ssec:mechanism}

Our work focuses on understanding the performance of the \textit{Dynamic Max-Min Fair (\dmmf) mechanism}.
To define \dmmf, we first assume that each agent $i$ has an associated \emph{fair share} $\a_i$, with $\sum_{i\in[n]} \a_i = 1$.
Fair shares are exogenously defined (i.e., before we run the mechanism) and represent the fraction of the resource that the principal wants each agent to have (which then determines her associated \emph{ideal utility}, see \cref{ssec:ideal}).

The \dmmf mechanism (\cref{algo:algo}) aims to keep the rounds allocated to every agent $i$ close to her fair share $\alpha_i$, compared to other agents.
Formally, if $A_i[t-1]$ is the allocation of agent $i$ up to round $t-1$ (i.e., number of rounds won), the agent with the smallest $\frac{A_i[t-1] + 1}{\a_i}$ who requests the item wins in round $t$ (we break ties arbitrarily).

Reemphasizing our comment in the introduction: \dmmf is a very simple deterministic policy that needs no knowledge of the agents' latent state, value distribution, or even the horizon $T$.

\begin{algorithm}[t]
\DontPrintSemicolon
\caption{Dynamic Max-Min Fair Algorithm}
\label{algo:algo}
\KwIn{Agents $[n]$, agents' fair shares $\{\a_i\}_{i\in[n]}$} 
\textbf{Initialize} agent allocations $A_i[0] = 0 \quad \forall \, i\in [n]$\;
\For{$t = 1, 2, \ldots$}
{
    Let $\mathcal N_t =$ agents who request for the item in round $t$\;

    \lIf{$|\mathcal N_t| = 0$}{
        Set $i^* = \varnothing$%
        \tcp*[f]{No winner}%
    }\lElse{
        Allocate to an agent $i^* \in \argmin_{i \in \mathcal N_t}\frac{A_i[t-1] + 1}{\a_i}$%
        \tcp*[f]{Minimum normalized allocation}
    }
    Set $A_i[t] = A_i[t-1] + \One{i = i^*} \quad \forall \, i\in [n]$%
    \tcp*{Update agents' total allocations}
}
\end{algorithm}

%% file: body/4.ideal.tex
\subsection{The Ideal Utility Benchmark}
\label{ssec:ideal}

When mechanisms use real money and agents have quasi-linear utilities, monetary values provide a way to compare agents' values and utilities.
In contrast, in non-monetary settings, there is no way to make interpersonal comparisons between agents, and so we need to use individual welfare benchmarks (i.e., from each agent's perspective).
The benchmark we use here is the \emph{ideal utility}, which starts by assuming that agent $i$ has an (exogenous) fair share $\alpha_i$ and then defines her benchmark utility as the top $\alpha_i$ fraction of items based on her own relative valuations in different rounds.
\cite{DBLP:conf/sigecom/GorokhBI21}~introduced the ideal utility based on earlier ideas from the economics literature on bargaining and the Fisher market.
\cite{DBLP:conf/sigecom/BanerjeeFT23}~subsequently extend this definition for multi-round demands, which we use in \cref{sec:mult}.

Following~\cite{DBLP:conf/sigecom/GorokhBI21,DBLP:conf/sigecom/BanerjeeFT23}, we define an agent $i$'s ideal utility to be \emph{the highest expected per-round utility she can get while only getting her fair share of the resource} (i.e., being allocated at most an $\alpha_i$ fraction of rounds). 
Formally, for each agent $i$ with fair-share $\alpha_i$, her ideal utility is the maximum expected per-round utility in a simplified setting with no agents other than $i$, but where agent $i$ has to request the item at most $\a_i$ fraction of the rounds.

Note, though, that the definition of an agent's ideal utility only takes as input her fair share $\alpha_i\in [0,1]$, and not those of other agents.
This motivates a generalization of the above definition to that of \textit{$\b$-ideal utility}: the maximum expected per-round utility an agent can receive with no competition if she is allowed to get the item a $\b \in [0,1]$ fraction of the rounds.
We denote the $\b$-ideal utility of agent $i$ with $\vstar_i(\b)$.
Now, extending the definition of~\cite{DBLP:conf/sigecom/BanerjeeFT23}, we have the following.

\begin{definition}[$\b$-ideal utility]\label{def:b_ideal}
    Fix $\b\in [0,1]$ and agent $i$ with value distribution $\mathcal{F}_i$ in her steady state.
    Her \textit{$\b$-ideal utility} $\vstar_i(\b)$ is the maximum expected utility she can get under no competition, but when the probability with which she requests the item is at most $\b$.
    Formally, if $\rho_\b:\R_+ \to [0, 1]$ denotes the (stationary) probability of requesting the item given its current value, then:
    \begin{equation}
    \label{eq:ind:def}
        \vstar_i(\b) \;=
        \quad
        \max_{\rho_\b:\R_+ \to [0, 1]} \quad \Ex[V \sim \mathcal{F}_i]{V \cdot \rho_\b(V)}
        \qquad \textrm{such that} \quad
        \Ex[V \sim \mathcal{F}_i]{\rho_\b(V)} \le \b
    \end{equation}
\end{definition}

The $\b$-ideal utility (and the corresponding optimal policy $\rho^{\star}_\b(\cdot)$) are important for our results, and so, before proceeding, we make three observations:
First, the ideal utility of agent $i$ with fair share $\a_i$ (as defined in~\cite{DBLP:conf/sigecom/GorokhBI21,DBLP:conf/sigecom/BanerjeeFT23}) is $\vstar_i(\alpha_i)$.
However, the more general definition lets us consider the value of a different request frequency, i.e., more/less aggressive than what her fair share is, even though the mechanism treats them as having share $\alpha_i$.

Second, note that optimization problem~\eqref{eq:ind:def} references $\mathcal{F}_i$, the agent's \emph{steady-state} value distribution; neither the policy $\rho_\b$ nor the expectations depend on her true latent state.
This is because, for sufficiently long horizon $T$, the value of $\vstar_i(\b)$ is roughly the same as the following:
the time-average utility of agent $i$ over $T$ rounds if she requests the item according to the optimal policy (that might be state, history, and time-dependent), which always gets the item at most $\b T$ rounds.
In particular, if we take $T \to \infty$, then the two are the same almost surely.

Third, note that $\vstar_i(\b)$ is a non-decreasing and concave function.
This is important to relate its values for different values of $\b$.
To see this, note that we can simplify~\eqref{eq:ind:def} to get that $\vstar_i(\b)$ is the expected value after replacing values in the bottom $(1-\b)$ part of the distribution with $0$.
For distributions with no atoms this is $\vstar(\b) = \Ex{V \One{V \ge F^{-1}(1-\b)}}$, where $F$ is the distribution's CDF.
This is clearly non-decreasing in $\b$ as higher $\b$ includes more non-negative values.
Further, the derivative in $\b$ is $F^{-1}(1-\b)$, the value at the top $\b$th fraction of the distribution, which is non-increasing in $\b$ and shows that the function is concave.

\subsection{Ideal Utility and Social Welfare}
\label{sec:ideal:welfare}

In the special case where all agents have the same value distribution and equal fair shares, making their values comparable, we can reason about their relative values even without knowing their distribution.
Consequently, maximizing social welfare becomes a reasonable goal. 
In this setting, where every agent $i$ has $\vstar_i(\cdot) = \vstar(\cdot)$, the quantity $n \vstar(1/n) T$ is an upper bound for the expected maximum social welfare.

This upper bound entails that every per-agent robustness guarantee for ideal utility also implies a social welfare guarantee.
When every agent receives utility at least $c \vstar(1/n) T$ for some $c \le 1$, this immediately implies that the resulting social welfare is at least a $c$ fraction of the expected maximum social welfare.
In other words, the robustness guarantees of the \dmmf mechanism immediately translate to Price of Anarchy bounds in symmetric settings\footnote{We note that the existence of Nash Equilibria for discrete values is guaranteed from standard results, e.g., \cite[Chapter 13]{DBLP:books/daglib/0070442}.}.
In particular, in \cref{cor:single:welfare}, we show that the PoA under \dmmf goes to $1$ as $n\rightarrow\infty$ under mild assumptions.

We note that the above technique can also be used to maximize social welfare for heterogeneous value distributions assuming values are expressed so that inter-agent comparisons are possible, the principal knows these distributions, and can use them to determine shares.
This follows from standard arguments involving ex-ante relaxations, wherein we set the fair shares according to the fractional allocation that maximizes the expected social welfare.

%% file: body/5.single.tex
\section{Performance Guarantees for \dmmf Under Single-Round Demands}
\label{sec:single}

In this section, we characterize the robustness of the \dmmf mechanism, i.e., how much utility an agent can get under \dmmf using a particular class of strategies, irrespective of the behavior of the other agents.
As we mentioned in \cref{sec:intro}, our focus on robust guarantees is because agents in practice are often imperfect optimizers, so we want strategies that make utility guarantees that hold under general conditions, e.g., not only under an equilibrium.
Because we focus on a single agent, we drop the $i$ subscript through this section.
We consider an agent with fair share $\a$, whose latent state Markov chain has decorrelation parameter $\g$, and who uses one among the following set of strategies:
\begin{definition}[$\b$-aggressive strategies]
\label{def:b-aggresive}
    For any fixed $\b\in[0,1]$, in a \textit{$\b$-aggressive strategy} the agent requests the resource in round $t$ with probability $\rstar_\b(V[t])$, i.e., with the probability suggested by the optimal policy in~\eqref{eq:ind:def} which realizes her $\b$-ideal utility.
    In other words, $\b$ parameterizes how aggressively the agent asks for resources.
\end{definition}

Note that $\b$ \emph{need not equal the agent's fair share} $\a$.
Using $U[t]$ to denote the agent's utility in round $t$, our main guarantee for \dmmf is as follows.

\begin{theorem}\label{thm:single:guar}
    Consider any agent with fair share $\alpha$ whose latent-state Markov chain satisfies $p(s',s) \ge \g \pi(s)$, for all $s,s'\in\mathcal S$.
    Fix any $\b \in (0,1]$ and let $\rstar_\b$ be the optimal policy of problem~\eqref{eq:ind:def}.
    Under the \dmmf mechanism, using the $\b$-aggressive strategy, then irrespective of how other agents behave, the agent is guaranteed total utility for every $\tau \in [T]$
    \begin{equation} \label{eq:generalguarantee}
        \sum_{t\in[\tau]} \Ex{\big.U[t]}
        \ge
        \g
        \frac{\a - (1-\a)\b(1-\g)}{\a + (1-\a)\b\g}
        \vstar(\b) \tau
        -
        \frac{\vstar(\b)}{\a + \b}
        .
    \end{equation}
\end{theorem}

\cref{eq:generalguarantee} provides a robust guarantee for the utility earned by an agent under \dmmf for any share $\a$, any $\b$-aggressive strategy, and any decorrelation parameter $\g$. 
This level of generality makes it hard to interpret directly.
Also, its proof is somewhat involved; we defer it to~\cref{ssec:proofofmainthm}. 

One reason the guarantee in~\cref{eq:generalguarantee} is difficult to parse is that it is in terms of $\vstar(\b)$, instead of the ideal utility $\vstar(\a)$ (which is the benchmark we want to compare against).
While this is not directly a useful comparison (since the agent cannot hope utility close to this, especially for $\b \gg \a$), $\b$ gives us a tuning parameter that achieves a range of different guarantees via a \emph{single} set of strategies. 
To demonstrate this, in the following two sections, we carefully choose $\b$ depending on setting-specific features to get strong performance guarantees under $(i)$ i.i.d. values ($\g = 1$) with structured distributions and $(ii)$ general correlated value distributions.

\subsection{Robustness and Efficiency of \dmmf under Independent Valuations}
\label{ssec:single:indep}

First, we specialize the guarantees in~\cref{thm:single:guar} to the case where agent valuations are independent across rounds (i.e., $\g = 1$).
The corollaries mostly follow by carefully setting the parameters in \cref{thm:single:guar}, the detailed proofs of which are deferred to \cref{ssec:app:single:indep}.

With independent distributions, by setting $\b = \a$, we immediately get a \emph{worst-case} robustness guarantee, i.e., one that does not depend on the agent's value distribution.

\begin{corollary}[Robustness under Worst-case Distributions]
\label{cor:indep:worst}
    If an agent with fair share $\a$ and i.i.d. values (i.e., $\g = 1$) uses the $\a$-aggressive strategy, then under \dmmf in every round $\tau \in [T]$
    \begin{equation}\label{eq:indep_gen}
        \sum_{t\in[T]} \Ex{\big.U[t]}
        \ge
        \frac{1}{2 -\a} \vstar(\a) \tau
        -
        \order{ 1 }
        .
    \end{equation}
\end{corollary}

This corollary follows by substituting $\g = 1$ and $\b = \a$ in~\cref{thm:single:guar}.
It recovers the guarantee of~\cite{DBLP:conf/sigecom/GorokhBI21,DBLP:conf/sigecom/BanerjeeFT23} that was obtained under the `repeated first-price with artificial currencies' mechanism, but using the much simpler \dmmf mechanism; to the best of our knowledge, this is the first robustness guarantee shown for a non-pseudo-market mechanism.
Additionally, compared to those previous results, our guarantee holds in every round (not just the final one $T$), the additive utility error is $\order*{1}$ instead of $\order*{\sqrt{T}}$, and the multiplicative factor is slightly better than $1/2$.
The first two facts have the following practical implications, which extend to all our results.

\begin{remark}[Handling dynamic populations]\label{rem:single:dynamic_pop}
    When agents enter or leave the mechanism, the principal might want to restart the allocation.
    If the \dmmf mechanism is restarted $m$ times, the only additional error in our guarantee is an $\order*{m}$ additive error.
    Furthermore, these restarts can happen at arbitrary rounds. 
    In contrast, restarting the pseudo-market mechanism of~\cite{DBLP:conf/sigecom/GorokhBI21,DBLP:conf/sigecom/BanerjeeFT23} needs to happen at pre-decided times to avoid unnecessary losses.
\end{remark}

The guarantee in~\cref{eq:indep_gen} is remarkable in that it shows that an agent, by requesting appropriately, can always realize at least $\frac{1}{2}$ her ideal utility, no matter her value distribution or how other agents behave.
One downside is that this may lead to the system overall losing efficiency because no agent requests the item in too many rounds.
Unfortunately, this downside is unavoidable in the worst case: in \cref{thm:imp:ind}, we show that no mechanism can guarantee every agent more than $1-1/e$ of her ideal utility for all value distributions.

\paragraph{Improved guarantees}
Using $\b$-aggressive strategies allows us to recover this efficiency loss in cases where an agent's distribution is structured.
First, we consider that the agent's value is drawn i.i.d. from a uniformly continuous probability density function (pdf), which is upper and lower bounded.
For example, if the agent's value is a uniform distribution every round, the agent can always guarantee a $1 - \sqrt{2 \a}$ fraction of her total ideal utility.

\begin{corollary}[Independent Values from Bounded pdf]
\label{cor:indep:density}
    Consider an agent with fair share $\a$ whose value is drawn i.i.d. (i.e., with $\g = 1$) from some interval $[\underline v, \overline v]$ via a probability density function $f$ satisfying $\l_1 \le f(v) \le \l_2$ for all $v \in [\underline v, \overline v]$.
    Let $\l = \l_2/\l_1$.
    If $\a \le \min\big\{ 2/\l,\l/2 \big\}$, using a $\b$-aggressive request strategy with $\b = \sqrt{\nicefrac{2 \a}{\l}}$, the agent can get 
    \begin{equation*}
        \sum_{t\in[\tau]} \Ex{\big.U[t]}
        \ge
        \qty( 1 - \sqrt{2 \l \a} ) \, \vstar(\a) \tau
        -
        \order{1}
        .
    \end{equation*}
\end{corollary}

Another case that showcases the strength of \cref{thm:single:guar} and $\b$-aggressive policies is when the agent's values are Bernoulli random variables.
If the value's mean is exactly $\a$, then we cannot improve over the guarantee of \cref{cor:indep:worst} (in \cref{thm:imp:dep} we show that such an agent cannot do better than $1/2$ of the ideal utility if $\a$ is small).
However, if the mean is not close to $\a$, we prove that a $\b$-aggressive policy guarantees almost all her ideal utility.

\begin{corollary}[Independent Bernoulli Values]
\label{cor:indep:bernoulli}
    Consider an agent with fair share $\a$ and i.i.d. values (i.e., $\g = 1$) from a Bernoulli distribution with mean $p$, for some $p \in (0, 1]$.
    Using a $p$-aggressive strategy, guarantees for every $\tau \in [T]$
    \begin{equation*}
        \sum_{t\in[\tau]} \Ex{\big.U[t]}
        \ge
        \frac{\max\{\a, p\}}{\a + p - \a p} \vstar(\a) \tau
        -
        \order{1}
        \ge
        \qty( 1 - \frac{1}{1 + \l} ) \vstar(\a) \tau
        -
        \order{1}
    \end{equation*}
    where $\l = \max\big\{\frac{\a}{p},\frac{p}{\a}\big\} \ge 1$.
    The bigger $\l$ becomes, the closer to optimal the previous bound becomes.
\end{corollary}

As we mentioned in \cref{sec:ideal:welfare}, in the special case when agents are \emph{symmetric} (identical distributions and fair shares), maximizing social welfare becomes a reasonable goal.
With $n$ symmetric agents, each with a fair share of $\a = 1/n$, the maximum social welfare is upper bounded by $n\vstar(1/n)T$.
The above corollaries show that each individual can achieve total utility close to $\vstar(\a) T$, implying close to optimal social welfare at all Nash equilibria.
In the settings of \cref{cor:indep:density,cor:indep:bernoulli}, when there is a large number of agents, we achieve close to full efficiency. 

\begin{corollary}[Social Welfare in Symmetric Settings]
\label{cor:single:welfare}
    When there are $n$ agents with identical fair shares and distributions, \cref{cor:indep:worst,cor:indep:density,cor:indep:bernoulli} prove that the Price of Anarchy is bounded.
    More specifically, ignoring the additive $\order*{1/\tau}$ errors, we have that the Price of Anarchy is at most
    \begin{itemize}
        \item $2 - 1/n$ for arbitrary independent valuations.
    
        \item $1 + \order*{ 1/\sqrt{n} }$ for valuations with bounded density functions.
        
        \item $1 + \order*{ n^{-\e} }$ for Bernoulli valuations with mean $p \le 1/n^{1+\e}$ or $p \ge 1/n^{1-\e}$ for some $\e>0$.
    \end{itemize}
\end{corollary}

\subsection{Robustness of \dmmf under Correlated Valuations}
\label{ssec:single:correl}

We next show that $\b$-aggressive strategies can also help in settings with correlated values.
Our first result considers the strategy of setting $\b = \a$ from~\cref{cor:indep:worst} (independent of $\g$ and the distribution) and shows that it achieves a non-trivial robustness guarantee.
We will offer a stronger guarantee in highly correlated settings, i.e., when the decorrelation parameter $\g$ is small.
While all the claims we make in this section follow from~\cref{thm:single:guar}, they require the concavity of $\vstar$ and somewhat involved algebraic manipulations; we defer the rigorous proofs to \cref{ssec:app:correlatedguarantees}.

\begin{corollary}[Robustness under Moderate Correlation]
\label{cor:single:big_gamma}
    Under \dmmf, an agent with fair share $\a$ and decorrelation parameter $\g$ following an $\a$-aggressive policy guarantees by every round $\tau \in [T]$
    \begin{equation*}
        \sum_{t\in[\tau]} \Ex{\big.U[t]}
        \ge
        \g
        \frac{\g + \a(1 - \g)}{1 + (1-\a)\g}
        \vstar(\a) \tau
        -
        \order{1}
        \ge
        \qty( \frac{\g^2}{1 + \g} + \frac{\a\g}{(1+\g)^2} ) \vstar(\a) \tau
        -
        \order{1}
        .
    \end{equation*}
\end{corollary}

When values are correlated and the agent uses a $\b$-aggressive strategy, the other agents can use her request to infer her latent state (e.g., when the agent has positive value only in one latent state).
Consequently, an adversary can use this to block the agent more effectively than the independent value setting.
Such adversarial behavior hurts the agent's utility in \textit{two} ways: she wins fewer high-valued rounds \emph{and} more low-valued rounds.
Since this effect is twofold and becomes stronger the smaller $\g$ is, the guarantee of \cref{cor:single:big_gamma} scales with $\g^2$ instead of $\g$ (for $\a \ll \g$).
If $\g$ becomes small enough, i.e., the Markov Chain's states are sufficiently correlated, it is better for the agent to use a $\b$-aggressive policy with $\b < \a$: while this possibly entails winning fewer rounds, it guarantees to win a higher fraction of requests, helping to control better which rounds are won.
More specifically, this improvement occurs if
\begin{equation}\label{eq:single:small_gamma_condition}
    \sqrt{1-\g} + 1 - \g
    >
    \frac{1}{1-\a}
    \iff
    \g
    <
    \frac{1}{2}\left(
        \sqrt{\frac{5-\a}{1-\a}}
        -
        \frac{1 + \a}{1 - \a}
    \right)
    \implies
    \g
    <
    \frac{\sqrt 5 - 1}{2}
    .
\end{equation} 

\begin{corollary}[Robustness under High Correlation]\label{cor:single:small_gamma}
    An agent with fair share $\a$ and decorrelation parameter $\g$ satisfying \cref{eq:single:small_gamma_condition}, can use a $\b$-aggressive policy with $\b = \frac{\a}{(1-\a)(\sqrt{1-\g} + 1 - \g)} < \a$ to guarantee for every round $\tau \in [T]$
    \begin{equation*}
        \sum_{t\in[\tau]} \Ex{\big.U[t]}
        \ge
        \frac{1 - \sqrt{1-\g}}{(1 - \a)(1 + \sqrt{1-\g})} \vstar(\a) \tau
        -
        \order{1}
        \ge
        \frac{\g}{4(1 - \a)} \vstar(\a) \tau
        -
        \order{1}
        .
   \end{equation*}
\end{corollary}

We compare all our bounds graphically in \cref{fig:bounds_comp}; moreover, we can get a universal bound for all $\g$ with a simpler $\b$, albeit weaker than the previous two combined.

\begin{figure}[t]
    \centering
    \resizebox{.4\textwidth}{!}{\input{figures/corollaries.pgf}}
    \caption{Comparison of
    the $\Omega(\g)$ lower bound implied by combining \cref{cor:single:big_gamma,cor:single:small_gamma},
    the $\Omega(\g^2)$ lower bound of just \cref{cor:single:big_gamma},
    the $\g\frac{1 + \g}{4 + 2\g}$ lower bound of \cref{cor:single:general_gamma},
    and the $\nicefrac{\g}{\g+1}$ upper bound of \cref{cor:imp:markov},
    as $\a \to 0$.}
    \label{fig:bounds_comp}
\end{figure}

\begin{corollary}[Robustness under Arbitrary Correlation]\label{cor:single:general_gamma}
    The $\frac{\a}{2}$-aggressive strategy guarantees for every round $\tau \in [T]$
    \begin{equation*}
    \sum_{t\in[\tau]} \Ex{\big.U[t]}
        \ge
        \g\frac{1 + \g}{4 + 2\g} \vstar(\a) \tau
        -
        \order{ 1 }
        .
    \end{equation*}
\end{corollary}

In \cref{cor:imp:markov}, we show that for the \dmmf mechanism, the above result is tight up to constant factors: an agent cannot guarantee more than $\frac{\g}{1 + \g}$ of her ideal utility when $\a$ is small.

Finally, note that using $\b$-aggressive strategies allows the other agents to infer information about the agent's latent state from her requests.
This means that robustness guarantees of such strategies depend on $\g$.
An alternate approach to deal with correlated valuations is to use request strategies that are \emph{state-independent}, whereby the probability of requesting is identical in every round.
We use this idea to get a guarantee that depends on the minimum state-visitation probability $\min_{s\in\calS}\pi(s)$.

\begin{corollary}
\label{cor:single:pi_min}
    Using a state-independent strategy, the agent can guarantee for every round $\tau \in [T]$
    \begin{equation*}
    \hspace{-5.5cm}
        \sum_{t\in[\tau]} \Ex{\big.U[t]}
        \ge
        \frac{\min_s\pi(s)}{\min_s\pi(s) + 1 - \a} \vstar(\a) \tau
        -
        \order{1}
        .
    \end{equation*}
\end{corollary}

We discuss this result in~\cref{ssec:state_independent}, along with a more general result for state-independent strategies that offer guarantees irrespective of $\g$.
In \cref{thm:imp:dep}, we show that the above corollary is tight when $\g \ll \a$: there exist settings where $\min_s \pi(s) = \a$ and the agent can guarantee at most $\a + \order*{ \a^2 }$ of her ideal utility in the \dmmf mechanism, while \cref{cor:single:pi_min} guarantees $\a$.

\subsection{Robustness in the \dmmf mechanism (Proof of Theorem \ref{thm:single:guar})}
\label{ssec:proofofmainthm}

In this section, we prove~\cref{thm:single:guar}, which implies all the utility guarantees we have presented.
We refer to the other agents simply as the adversary.
To prove our guarantee, we first prove an invariant under \dmmf for the normalized allocation an agent gets under any strategy.
For this, we need some additional notation: under any request strategy being used by the agent, we define
\begin{itemize}
    \item $\req[t] = 1$ indicates that the agent requests the item in round $t$; $\req[t] = 0$ otherwise.
    \item $\blk[t] = 1$ indicates that the adversary wins the item in round $t$ and would have done so regardless of whether the agent requests the item or not; $\blk[t] = 0$ otherwise.
    \item $W[t] = 1$ if the agent wins round $t$; $W[t] = 0$ otherwise; note $W[t] = \req[t] (1 - \blk[t])$.
    \item $U[t]$ is the utility of the agent in round $t$; note $U[t] = V[t] W[t] = V[t] \req[t] (1 - \blk[t])$.
\end{itemize}

We start with the following simple lemma that bounds the number of times the adversary can block the agent. We do so for any strategy followed by the agent and the adversary, not only the $\b$-aggressive strategy that we will consider for the agent in subsequent results.

\begin{lemma}\label{lem:single:dmmf}
    For every round $\tau$ and agent's strategy,
        $\frac{1}{1-\a} \sum_{t\in[\tau]} \Blk[t]
        \le
        \frac{1}{\a} \left( 1 + \sum_{t\in[\tau]} W[t] \right)
        $.
\end{lemma}

The lemma above makes intuitive sense: the number of rounds where the agent gets blocked (normalized by the adversary's fair share) cannot be more than the number of rounds the agent wins (normalized by her fair share).
If the inequality in \cref{lem:single:dmmf} is violated in a round, then the agent in that round has won significantly less than the adversary.
This means that she can win the item if she requests it.
This indicates why it is crucial to bound rounds when $\blk[t] = 1$, not just the rounds the adversary wins.
In contrast, the number of rounds the adversary wins might not correlate with the rounds won by the agent.
We defer the full proof of the lemma to \cref{ssec:app:lemma_proof}.

While this lemma offers a lower bound for the number of rounds won, it does not offer any bounds for the resulting utility.
To get meaningful utility bounds, we need to relate $\Ex{U[t]}$ and $\Ex{W[t]}$.
This is easy if the agent's values are independent across time: for any $\b$-aggressive strategy, using the independence of the agent's behavior across rounds, we can show that $\Ex{U[t]} = \vstar(\b)\Ex{1 - \blk[t]}$ and $\Ex{W[t]} = \b\Ex{1 - \blk[t]}$.
The last two equalities and \cref{lem:single:dmmf} imply \cref{thm:single:guar} for $\g = 1$.

For the full result of \cref{thm:single:guar} (extending to dependent values), we need more complex arguments. 
For example, consider that the Markov chain's transitions are deterministic and the agent has positive values in only one state.
If the agent requests only when her value is positive, the adversary can fully predict her future requests.
This information leakage is controlled by the decorrelation parameter $\g$.

\begin{figure}[t]
    \centering
    \input{figures/dependence_graph}
    \caption{Dependence graph of random variables in a round $t$. We note that $\blk[t]$ also depends on the agents' allocations up to round $t-1$, which is not depicted in this figure.}
    \label{fig:dependence}
\end{figure}

\begin{proof}[Proof of \cref{thm:single:guar}]
    Fix a $\b \in (0, 1]$.
    Under the $\b$-aggressive strategy, in every round $t$, $\req[t]$ is a Bernoulli random variable with mean $\rstar_\b(V[t])$.
    We first decompose the expected utility that the agent wins in some round $t$:
    \begin{alignat*}{3}
        \Line{
            \Ex{\big.U[t]}
        }{=}{
            \Ex{(1 - \Blk[t]) V[t] \req[t] \big.}
        }{}
        \\
        \Line{}{=}{
            \sum_{s',s\in \calS}
            \ExC{(1 - \Blk[t]) V[t] \req[t]}{
                S[t-1] = s', S[t] = s
            }
            \pi(s')p(s',s) 
        }{}
        \\
        \Line{}{=}{
            \sum_{s',s\in \calS}
            \ExC{1 - \Blk[t]}{
                S[t-1] = s'
            }
            \ExC{V[t] \req[t]}{
                S[t] = s
            }
            \pi(s')p(s',s)
        }{}
    \end{alignat*}
    where in the last inequality we used the fact that $\blk[t]$ and $V[t],\req[t]$ are independent conditioned on the states $S[t-1], S[t]$.
    Specifically, $V[t], \req[t]$ depend directly only on $S[t]$, while $\blk[t]$ depends on the history of the previous rounds (allocations and the variables $\req[1],\ldots,\req[t-1]$) which is independent of $S[t]$ conditioned on $S[t-1]$.
    We showcase this dependence in \cref{fig:dependence}.
    In the above equality, we now use that $p(s',s) \ge \g\pi(s)$
    \begin{alignat}{3}\label{eq:single:utility_bound}
        \Line{
            \Ex{\big.U[t]}
        }{\ge}{
            \g
            \left(
                \sum_{s'\in \calS}
                \ExC{1 - \Blk[t]}{
                    S[t-1] = s'
                }
                \pi(s')
            \right)
            \left(
                \sum_{s\in \calS}
                \ExC{V[t] \req[t]}{
                    S[t] = s
                }
                \pi(s)
            \right)
        }{}
        \nonumber\\
        \Line{}{=}{
            \g \Ex{1 - \blk[t]} \vstar(\b)
        }{}
    \end{alignat}
    where in the last equality we used that $\Ex{V[t] \req[t]} = \Ex{V[t] \rstar_\b(V[t])} = \vstar(\b)$.
    Similarly with the above reasoning, for the random variable $W[t]$, we have that
    \begin{alignat}{3}\label{eq:single:win_bound}
        \Line{
            \Ex{W[t]}
        }{=}{
            \Ex{(1 - \Blk[t]) \req[t]}
            =\;
            \Ex{\req[t]}
            -
            \Ex{\Blk[t] \req[t]}
        }{}
        \nonumber\\
        \Line{}{=}{
            \Ex{\req[t]}
            -
            \sum_{s',s\in \calS}
            \Ex{\Blk[t] 
                \;\Big|\; S[t-1] = s'
            }
            \Ex{\req[t]
                \;\Big|\; S[t] = s
            }
            \pi(s')p(s',s)
        }{}
        \nonumber\\
        \Line{}{\le}{
            \Ex{\req[t]}
            -
            \g
                \sum_{s'\in \calS}
                \Ex{\Blk[t] 
                    \;\Big|\; S[t-1] = s'
                }
                \pi(s')
                \sum_{s\in \calS}
                \Ex{\req[t]
                    \;\Big|\; S[t] = s
                }
                \pi(s)
        }{}
        \nonumber\\
        \Line{}{=}{
            \Ex{\req[t]}
            -
            \g
            \Ex{\Blk[t]}
            \Ex{\req[t]}
            \le
            \b
            -
            \g
            \Ex{\Blk[t]}
            \b
        }{}
    \end{alignat}
    where in the last inequality we used the fact that $\Ex{\req[t]} \le \b$.
    Summing \cref{eq:single:win_bound} over $t\in [\tau]$ for any $\tau$ and combining it with the expectation of \cref{lem:mult:dmmf}, we get
    \begin{equation}\label{eq:single:4}
        \frac{\a + (1-\a) \b \g}{1-\a}
        \sum_{t\in[\tau]} \Ex{\Blk[t]}
        \le
        \b \tau + 1
        .
    \end{equation}

    Summing \cref{eq:single:utility_bound} over $t\in [\tau]$ for any $\tau$ and bounding the $\sum_t \Ex{\Blk[t]}$ term with \cref{eq:single:4} makes
    \begin{alignat*}{3}
        \Line{
            \sum_{t\in[\tau]} \Ex{\big.U[t]}
        }{\ge}{
            \vstar(\b) \g \left(
                \tau - \frac{1-\a}{\a + (1 - \a)\b\g} \b \tau - \frac{1-\a}{\a + (1 - \a)\b\g}
            \right)
        }{}
    \end{alignat*}
    The theorem follows by re-arranging and bounding the constant term using $1 - \a \le 1$ and $\g \le 1$.
\end{proof}

\subsection{State independent policies}
\label{ssec:state_independent}

In this section, we prove a utility guarantee that does not depend on $\gamma$, hence may offer better guarantees when $\g$ is small.
Similar to \cref{thm:single:guar}, our bound is a fraction of the $\b$-independent utility, for any $\b$ chosen by the agent.
However, instead of $\g$, the guarantee is parametrized by $\sigma(\b)$, which is a property of the optimal policy $\rstar_\b$.
$\sigma(\b)$ is the ratio between the average and maximum request probabilities of each state in the Markov chain: $\E_{V\sim\calF_s}[\rstar_\b(V)]$ for $s\in\calS$.
If these probabilities are the same for every state then $\sigma(\b) = 1$ and the bound of our theorem becomes identical to the one in \cref{thm:single:guar} for $\g = 1$.
The result for this special case makes sense: if the policy $\rstar_\b$ requests the item with the same probability every round, even if the agent's values are dependent across rounds, her requests are independent.
This means that the adversary cannot infer her state from the request, and so cannot take advantage of the dependence.
We therefore get results that match the independent values case.
In general, however, $\sigma(\b)$ can take values less than $1$ which makes our guarantee smaller.
We note that a simple lower bound is $\sigma(\b) \ge \min_{s\in\calS}\pi(s)$, which we use to get \cref{cor:single:pi_min}.

\begin{theorem}\label{thm:single:state_independent}
    Fix $\b \in (0,1]$.
    There is a strategy that guarantees total utility for every round $\tau$
    \begin{equation*}
        \sum_{t\in[\tau]} \Ex{\big.U[t]}
        \ge
        \frac{\a}{\a + \frac{\b}{\sigma(\b)} - \a \frac{\b}{\sigma(\b)}}
        \vstar(\b) \tau
        -
        \frac{\vstar(\b)}{\a + \b}
    \end{equation*}
    where
    $
        \sigma(\b)
        =
        \frac{ \Ex[V\sim \calF]{\rstar_\b(V)} }{ \max_{s\in\calS}\Ex[V\sim \calF_s]{\rstar_\b(V)} }
    $.
\end{theorem}

The theorem's proof is based on the agent requesting every round with probability $\max_{s\in\calS}\E_{V\sim \calF_s}[\rstar_\b(V)]$, regardless of her state.
This implies that the agent's requests are independent across rounds, preventing the adversary from blocking the agent effectively.
In addition, the agent's utility under this policy with no competition is at least $\vstar(\b)$ since her probability of requesting the resource in every round is not smaller.
However, this strategy makes the agent request with a higher overall probability, $\b / \sigma(\b)$, which is why we get such a factor in the result.

\begin{proof}
    Fix a $\b \in (0,1]$. 
    We will study the agent's utility under a strategy that requests the item with the same probability in every round, for any state of that round.
    We set that probability, $p$, to be equal to the maximum probability of requesting the item in any state by policy $\rstar_\b$, i.e., $p = \max_{s\in\calS} \E_{V\sim\calF_s}[\rstar_\b(V)]$.
    Most importantly, we assume that the agent requests the item in round $t$ if $V[t]$ is in the top $p$-th quartile of the distribution $\calF_{S[t]}$.
    This means that $\req[t]$ is a Bernoulli random variable with mean $p$ even if we condition on $S[t] = s$, for any state $s\in\calS$.
    It is easy to notice that $p = \Ex{\req[t]} = \E_{V\sim \calF}[\rstar_\b(V)] / \sigma(\b)  \le \b / \sigma(\b)$ and that $\Ex{V[t]\req[t]} \ge \vstar(\b)$; the last inequality follows because of the definition of $p$ which implies that the probability of requesting the item with value $V[t]$ is greater than $\rstar_\b(V[t])$.

    Because the agent's requests are independent of her state, the adversary cannot extract any information about her requests from past rounds.
    Therefore, $\Blk[t]$ is independent of $V[t], S[t], \req[t]$.
    Using the independence, we have that
    \begin{align}\label{eq:single:71}
        \Ex{\big.U[t]}
        =
        \Ex{(1 - \Blk[t]) V[t] \req[t]}
        \ge
        \vstar(\b) \Ex{1 - \Blk[t]}
        .
    \end{align}
    Similarly, we have that
    \begin{align*}
        \Ex{W[t]}
        =
        \Ex{(1 - \Blk[t]) \req[t]}
        =
        \Ex{1 - \Blk[t]} \Ex{\req[t]}
        \le
        \frac{\b}{\sigma(\b)} \Ex{1 - \Blk[t]}
        .
    \end{align*}

    Combining the above inequality with the expectation of \cref{lem:single:dmmf} for any round $\tau$ we get
    \begin{align*}
        \frac{\a}{1 - \a} \sum_{t\in[\tau]} \Ex{\blk[t]} \le 1 + \frac{\b}{\sigma(\b)} \Ex{\tau - \sum_{t\in[\tau]}\Blk[t]}
        .
    \end{align*}

    Adding \cref{eq:single:71} for $t \in [\tau]$ and combining with the above inequality we get
    \begin{align*}
        \sum_{t\in[\tau]} \Ex{\big.U[t]}
        \ge
        \vstar(\b) \left( 
            \tau
            -
            \frac{\b - \a \b}{\a \sigma(\b) + \b - \a \b} \tau
            -
            \frac{\sigma(\b)(1 - \a)}{\a \sigma(\b) + \b - \a \b}
        \right)
        .
    \end{align*}
    The theorem follows by bounding: $\frac{\sigma(\b)(1 - \a)}{\a \sigma(\b) + \b - \a \b} \le \frac{1 - \a}{\a  + \b - \a \b} \le \frac{1}{\a  + \b}$ using $\sigma(\b) \le 1$ and then $\a \ge 0$.
\end{proof}

%% file: figures/corollaries.pgf
\begingroup%
\makeatletter%
\begin{pgfpicture}%
\pgfpathrectangle{\pgfpointorigin}{\pgfqpoint{5.930399in}{4.481729in}}%
\pgfusepath{use as bounding box, clip}%
\begin{pgfscope}%
\pgfsetbuttcap%
\pgfsetmiterjoin%
\pgfsetlinewidth{0.000000pt}%
\definecolor{currentstroke}{rgb}{0.000000,0.000000,0.000000}%
\pgfsetstrokecolor{currentstroke}%
\pgfsetstrokeopacity{0.000000}%
\pgfsetdash{}{0pt}%
\pgfpathmoveto{\pgfqpoint{0.000000in}{-0.000000in}}%
\pgfpathlineto{\pgfqpoint{5.930399in}{-0.000000in}}%
\pgfpathlineto{\pgfqpoint{5.930399in}{4.481729in}}%
\pgfpathlineto{\pgfqpoint{0.000000in}{4.481729in}}%
\pgfpathlineto{\pgfqpoint{0.000000in}{-0.000000in}}%
\pgfpathclose%
\pgfusepath{}%
\end{pgfscope}%
\begin{pgfscope}%
\pgfsetbuttcap%
\pgfsetmiterjoin%
\pgfsetlinewidth{0.000000pt}%
\definecolor{currentstroke}{rgb}{0.000000,0.000000,0.000000}%
\pgfsetstrokecolor{currentstroke}%
\pgfsetstrokeopacity{0.000000}%
\pgfsetdash{}{0pt}%
\pgfpathmoveto{\pgfqpoint{0.870399in}{0.685729in}}%
\pgfpathlineto{\pgfqpoint{5.830399in}{0.685729in}}%
\pgfpathlineto{\pgfqpoint{5.830399in}{4.381729in}}%
\pgfpathlineto{\pgfqpoint{0.870399in}{4.381729in}}%
\pgfpathlineto{\pgfqpoint{0.870399in}{0.685729in}}%
\pgfpathclose%
\pgfusepath{}%
\end{pgfscope}%
\begin{pgfscope}%
\pgfpathrectangle{\pgfqpoint{0.870399in}{0.685729in}}{\pgfqpoint{4.960000in}{3.696000in}}%
\pgfusepath{clip}%
\pgfsetrectcap%
\pgfsetroundjoin%
\pgfsetlinewidth{0.803000pt}%
\definecolor{currentstroke}{rgb}{0.690196,0.690196,0.690196}%
\pgfsetstrokecolor{currentstroke}%
\pgfsetstrokeopacity{0.500000}%
\pgfsetdash{}{0pt}%
\pgfpathmoveto{\pgfqpoint{1.095854in}{0.685729in}}%
\pgfpathlineto{\pgfqpoint{1.095854in}{4.381729in}}%
\pgfusepath{stroke}%
\end{pgfscope}%
\begin{pgfscope}%
\pgfsetbuttcap%
\pgfsetroundjoin%
\definecolor{currentfill}{rgb}{0.000000,0.000000,0.000000}%
\pgfsetfillcolor{currentfill}%
\pgfsetlinewidth{0.803000pt}%
\definecolor{currentstroke}{rgb}{0.000000,0.000000,0.000000}%
\pgfsetstrokecolor{currentstroke}%
\pgfsetdash{}{0pt}%
\pgfsys@defobject{currentmarker}{\pgfqpoint{0.000000in}{-0.048611in}}{\pgfqpoint{0.000000in}{0.000000in}}{%
\pgfpathmoveto{\pgfqpoint{0.000000in}{0.000000in}}%
\pgfpathlineto{\pgfqpoint{0.000000in}{-0.048611in}}%
\pgfusepath{stroke,fill}%
}%
\begin{pgfscope}%
\pgfsys@transformshift{1.095854in}{0.685729in}%
\pgfsys@useobject{currentmarker}{}%
\end{pgfscope}%
\end{pgfscope}%
\begin{pgfscope}%
\definecolor{textcolor}{rgb}{0.000000,0.000000,0.000000}%
\pgfsetstrokecolor{textcolor}%
\pgfsetfillcolor{textcolor}%
\pgftext[x=1.095854in,y=0.588507in,,top]{\color{textcolor}\ttfamily\fontsize{16.000000}{19.200000}\selectfont 0.0}%
\end{pgfscope}%
\begin{pgfscope}%
\pgfpathrectangle{\pgfqpoint{0.870399in}{0.685729in}}{\pgfqpoint{4.960000in}{3.696000in}}%
\pgfusepath{clip}%
\pgfsetrectcap%
\pgfsetroundjoin%
\pgfsetlinewidth{0.803000pt}%
\definecolor{currentstroke}{rgb}{0.690196,0.690196,0.690196}%
\pgfsetstrokecolor{currentstroke}%
\pgfsetstrokeopacity{0.500000}%
\pgfsetdash{}{0pt}%
\pgfpathmoveto{\pgfqpoint{1.997672in}{0.685729in}}%
\pgfpathlineto{\pgfqpoint{1.997672in}{4.381729in}}%
\pgfusepath{stroke}%
\end{pgfscope}%
\begin{pgfscope}%
\pgfsetbuttcap%
\pgfsetroundjoin%
\definecolor{currentfill}{rgb}{0.000000,0.000000,0.000000}%
\pgfsetfillcolor{currentfill}%
\pgfsetlinewidth{0.803000pt}%
\definecolor{currentstroke}{rgb}{0.000000,0.000000,0.000000}%
\pgfsetstrokecolor{currentstroke}%
\pgfsetdash{}{0pt}%
\pgfsys@defobject{currentmarker}{\pgfqpoint{0.000000in}{-0.048611in}}{\pgfqpoint{0.000000in}{0.000000in}}{%
\pgfpathmoveto{\pgfqpoint{0.000000in}{0.000000in}}%
\pgfpathlineto{\pgfqpoint{0.000000in}{-0.048611in}}%
\pgfusepath{stroke,fill}%
}%
\begin{pgfscope}%
\pgfsys@transformshift{1.997672in}{0.685729in}%
\pgfsys@useobject{currentmarker}{}%
\end{pgfscope}%
\end{pgfscope}%
\begin{pgfscope}%
\definecolor{textcolor}{rgb}{0.000000,0.000000,0.000000}%
\pgfsetstrokecolor{textcolor}%
\pgfsetfillcolor{textcolor}%
\pgftext[x=1.997672in,y=0.588507in,,top]{\color{textcolor}\ttfamily\fontsize{16.000000}{19.200000}\selectfont 0.2}%
\end{pgfscope}%
\begin{pgfscope}%
\pgfpathrectangle{\pgfqpoint{0.870399in}{0.685729in}}{\pgfqpoint{4.960000in}{3.696000in}}%
\pgfusepath{clip}%
\pgfsetrectcap%
\pgfsetroundjoin%
\pgfsetlinewidth{0.803000pt}%
\definecolor{currentstroke}{rgb}{0.690196,0.690196,0.690196}%
\pgfsetstrokecolor{currentstroke}%
\pgfsetstrokeopacity{0.500000}%
\pgfsetdash{}{0pt}%
\pgfpathmoveto{\pgfqpoint{2.899490in}{0.685729in}}%
\pgfpathlineto{\pgfqpoint{2.899490in}{4.381729in}}%
\pgfusepath{stroke}%
\end{pgfscope}%
\begin{pgfscope}%
\pgfsetbuttcap%
\pgfsetroundjoin%
\definecolor{currentfill}{rgb}{0.000000,0.000000,0.000000}%
\pgfsetfillcolor{currentfill}%
\pgfsetlinewidth{0.803000pt}%
\definecolor{currentstroke}{rgb}{0.000000,0.000000,0.000000}%
\pgfsetstrokecolor{currentstroke}%
\pgfsetdash{}{0pt}%
\pgfsys@defobject{currentmarker}{\pgfqpoint{0.000000in}{-0.048611in}}{\pgfqpoint{0.000000in}{0.000000in}}{%
\pgfpathmoveto{\pgfqpoint{0.000000in}{0.000000in}}%
\pgfpathlineto{\pgfqpoint{0.000000in}{-0.048611in}}%
\pgfusepath{stroke,fill}%
}%
\begin{pgfscope}%
\pgfsys@transformshift{2.899490in}{0.685729in}%
\pgfsys@useobject{currentmarker}{}%
\end{pgfscope}%
\end{pgfscope}%
\begin{pgfscope}%
\definecolor{textcolor}{rgb}{0.000000,0.000000,0.000000}%
\pgfsetstrokecolor{textcolor}%
\pgfsetfillcolor{textcolor}%
\pgftext[x=2.899490in,y=0.588507in,,top]{\color{textcolor}\ttfamily\fontsize{16.000000}{19.200000}\selectfont 0.4}%
\end{pgfscope}%
\begin{pgfscope}%
\pgfpathrectangle{\pgfqpoint{0.870399in}{0.685729in}}{\pgfqpoint{4.960000in}{3.696000in}}%
\pgfusepath{clip}%
\pgfsetrectcap%
\pgfsetroundjoin%
\pgfsetlinewidth{0.803000pt}%
\definecolor{currentstroke}{rgb}{0.690196,0.690196,0.690196}%
\pgfsetstrokecolor{currentstroke}%
\pgfsetstrokeopacity{0.500000}%
\pgfsetdash{}{0pt}%
\pgfpathmoveto{\pgfqpoint{3.801308in}{0.685729in}}%
\pgfpathlineto{\pgfqpoint{3.801308in}{4.381729in}}%
\pgfusepath{stroke}%
\end{pgfscope}%
\begin{pgfscope}%
\pgfsetbuttcap%
\pgfsetroundjoin%
\definecolor{currentfill}{rgb}{0.000000,0.000000,0.000000}%
\pgfsetfillcolor{currentfill}%
\pgfsetlinewidth{0.803000pt}%
\definecolor{currentstroke}{rgb}{0.000000,0.000000,0.000000}%
\pgfsetstrokecolor{currentstroke}%
\pgfsetdash{}{0pt}%
\pgfsys@defobject{currentmarker}{\pgfqpoint{0.000000in}{-0.048611in}}{\pgfqpoint{0.000000in}{0.000000in}}{%
\pgfpathmoveto{\pgfqpoint{0.000000in}{0.000000in}}%
\pgfpathlineto{\pgfqpoint{0.000000in}{-0.048611in}}%
\pgfusepath{stroke,fill}%
}%
\begin{pgfscope}%
\pgfsys@transformshift{3.801308in}{0.685729in}%
\pgfsys@useobject{currentmarker}{}%
\end{pgfscope}%
\end{pgfscope}%
\begin{pgfscope}%
\definecolor{textcolor}{rgb}{0.000000,0.000000,0.000000}%
\pgfsetstrokecolor{textcolor}%
\pgfsetfillcolor{textcolor}%
\pgftext[x=3.801308in,y=0.588507in,,top]{\color{textcolor}\ttfamily\fontsize{16.000000}{19.200000}\selectfont 0.6}%
\end{pgfscope}%
\begin{pgfscope}%
\pgfpathrectangle{\pgfqpoint{0.870399in}{0.685729in}}{\pgfqpoint{4.960000in}{3.696000in}}%
\pgfusepath{clip}%
\pgfsetrectcap%
\pgfsetroundjoin%
\pgfsetlinewidth{0.803000pt}%
\definecolor{currentstroke}{rgb}{0.690196,0.690196,0.690196}%
\pgfsetstrokecolor{currentstroke}%
\pgfsetstrokeopacity{0.500000}%
\pgfsetdash{}{0pt}%
\pgfpathmoveto{\pgfqpoint{4.703126in}{0.685729in}}%
\pgfpathlineto{\pgfqpoint{4.703126in}{4.381729in}}%
\pgfusepath{stroke}%
\end{pgfscope}%
\begin{pgfscope}%
\pgfsetbuttcap%
\pgfsetroundjoin%
\definecolor{currentfill}{rgb}{0.000000,0.000000,0.000000}%
\pgfsetfillcolor{currentfill}%
\pgfsetlinewidth{0.803000pt}%
\definecolor{currentstroke}{rgb}{0.000000,0.000000,0.000000}%
\pgfsetstrokecolor{currentstroke}%
\pgfsetdash{}{0pt}%
\pgfsys@defobject{currentmarker}{\pgfqpoint{0.000000in}{-0.048611in}}{\pgfqpoint{0.000000in}{0.000000in}}{%
\pgfpathmoveto{\pgfqpoint{0.000000in}{0.000000in}}%
\pgfpathlineto{\pgfqpoint{0.000000in}{-0.048611in}}%
\pgfusepath{stroke,fill}%
}%
\begin{pgfscope}%
\pgfsys@transformshift{4.703126in}{0.685729in}%
\pgfsys@useobject{currentmarker}{}%
\end{pgfscope}%
\end{pgfscope}%
\begin{pgfscope}%
\definecolor{textcolor}{rgb}{0.000000,0.000000,0.000000}%
\pgfsetstrokecolor{textcolor}%
\pgfsetfillcolor{textcolor}%
\pgftext[x=4.703126in,y=0.588507in,,top]{\color{textcolor}\ttfamily\fontsize{16.000000}{19.200000}\selectfont 0.8}%
\end{pgfscope}%
\begin{pgfscope}%
\pgfpathrectangle{\pgfqpoint{0.870399in}{0.685729in}}{\pgfqpoint{4.960000in}{3.696000in}}%
\pgfusepath{clip}%
\pgfsetrectcap%
\pgfsetroundjoin%
\pgfsetlinewidth{0.803000pt}%
\definecolor{currentstroke}{rgb}{0.690196,0.690196,0.690196}%
\pgfsetstrokecolor{currentstroke}%
\pgfsetstrokeopacity{0.500000}%
\pgfsetdash{}{0pt}%
\pgfpathmoveto{\pgfqpoint{5.604945in}{0.685729in}}%
\pgfpathlineto{\pgfqpoint{5.604945in}{4.381729in}}%
\pgfusepath{stroke}%
\end{pgfscope}%
\begin{pgfscope}%
\pgfsetbuttcap%
\pgfsetroundjoin%
\definecolor{currentfill}{rgb}{0.000000,0.000000,0.000000}%
\pgfsetfillcolor{currentfill}%
\pgfsetlinewidth{0.803000pt}%
\definecolor{currentstroke}{rgb}{0.000000,0.000000,0.000000}%
\pgfsetstrokecolor{currentstroke}%
\pgfsetdash{}{0pt}%
\pgfsys@defobject{currentmarker}{\pgfqpoint{0.000000in}{-0.048611in}}{\pgfqpoint{0.000000in}{0.000000in}}{%
\pgfpathmoveto{\pgfqpoint{0.000000in}{0.000000in}}%
\pgfpathlineto{\pgfqpoint{0.000000in}{-0.048611in}}%
\pgfusepath{stroke,fill}%
}%
\begin{pgfscope}%
\pgfsys@transformshift{5.604945in}{0.685729in}%
\pgfsys@useobject{currentmarker}{}%
\end{pgfscope}%
\end{pgfscope}%
\begin{pgfscope}%
\definecolor{textcolor}{rgb}{0.000000,0.000000,0.000000}%
\pgfsetstrokecolor{textcolor}%
\pgfsetfillcolor{textcolor}%
\pgftext[x=5.604945in,y=0.588507in,,top]{\color{textcolor}\ttfamily\fontsize{16.000000}{19.200000}\selectfont 1.0}%
\end{pgfscope}%
\begin{pgfscope}%
\definecolor{textcolor}{rgb}{0.000000,0.000000,0.000000}%
\pgfsetstrokecolor{textcolor}%
\pgfsetfillcolor{textcolor}%
\pgftext[x=3.350399in,y=0.316697in,,top]{\color{textcolor}\ttfamily\fontsize{16.000000}{19.200000}\selectfont decorrelation parameter \(\displaystyle \gamma\)}%
\end{pgfscope}%
\begin{pgfscope}%
\pgfpathrectangle{\pgfqpoint{0.870399in}{0.685729in}}{\pgfqpoint{4.960000in}{3.696000in}}%
\pgfusepath{clip}%
\pgfsetrectcap%
\pgfsetroundjoin%
\pgfsetlinewidth{0.803000pt}%
\definecolor{currentstroke}{rgb}{0.690196,0.690196,0.690196}%
\pgfsetstrokecolor{currentstroke}%
\pgfsetstrokeopacity{0.500000}%
\pgfsetdash{}{0pt}%
\pgfpathmoveto{\pgfqpoint{0.870399in}{0.853729in}}%
\pgfpathlineto{\pgfqpoint{5.830399in}{0.853729in}}%
\pgfusepath{stroke}%
\end{pgfscope}%
\begin{pgfscope}%
\pgfsetbuttcap%
\pgfsetroundjoin%
\definecolor{currentfill}{rgb}{0.000000,0.000000,0.000000}%
\pgfsetfillcolor{currentfill}%
\pgfsetlinewidth{0.803000pt}%
\definecolor{currentstroke}{rgb}{0.000000,0.000000,0.000000}%
\pgfsetstrokecolor{currentstroke}%
\pgfsetdash{}{0pt}%
\pgfsys@defobject{currentmarker}{\pgfqpoint{-0.048611in}{0.000000in}}{\pgfqpoint{-0.000000in}{0.000000in}}{%
\pgfpathmoveto{\pgfqpoint{-0.000000in}{0.000000in}}%
\pgfpathlineto{\pgfqpoint{-0.048611in}{0.000000in}}%
\pgfusepath{stroke,fill}%
}%
\begin{pgfscope}%
\pgfsys@transformshift{0.870399in}{0.853729in}%
\pgfsys@useobject{currentmarker}{}%
\end{pgfscope}%
\end{pgfscope}%
\begin{pgfscope}%
\definecolor{textcolor}{rgb}{0.000000,0.000000,0.000000}%
\pgfsetstrokecolor{textcolor}%
\pgfsetfillcolor{textcolor}%
\pgftext[x=0.371810in, y=0.768714in, left, base]{\color{textcolor}\ttfamily\fontsize{16.000000}{19.200000}\selectfont 0.0}%
\end{pgfscope}%
\begin{pgfscope}%
\pgfpathrectangle{\pgfqpoint{0.870399in}{0.685729in}}{\pgfqpoint{4.960000in}{3.696000in}}%
\pgfusepath{clip}%
\pgfsetrectcap%
\pgfsetroundjoin%
\pgfsetlinewidth{0.803000pt}%
\definecolor{currentstroke}{rgb}{0.690196,0.690196,0.690196}%
\pgfsetstrokecolor{currentstroke}%
\pgfsetstrokeopacity{0.500000}%
\pgfsetdash{}{0pt}%
\pgfpathmoveto{\pgfqpoint{0.870399in}{1.525729in}}%
\pgfpathlineto{\pgfqpoint{5.830399in}{1.525729in}}%
\pgfusepath{stroke}%
\end{pgfscope}%
\begin{pgfscope}%
\pgfsetbuttcap%
\pgfsetroundjoin%
\definecolor{currentfill}{rgb}{0.000000,0.000000,0.000000}%
\pgfsetfillcolor{currentfill}%
\pgfsetlinewidth{0.803000pt}%
\definecolor{currentstroke}{rgb}{0.000000,0.000000,0.000000}%
\pgfsetstrokecolor{currentstroke}%
\pgfsetdash{}{0pt}%
\pgfsys@defobject{currentmarker}{\pgfqpoint{-0.048611in}{0.000000in}}{\pgfqpoint{-0.000000in}{0.000000in}}{%
\pgfpathmoveto{\pgfqpoint{-0.000000in}{0.000000in}}%
\pgfpathlineto{\pgfqpoint{-0.048611in}{0.000000in}}%
\pgfusepath{stroke,fill}%
}%
\begin{pgfscope}%
\pgfsys@transformshift{0.870399in}{1.525729in}%
\pgfsys@useobject{currentmarker}{}%
\end{pgfscope}%
\end{pgfscope}%
\begin{pgfscope}%
\definecolor{textcolor}{rgb}{0.000000,0.000000,0.000000}%
\pgfsetstrokecolor{textcolor}%
\pgfsetfillcolor{textcolor}%
\pgftext[x=0.371810in, y=1.440714in, left, base]{\color{textcolor}\ttfamily\fontsize{16.000000}{19.200000}\selectfont 0.1}%
\end{pgfscope}%
\begin{pgfscope}%
\pgfpathrectangle{\pgfqpoint{0.870399in}{0.685729in}}{\pgfqpoint{4.960000in}{3.696000in}}%
\pgfusepath{clip}%
\pgfsetrectcap%
\pgfsetroundjoin%
\pgfsetlinewidth{0.803000pt}%
\definecolor{currentstroke}{rgb}{0.690196,0.690196,0.690196}%
\pgfsetstrokecolor{currentstroke}%
\pgfsetstrokeopacity{0.500000}%
\pgfsetdash{}{0pt}%
\pgfpathmoveto{\pgfqpoint{0.870399in}{2.197729in}}%
\pgfpathlineto{\pgfqpoint{5.830399in}{2.197729in}}%
\pgfusepath{stroke}%
\end{pgfscope}%
\begin{pgfscope}%
\pgfsetbuttcap%
\pgfsetroundjoin%
\definecolor{currentfill}{rgb}{0.000000,0.000000,0.000000}%
\pgfsetfillcolor{currentfill}%
\pgfsetlinewidth{0.803000pt}%
\definecolor{currentstroke}{rgb}{0.000000,0.000000,0.000000}%
\pgfsetstrokecolor{currentstroke}%
\pgfsetdash{}{0pt}%
\pgfsys@defobject{currentmarker}{\pgfqpoint{-0.048611in}{0.000000in}}{\pgfqpoint{-0.000000in}{0.000000in}}{%
\pgfpathmoveto{\pgfqpoint{-0.000000in}{0.000000in}}%
\pgfpathlineto{\pgfqpoint{-0.048611in}{0.000000in}}%
\pgfusepath{stroke,fill}%
}%
\begin{pgfscope}%
\pgfsys@transformshift{0.870399in}{2.197729in}%
\pgfsys@useobject{currentmarker}{}%
\end{pgfscope}%
\end{pgfscope}%
\begin{pgfscope}%
\definecolor{textcolor}{rgb}{0.000000,0.000000,0.000000}%
\pgfsetstrokecolor{textcolor}%
\pgfsetfillcolor{textcolor}%
\pgftext[x=0.371810in, y=2.112714in, left, base]{\color{textcolor}\ttfamily\fontsize{16.000000}{19.200000}\selectfont 0.2}%
\end{pgfscope}%
\begin{pgfscope}%
\pgfpathrectangle{\pgfqpoint{0.870399in}{0.685729in}}{\pgfqpoint{4.960000in}{3.696000in}}%
\pgfusepath{clip}%
\pgfsetrectcap%
\pgfsetroundjoin%
\pgfsetlinewidth{0.803000pt}%
\definecolor{currentstroke}{rgb}{0.690196,0.690196,0.690196}%
\pgfsetstrokecolor{currentstroke}%
\pgfsetstrokeopacity{0.500000}%
\pgfsetdash{}{0pt}%
\pgfpathmoveto{\pgfqpoint{0.870399in}{2.869729in}}%
\pgfpathlineto{\pgfqpoint{5.830399in}{2.869729in}}%
\pgfusepath{stroke}%
\end{pgfscope}%
\begin{pgfscope}%
\pgfsetbuttcap%
\pgfsetroundjoin%
\definecolor{currentfill}{rgb}{0.000000,0.000000,0.000000}%
\pgfsetfillcolor{currentfill}%
\pgfsetlinewidth{0.803000pt}%
\definecolor{currentstroke}{rgb}{0.000000,0.000000,0.000000}%
\pgfsetstrokecolor{currentstroke}%
\pgfsetdash{}{0pt}%
\pgfsys@defobject{currentmarker}{\pgfqpoint{-0.048611in}{0.000000in}}{\pgfqpoint{-0.000000in}{0.000000in}}{%
\pgfpathmoveto{\pgfqpoint{-0.000000in}{0.000000in}}%
\pgfpathlineto{\pgfqpoint{-0.048611in}{0.000000in}}%
\pgfusepath{stroke,fill}%
}%
\begin{pgfscope}%
\pgfsys@transformshift{0.870399in}{2.869729in}%
\pgfsys@useobject{currentmarker}{}%
\end{pgfscope}%
\end{pgfscope}%
\begin{pgfscope}%
\definecolor{textcolor}{rgb}{0.000000,0.000000,0.000000}%
\pgfsetstrokecolor{textcolor}%
\pgfsetfillcolor{textcolor}%
\pgftext[x=0.371810in, y=2.784714in, left, base]{\color{textcolor}\ttfamily\fontsize{16.000000}{19.200000}\selectfont 0.3}%
\end{pgfscope}%
\begin{pgfscope}%
\pgfpathrectangle{\pgfqpoint{0.870399in}{0.685729in}}{\pgfqpoint{4.960000in}{3.696000in}}%
\pgfusepath{clip}%
\pgfsetrectcap%
\pgfsetroundjoin%
\pgfsetlinewidth{0.803000pt}%
\definecolor{currentstroke}{rgb}{0.690196,0.690196,0.690196}%
\pgfsetstrokecolor{currentstroke}%
\pgfsetstrokeopacity{0.500000}%
\pgfsetdash{}{0pt}%
\pgfpathmoveto{\pgfqpoint{0.870399in}{3.541729in}}%
\pgfpathlineto{\pgfqpoint{5.830399in}{3.541729in}}%
\pgfusepath{stroke}%
\end{pgfscope}%
\begin{pgfscope}%
\pgfsetbuttcap%
\pgfsetroundjoin%
\definecolor{currentfill}{rgb}{0.000000,0.000000,0.000000}%
\pgfsetfillcolor{currentfill}%
\pgfsetlinewidth{0.803000pt}%
\definecolor{currentstroke}{rgb}{0.000000,0.000000,0.000000}%
\pgfsetstrokecolor{currentstroke}%
\pgfsetdash{}{0pt}%
\pgfsys@defobject{currentmarker}{\pgfqpoint{-0.048611in}{0.000000in}}{\pgfqpoint{-0.000000in}{0.000000in}}{%
\pgfpathmoveto{\pgfqpoint{-0.000000in}{0.000000in}}%
\pgfpathlineto{\pgfqpoint{-0.048611in}{0.000000in}}%
\pgfusepath{stroke,fill}%
}%
\begin{pgfscope}%
\pgfsys@transformshift{0.870399in}{3.541729in}%
\pgfsys@useobject{currentmarker}{}%
\end{pgfscope}%
\end{pgfscope}%
\begin{pgfscope}%
\definecolor{textcolor}{rgb}{0.000000,0.000000,0.000000}%
\pgfsetstrokecolor{textcolor}%
\pgfsetfillcolor{textcolor}%
\pgftext[x=0.371810in, y=3.456714in, left, base]{\color{textcolor}\ttfamily\fontsize{16.000000}{19.200000}\selectfont 0.4}%
\end{pgfscope}%
\begin{pgfscope}%
\pgfpathrectangle{\pgfqpoint{0.870399in}{0.685729in}}{\pgfqpoint{4.960000in}{3.696000in}}%
\pgfusepath{clip}%
\pgfsetrectcap%
\pgfsetroundjoin%
\pgfsetlinewidth{0.803000pt}%
\definecolor{currentstroke}{rgb}{0.690196,0.690196,0.690196}%
\pgfsetstrokecolor{currentstroke}%
\pgfsetstrokeopacity{0.500000}%
\pgfsetdash{}{0pt}%
\pgfpathmoveto{\pgfqpoint{0.870399in}{4.213729in}}%
\pgfpathlineto{\pgfqpoint{5.830399in}{4.213729in}}%
\pgfusepath{stroke}%
\end{pgfscope}%
\begin{pgfscope}%
\pgfsetbuttcap%
\pgfsetroundjoin%
\definecolor{currentfill}{rgb}{0.000000,0.000000,0.000000}%
\pgfsetfillcolor{currentfill}%
\pgfsetlinewidth{0.803000pt}%
\definecolor{currentstroke}{rgb}{0.000000,0.000000,0.000000}%
\pgfsetstrokecolor{currentstroke}%
\pgfsetdash{}{0pt}%
\pgfsys@defobject{currentmarker}{\pgfqpoint{-0.048611in}{0.000000in}}{\pgfqpoint{-0.000000in}{0.000000in}}{%
\pgfpathmoveto{\pgfqpoint{-0.000000in}{0.000000in}}%
\pgfpathlineto{\pgfqpoint{-0.048611in}{0.000000in}}%
\pgfusepath{stroke,fill}%
}%
\begin{pgfscope}%
\pgfsys@transformshift{0.870399in}{4.213729in}%
\pgfsys@useobject{currentmarker}{}%
\end{pgfscope}%
\end{pgfscope}%
\begin{pgfscope}%
\definecolor{textcolor}{rgb}{0.000000,0.000000,0.000000}%
\pgfsetstrokecolor{textcolor}%
\pgfsetfillcolor{textcolor}%
\pgftext[x=0.371810in, y=4.128714in, left, base]{\color{textcolor}\ttfamily\fontsize{16.000000}{19.200000}\selectfont 0.5}%
\end{pgfscope}%
\begin{pgfscope}%
\definecolor{textcolor}{rgb}{0.000000,0.000000,0.000000}%
\pgfsetstrokecolor{textcolor}%
\pgfsetfillcolor{textcolor}%
\pgftext[x=0.316254in,y=2.533729in,,bottom,rotate=90.000000]{\color{textcolor}\ttfamily\fontsize{16.000000}{19.200000}\selectfont fraction of ideal utility}%
\end{pgfscope}%
\begin{pgfscope}%
\pgfpathrectangle{\pgfqpoint{0.870399in}{0.685729in}}{\pgfqpoint{4.960000in}{3.696000in}}%
\pgfusepath{clip}%
\pgfsetrectcap%
\pgfsetroundjoin%
\pgfsetlinewidth{1.505625pt}%
\definecolor{currentstroke}{rgb}{1.000000,0.341176,0.200000}%
\pgfsetstrokecolor{currentstroke}%
\pgfsetdash{}{0pt}%
\pgfpathmoveto{\pgfqpoint{1.095854in}{0.853729in}}%
\pgfpathlineto{\pgfqpoint{1.141400in}{0.870785in}}%
\pgfpathlineto{\pgfqpoint{1.186946in}{0.888015in}}%
\pgfpathlineto{\pgfqpoint{1.232493in}{0.905424in}}%
\pgfpathlineto{\pgfqpoint{1.278039in}{0.923015in}}%
\pgfpathlineto{\pgfqpoint{1.323586in}{0.940790in}}%
\pgfpathlineto{\pgfqpoint{1.369132in}{0.958755in}}%
\pgfpathlineto{\pgfqpoint{1.414678in}{0.976912in}}%
\pgfpathlineto{\pgfqpoint{1.460225in}{0.995265in}}%
\pgfpathlineto{\pgfqpoint{1.505771in}{1.013820in}}%
\pgfpathlineto{\pgfqpoint{1.551317in}{1.032579in}}%
\pgfpathlineto{\pgfqpoint{1.596864in}{1.051547in}}%
\pgfpathlineto{\pgfqpoint{1.642410in}{1.070729in}}%
\pgfpathlineto{\pgfqpoint{1.687957in}{1.090129in}}%
\pgfpathlineto{\pgfqpoint{1.733503in}{1.109752in}}%
\pgfpathlineto{\pgfqpoint{1.779049in}{1.129603in}}%
\pgfpathlineto{\pgfqpoint{1.824596in}{1.149686in}}%
\pgfpathlineto{\pgfqpoint{1.870142in}{1.170008in}}%
\pgfpathlineto{\pgfqpoint{1.915688in}{1.190573in}}%
\pgfpathlineto{\pgfqpoint{1.961235in}{1.211387in}}%
\pgfpathlineto{\pgfqpoint{2.006781in}{1.232456in}}%
\pgfpathlineto{\pgfqpoint{2.052328in}{1.253786in}}%
\pgfpathlineto{\pgfqpoint{2.097874in}{1.275382in}}%
\pgfpathlineto{\pgfqpoint{2.143420in}{1.297253in}}%
\pgfpathlineto{\pgfqpoint{2.188967in}{1.319403in}}%
\pgfpathlineto{\pgfqpoint{2.234513in}{1.341840in}}%
\pgfpathlineto{\pgfqpoint{2.280059in}{1.364571in}}%
\pgfpathlineto{\pgfqpoint{2.325606in}{1.387604in}}%
\pgfpathlineto{\pgfqpoint{2.371152in}{1.410946in}}%
\pgfpathlineto{\pgfqpoint{2.416699in}{1.434605in}}%
\pgfpathlineto{\pgfqpoint{2.462245in}{1.458590in}}%
\pgfpathlineto{\pgfqpoint{2.507791in}{1.482909in}}%
\pgfpathlineto{\pgfqpoint{2.553338in}{1.507572in}}%
\pgfpathlineto{\pgfqpoint{2.598884in}{1.532587in}}%
\pgfpathlineto{\pgfqpoint{2.644430in}{1.557964in}}%
\pgfpathlineto{\pgfqpoint{2.689977in}{1.583715in}}%
\pgfpathlineto{\pgfqpoint{2.735523in}{1.609849in}}%
\pgfpathlineto{\pgfqpoint{2.781070in}{1.636377in}}%
\pgfpathlineto{\pgfqpoint{2.826616in}{1.663311in}}%
\pgfpathlineto{\pgfqpoint{2.872162in}{1.690663in}}%
\pgfpathlineto{\pgfqpoint{2.917709in}{1.718447in}}%
\pgfpathlineto{\pgfqpoint{2.963255in}{1.746674in}}%
\pgfpathlineto{\pgfqpoint{3.008801in}{1.775360in}}%
\pgfpathlineto{\pgfqpoint{3.054348in}{1.804518in}}%
\pgfpathlineto{\pgfqpoint{3.099894in}{1.834164in}}%
\pgfpathlineto{\pgfqpoint{3.145440in}{1.864314in}}%
\pgfpathlineto{\pgfqpoint{3.190987in}{1.894984in}}%
\pgfpathlineto{\pgfqpoint{3.236533in}{1.926194in}}%
\pgfpathlineto{\pgfqpoint{3.282080in}{1.957960in}}%
\pgfpathlineto{\pgfqpoint{3.327626in}{1.990304in}}%
\pgfpathlineto{\pgfqpoint{3.373172in}{2.023245in}}%
\pgfpathlineto{\pgfqpoint{3.418719in}{2.056806in}}%
\pgfpathlineto{\pgfqpoint{3.464265in}{2.091010in}}%
\pgfpathlineto{\pgfqpoint{3.509811in}{2.125880in}}%
\pgfpathlineto{\pgfqpoint{3.555358in}{2.161444in}}%
\pgfpathlineto{\pgfqpoint{3.600904in}{2.197729in}}%
\pgfpathlineto{\pgfqpoint{3.646451in}{2.234763in}}%
\pgfpathlineto{\pgfqpoint{3.691997in}{2.272579in}}%
\pgfpathlineto{\pgfqpoint{3.737543in}{2.311208in}}%
\pgfpathlineto{\pgfqpoint{3.783090in}{2.350687in}}%
\pgfpathlineto{\pgfqpoint{3.828636in}{2.391053in}}%
\pgfpathlineto{\pgfqpoint{3.874182in}{2.432346in}}%
\pgfpathlineto{\pgfqpoint{3.919729in}{2.474388in}}%
\pgfpathlineto{\pgfqpoint{3.965275in}{2.516759in}}%
\pgfpathlineto{\pgfqpoint{4.010822in}{2.559444in}}%
\pgfpathlineto{\pgfqpoint{4.056368in}{2.602435in}}%
\pgfpathlineto{\pgfqpoint{4.101914in}{2.645729in}}%
\pgfpathlineto{\pgfqpoint{4.147461in}{2.689319in}}%
\pgfpathlineto{\pgfqpoint{4.193007in}{2.733199in}}%
\pgfpathlineto{\pgfqpoint{4.238553in}{2.777365in}}%
\pgfpathlineto{\pgfqpoint{4.284100in}{2.821812in}}%
\pgfpathlineto{\pgfqpoint{4.329646in}{2.866535in}}%
\pgfpathlineto{\pgfqpoint{4.375193in}{2.911528in}}%
\pgfpathlineto{\pgfqpoint{4.420739in}{2.956787in}}%
\pgfpathlineto{\pgfqpoint{4.466285in}{3.002308in}}%
\pgfpathlineto{\pgfqpoint{4.511832in}{3.048086in}}%
\pgfpathlineto{\pgfqpoint{4.557378in}{3.094117in}}%
\pgfpathlineto{\pgfqpoint{4.602924in}{3.140396in}}%
\pgfpathlineto{\pgfqpoint{4.648471in}{3.186918in}}%
\pgfpathlineto{\pgfqpoint{4.694017in}{3.233681in}}%
\pgfpathlineto{\pgfqpoint{4.739564in}{3.280680in}}%
\pgfpathlineto{\pgfqpoint{4.785110in}{3.327911in}}%
\pgfpathlineto{\pgfqpoint{4.830656in}{3.375370in}}%
\pgfpathlineto{\pgfqpoint{4.876203in}{3.423053in}}%
\pgfpathlineto{\pgfqpoint{4.921749in}{3.470957in}}%
\pgfpathlineto{\pgfqpoint{4.967295in}{3.519078in}}%
\pgfpathlineto{\pgfqpoint{5.012842in}{3.567413in}}%
\pgfpathlineto{\pgfqpoint{5.058388in}{3.615958in}}%
\pgfpathlineto{\pgfqpoint{5.103935in}{3.664709in}}%
\pgfpathlineto{\pgfqpoint{5.149481in}{3.713664in}}%
\pgfpathlineto{\pgfqpoint{5.195027in}{3.762820in}}%
\pgfpathlineto{\pgfqpoint{5.240574in}{3.812172in}}%
\pgfpathlineto{\pgfqpoint{5.286120in}{3.861719in}}%
\pgfpathlineto{\pgfqpoint{5.331666in}{3.911456in}}%
\pgfpathlineto{\pgfqpoint{5.377213in}{3.961382in}}%
\pgfpathlineto{\pgfqpoint{5.422759in}{4.011492in}}%
\pgfpathlineto{\pgfqpoint{5.468306in}{4.061785in}}%
\pgfpathlineto{\pgfqpoint{5.513852in}{4.112257in}}%
\pgfpathlineto{\pgfqpoint{5.559398in}{4.162906in}}%
\pgfpathlineto{\pgfqpoint{5.604945in}{4.213729in}}%
\pgfusepath{stroke}%
\end{pgfscope}%
\begin{pgfscope}%
\pgfpathrectangle{\pgfqpoint{0.870399in}{0.685729in}}{\pgfqpoint{4.960000in}{3.696000in}}%
\pgfusepath{clip}%
\pgfsetbuttcap%
\pgfsetroundjoin%
\pgfsetlinewidth{1.505625pt}%
\definecolor{currentstroke}{rgb}{0.501961,0.501961,0.501961}%
\pgfsetstrokecolor{currentstroke}%
\pgfsetstrokeopacity{0.500000}%
\pgfsetdash{{5.550000pt}{2.400000pt}}{0.000000pt}%
\pgfpathmoveto{\pgfqpoint{1.095854in}{0.853729in}}%
\pgfpathlineto{\pgfqpoint{1.141400in}{0.854408in}}%
\pgfpathlineto{\pgfqpoint{1.186946in}{0.856417in}}%
\pgfpathlineto{\pgfqpoint{1.232493in}{0.859718in}}%
\pgfpathlineto{\pgfqpoint{1.278039in}{0.864273in}}%
\pgfpathlineto{\pgfqpoint{1.323586in}{0.870046in}}%
\pgfpathlineto{\pgfqpoint{1.369132in}{0.877002in}}%
\pgfpathlineto{\pgfqpoint{1.414678in}{0.885107in}}%
\pgfpathlineto{\pgfqpoint{1.460225in}{0.894329in}}%
\pgfpathlineto{\pgfqpoint{1.505771in}{0.904638in}}%
\pgfpathlineto{\pgfqpoint{1.551317in}{0.916003in}}%
\pgfpathlineto{\pgfqpoint{1.596864in}{0.928396in}}%
\pgfpathlineto{\pgfqpoint{1.642410in}{0.941788in}}%
\pgfpathlineto{\pgfqpoint{1.687957in}{0.956153in}}%
\pgfpathlineto{\pgfqpoint{1.733503in}{0.971466in}}%
\pgfpathlineto{\pgfqpoint{1.779049in}{0.987700in}}%
\pgfpathlineto{\pgfqpoint{1.824596in}{1.004833in}}%
\pgfpathlineto{\pgfqpoint{1.870142in}{1.022841in}}%
\pgfpathlineto{\pgfqpoint{1.915688in}{1.041701in}}%
\pgfpathlineto{\pgfqpoint{1.961235in}{1.061392in}}%
\pgfpathlineto{\pgfqpoint{2.006781in}{1.081893in}}%
\pgfpathlineto{\pgfqpoint{2.052328in}{1.103183in}}%
\pgfpathlineto{\pgfqpoint{2.097874in}{1.125244in}}%
\pgfpathlineto{\pgfqpoint{2.143420in}{1.148056in}}%
\pgfpathlineto{\pgfqpoint{2.188967in}{1.171600in}}%
\pgfpathlineto{\pgfqpoint{2.234513in}{1.195860in}}%
\pgfpathlineto{\pgfqpoint{2.280059in}{1.220817in}}%
\pgfpathlineto{\pgfqpoint{2.325606in}{1.246456in}}%
\pgfpathlineto{\pgfqpoint{2.371152in}{1.272760in}}%
\pgfpathlineto{\pgfqpoint{2.416699in}{1.299714in}}%
\pgfpathlineto{\pgfqpoint{2.462245in}{1.327302in}}%
\pgfpathlineto{\pgfqpoint{2.507791in}{1.355510in}}%
\pgfpathlineto{\pgfqpoint{2.553338in}{1.384323in}}%
\pgfpathlineto{\pgfqpoint{2.598884in}{1.413729in}}%
\pgfpathlineto{\pgfqpoint{2.644430in}{1.443713in}}%
\pgfpathlineto{\pgfqpoint{2.689977in}{1.474263in}}%
\pgfpathlineto{\pgfqpoint{2.735523in}{1.505365in}}%
\pgfpathlineto{\pgfqpoint{2.781070in}{1.537009in}}%
\pgfpathlineto{\pgfqpoint{2.826616in}{1.569181in}}%
\pgfpathlineto{\pgfqpoint{2.872162in}{1.601871in}}%
\pgfpathlineto{\pgfqpoint{2.917709in}{1.635067in}}%
\pgfpathlineto{\pgfqpoint{2.963255in}{1.668759in}}%
\pgfpathlineto{\pgfqpoint{3.008801in}{1.702936in}}%
\pgfpathlineto{\pgfqpoint{3.054348in}{1.737587in}}%
\pgfpathlineto{\pgfqpoint{3.099894in}{1.772703in}}%
\pgfpathlineto{\pgfqpoint{3.145440in}{1.808274in}}%
\pgfpathlineto{\pgfqpoint{3.190987in}{1.844291in}}%
\pgfpathlineto{\pgfqpoint{3.236533in}{1.880744in}}%
\pgfpathlineto{\pgfqpoint{3.282080in}{1.917625in}}%
\pgfpathlineto{\pgfqpoint{3.327626in}{1.954925in}}%
\pgfpathlineto{\pgfqpoint{3.373172in}{1.992635in}}%
\pgfpathlineto{\pgfqpoint{3.418719in}{2.030747in}}%
\pgfpathlineto{\pgfqpoint{3.464265in}{2.069254in}}%
\pgfpathlineto{\pgfqpoint{3.509811in}{2.108147in}}%
\pgfpathlineto{\pgfqpoint{3.555358in}{2.147419in}}%
\pgfpathlineto{\pgfqpoint{3.600904in}{2.187062in}}%
\pgfpathlineto{\pgfqpoint{3.646451in}{2.227070in}}%
\pgfpathlineto{\pgfqpoint{3.691997in}{2.267435in}}%
\pgfpathlineto{\pgfqpoint{3.737543in}{2.308151in}}%
\pgfpathlineto{\pgfqpoint{3.783090in}{2.349210in}}%
\pgfpathlineto{\pgfqpoint{3.828636in}{2.390607in}}%
\pgfpathlineto{\pgfqpoint{3.874182in}{2.432335in}}%
\pgfpathlineto{\pgfqpoint{3.919729in}{2.474388in}}%
\pgfpathlineto{\pgfqpoint{3.965275in}{2.516759in}}%
\pgfpathlineto{\pgfqpoint{4.010822in}{2.559444in}}%
\pgfpathlineto{\pgfqpoint{4.056368in}{2.602435in}}%
\pgfpathlineto{\pgfqpoint{4.101914in}{2.645729in}}%
\pgfpathlineto{\pgfqpoint{4.147461in}{2.689319in}}%
\pgfpathlineto{\pgfqpoint{4.193007in}{2.733199in}}%
\pgfpathlineto{\pgfqpoint{4.238553in}{2.777365in}}%
\pgfpathlineto{\pgfqpoint{4.284100in}{2.821812in}}%
\pgfpathlineto{\pgfqpoint{4.329646in}{2.866535in}}%
\pgfpathlineto{\pgfqpoint{4.375193in}{2.911528in}}%
\pgfpathlineto{\pgfqpoint{4.420739in}{2.956787in}}%
\pgfpathlineto{\pgfqpoint{4.466285in}{3.002308in}}%
\pgfpathlineto{\pgfqpoint{4.511832in}{3.048086in}}%
\pgfpathlineto{\pgfqpoint{4.557378in}{3.094117in}}%
\pgfpathlineto{\pgfqpoint{4.602924in}{3.140396in}}%
\pgfpathlineto{\pgfqpoint{4.648471in}{3.186918in}}%
\pgfpathlineto{\pgfqpoint{4.694017in}{3.233681in}}%
\pgfpathlineto{\pgfqpoint{4.739564in}{3.280680in}}%
\pgfpathlineto{\pgfqpoint{4.785110in}{3.327911in}}%
\pgfpathlineto{\pgfqpoint{4.830656in}{3.375370in}}%
\pgfpathlineto{\pgfqpoint{4.876203in}{3.423053in}}%
\pgfpathlineto{\pgfqpoint{4.921749in}{3.470957in}}%
\pgfpathlineto{\pgfqpoint{4.967295in}{3.519078in}}%
\pgfpathlineto{\pgfqpoint{5.012842in}{3.567413in}}%
\pgfpathlineto{\pgfqpoint{5.058388in}{3.615958in}}%
\pgfpathlineto{\pgfqpoint{5.103935in}{3.664709in}}%
\pgfpathlineto{\pgfqpoint{5.149481in}{3.713664in}}%
\pgfpathlineto{\pgfqpoint{5.195027in}{3.762820in}}%
\pgfpathlineto{\pgfqpoint{5.240574in}{3.812172in}}%
\pgfpathlineto{\pgfqpoint{5.286120in}{3.861719in}}%
\pgfpathlineto{\pgfqpoint{5.331666in}{3.911456in}}%
\pgfpathlineto{\pgfqpoint{5.377213in}{3.961382in}}%
\pgfpathlineto{\pgfqpoint{5.422759in}{4.011492in}}%
\pgfpathlineto{\pgfqpoint{5.468306in}{4.061785in}}%
\pgfpathlineto{\pgfqpoint{5.513852in}{4.112257in}}%
\pgfpathlineto{\pgfqpoint{5.559398in}{4.162906in}}%
\pgfpathlineto{\pgfqpoint{5.604945in}{4.213729in}}%
\pgfusepath{stroke}%
\end{pgfscope}%
\begin{pgfscope}%
\pgfpathrectangle{\pgfqpoint{0.870399in}{0.685729in}}{\pgfqpoint{4.960000in}{3.696000in}}%
\pgfusepath{clip}%
\pgfsetbuttcap%
\pgfsetroundjoin%
\pgfsetlinewidth{1.505625pt}%
\definecolor{currentstroke}{rgb}{0.000000,0.556863,0.619608}%
\pgfsetstrokecolor{currentstroke}%
\pgfsetdash{{9.600000pt}{2.400000pt}{1.500000pt}{2.400000pt}}{0.000000pt}%
\pgfpathmoveto{\pgfqpoint{1.095854in}{0.853729in}}%
\pgfpathlineto{\pgfqpoint{1.141400in}{0.870784in}}%
\pgfpathlineto{\pgfqpoint{1.186946in}{0.888008in}}%
\pgfpathlineto{\pgfqpoint{1.232493in}{0.905398in}}%
\pgfpathlineto{\pgfqpoint{1.278039in}{0.922952in}}%
\pgfpathlineto{\pgfqpoint{1.323586in}{0.940667in}}%
\pgfpathlineto{\pgfqpoint{1.369132in}{0.958542in}}%
\pgfpathlineto{\pgfqpoint{1.414678in}{0.976573in}}%
\pgfpathlineto{\pgfqpoint{1.460225in}{0.994759in}}%
\pgfpathlineto{\pgfqpoint{1.505771in}{1.013096in}}%
\pgfpathlineto{\pgfqpoint{1.551317in}{1.031584in}}%
\pgfpathlineto{\pgfqpoint{1.596864in}{1.050220in}}%
\pgfpathlineto{\pgfqpoint{1.642410in}{1.069002in}}%
\pgfpathlineto{\pgfqpoint{1.687957in}{1.087927in}}%
\pgfpathlineto{\pgfqpoint{1.733503in}{1.106994in}}%
\pgfpathlineto{\pgfqpoint{1.779049in}{1.126200in}}%
\pgfpathlineto{\pgfqpoint{1.824596in}{1.145544in}}%
\pgfpathlineto{\pgfqpoint{1.870142in}{1.165024in}}%
\pgfpathlineto{\pgfqpoint{1.915688in}{1.184638in}}%
\pgfpathlineto{\pgfqpoint{1.961235in}{1.204384in}}%
\pgfpathlineto{\pgfqpoint{2.006781in}{1.224260in}}%
\pgfpathlineto{\pgfqpoint{2.052328in}{1.244264in}}%
\pgfpathlineto{\pgfqpoint{2.097874in}{1.264396in}}%
\pgfpathlineto{\pgfqpoint{2.143420in}{1.284652in}}%
\pgfpathlineto{\pgfqpoint{2.188967in}{1.305031in}}%
\pgfpathlineto{\pgfqpoint{2.234513in}{1.325532in}}%
\pgfpathlineto{\pgfqpoint{2.280059in}{1.346153in}}%
\pgfpathlineto{\pgfqpoint{2.325606in}{1.366893in}}%
\pgfpathlineto{\pgfqpoint{2.371152in}{1.387749in}}%
\pgfpathlineto{\pgfqpoint{2.416699in}{1.408720in}}%
\pgfpathlineto{\pgfqpoint{2.462245in}{1.429805in}}%
\pgfpathlineto{\pgfqpoint{2.507791in}{1.451003in}}%
\pgfpathlineto{\pgfqpoint{2.553338in}{1.472311in}}%
\pgfpathlineto{\pgfqpoint{2.598884in}{1.493729in}}%
\pgfpathlineto{\pgfqpoint{2.644430in}{1.515254in}}%
\pgfpathlineto{\pgfqpoint{2.689977in}{1.536887in}}%
\pgfpathlineto{\pgfqpoint{2.735523in}{1.558624in}}%
\pgfpathlineto{\pgfqpoint{2.781070in}{1.580465in}}%
\pgfpathlineto{\pgfqpoint{2.826616in}{1.602409in}}%
\pgfpathlineto{\pgfqpoint{2.872162in}{1.624454in}}%
\pgfpathlineto{\pgfqpoint{2.917709in}{1.646599in}}%
\pgfpathlineto{\pgfqpoint{2.963255in}{1.668842in}}%
\pgfpathlineto{\pgfqpoint{3.008801in}{1.691183in}}%
\pgfpathlineto{\pgfqpoint{3.054348in}{1.713621in}}%
\pgfpathlineto{\pgfqpoint{3.099894in}{1.736153in}}%
\pgfpathlineto{\pgfqpoint{3.145440in}{1.758779in}}%
\pgfpathlineto{\pgfqpoint{3.190987in}{1.781498in}}%
\pgfpathlineto{\pgfqpoint{3.236533in}{1.804309in}}%
\pgfpathlineto{\pgfqpoint{3.282080in}{1.827210in}}%
\pgfpathlineto{\pgfqpoint{3.327626in}{1.850200in}}%
\pgfpathlineto{\pgfqpoint{3.373172in}{1.873279in}}%
\pgfpathlineto{\pgfqpoint{3.418719in}{1.896445in}}%
\pgfpathlineto{\pgfqpoint{3.464265in}{1.919697in}}%
\pgfpathlineto{\pgfqpoint{3.509811in}{1.943035in}}%
\pgfpathlineto{\pgfqpoint{3.555358in}{1.966456in}}%
\pgfpathlineto{\pgfqpoint{3.600904in}{1.989961in}}%
\pgfpathlineto{\pgfqpoint{3.646451in}{2.013548in}}%
\pgfpathlineto{\pgfqpoint{3.691997in}{2.037216in}}%
\pgfpathlineto{\pgfqpoint{3.737543in}{2.060964in}}%
\pgfpathlineto{\pgfqpoint{3.783090in}{2.084791in}}%
\pgfpathlineto{\pgfqpoint{3.828636in}{2.108697in}}%
\pgfpathlineto{\pgfqpoint{3.874182in}{2.132681in}}%
\pgfpathlineto{\pgfqpoint{3.919729in}{2.156741in}}%
\pgfpathlineto{\pgfqpoint{3.965275in}{2.180876in}}%
\pgfpathlineto{\pgfqpoint{4.010822in}{2.205087in}}%
\pgfpathlineto{\pgfqpoint{4.056368in}{2.229371in}}%
\pgfpathlineto{\pgfqpoint{4.101914in}{2.253729in}}%
\pgfpathlineto{\pgfqpoint{4.147461in}{2.278159in}}%
\pgfpathlineto{\pgfqpoint{4.193007in}{2.302660in}}%
\pgfpathlineto{\pgfqpoint{4.238553in}{2.327232in}}%
\pgfpathlineto{\pgfqpoint{4.284100in}{2.351875in}}%
\pgfpathlineto{\pgfqpoint{4.329646in}{2.376586in}}%
\pgfpathlineto{\pgfqpoint{4.375193in}{2.401365in}}%
\pgfpathlineto{\pgfqpoint{4.420739in}{2.426212in}}%
\pgfpathlineto{\pgfqpoint{4.466285in}{2.451126in}}%
\pgfpathlineto{\pgfqpoint{4.511832in}{2.476107in}}%
\pgfpathlineto{\pgfqpoint{4.557378in}{2.501152in}}%
\pgfpathlineto{\pgfqpoint{4.602924in}{2.526262in}}%
\pgfpathlineto{\pgfqpoint{4.648471in}{2.551436in}}%
\pgfpathlineto{\pgfqpoint{4.694017in}{2.576674in}}%
\pgfpathlineto{\pgfqpoint{4.739564in}{2.601974in}}%
\pgfpathlineto{\pgfqpoint{4.785110in}{2.627336in}}%
\pgfpathlineto{\pgfqpoint{4.830656in}{2.652759in}}%
\pgfpathlineto{\pgfqpoint{4.876203in}{2.678243in}}%
\pgfpathlineto{\pgfqpoint{4.921749in}{2.703787in}}%
\pgfpathlineto{\pgfqpoint{4.967295in}{2.729390in}}%
\pgfpathlineto{\pgfqpoint{5.012842in}{2.755052in}}%
\pgfpathlineto{\pgfqpoint{5.058388in}{2.780772in}}%
\pgfpathlineto{\pgfqpoint{5.103935in}{2.806549in}}%
\pgfpathlineto{\pgfqpoint{5.149481in}{2.832384in}}%
\pgfpathlineto{\pgfqpoint{5.195027in}{2.858274in}}%
\pgfpathlineto{\pgfqpoint{5.240574in}{2.884221in}}%
\pgfpathlineto{\pgfqpoint{5.286120in}{2.910222in}}%
\pgfpathlineto{\pgfqpoint{5.331666in}{2.936278in}}%
\pgfpathlineto{\pgfqpoint{5.377213in}{2.962388in}}%
\pgfpathlineto{\pgfqpoint{5.422759in}{2.988552in}}%
\pgfpathlineto{\pgfqpoint{5.468306in}{3.014768in}}%
\pgfpathlineto{\pgfqpoint{5.513852in}{3.041037in}}%
\pgfpathlineto{\pgfqpoint{5.559398in}{3.067357in}}%
\pgfpathlineto{\pgfqpoint{5.604945in}{3.093729in}}%
\pgfusepath{stroke}%
\end{pgfscope}%
\begin{pgfscope}%
\pgfpathrectangle{\pgfqpoint{0.870399in}{0.685729in}}{\pgfqpoint{4.960000in}{3.696000in}}%
\pgfusepath{clip}%
\pgfsetbuttcap%
\pgfsetroundjoin%
\pgfsetlinewidth{1.505625pt}%
\definecolor{currentstroke}{rgb}{0.619608,0.000000,0.239216}%
\pgfsetstrokecolor{currentstroke}%
\pgfsetdash{{1.500000pt}{2.475000pt}}{0.000000pt}%
\pgfpathmoveto{\pgfqpoint{1.095854in}{0.853729in}}%
\pgfpathlineto{\pgfqpoint{1.141400in}{0.920929in}}%
\pgfpathlineto{\pgfqpoint{1.186946in}{0.986798in}}%
\pgfpathlineto{\pgfqpoint{1.232493in}{1.051376in}}%
\pgfpathlineto{\pgfqpoint{1.278039in}{1.114700in}}%
\pgfpathlineto{\pgfqpoint{1.323586in}{1.176806in}}%
\pgfpathlineto{\pgfqpoint{1.369132in}{1.237729in}}%
\pgfpathlineto{\pgfqpoint{1.414678in}{1.297502in}}%
\pgfpathlineto{\pgfqpoint{1.460225in}{1.356159in}}%
\pgfpathlineto{\pgfqpoint{1.505771in}{1.413729in}}%
\pgfpathlineto{\pgfqpoint{1.551317in}{1.470243in}}%
\pgfpathlineto{\pgfqpoint{1.596864in}{1.525729in}}%
\pgfpathlineto{\pgfqpoint{1.642410in}{1.580215in}}%
\pgfpathlineto{\pgfqpoint{1.687957in}{1.633729in}}%
\pgfpathlineto{\pgfqpoint{1.733503in}{1.686295in}}%
\pgfpathlineto{\pgfqpoint{1.779049in}{1.737939in}}%
\pgfpathlineto{\pgfqpoint{1.824596in}{1.788685in}}%
\pgfpathlineto{\pgfqpoint{1.870142in}{1.838556in}}%
\pgfpathlineto{\pgfqpoint{1.915688in}{1.887575in}}%
\pgfpathlineto{\pgfqpoint{1.961235in}{1.935763in}}%
\pgfpathlineto{\pgfqpoint{2.006781in}{1.983141in}}%
\pgfpathlineto{\pgfqpoint{2.052328in}{2.029729in}}%
\pgfpathlineto{\pgfqpoint{2.097874in}{2.075547in}}%
\pgfpathlineto{\pgfqpoint{2.143420in}{2.120614in}}%
\pgfpathlineto{\pgfqpoint{2.188967in}{2.164948in}}%
\pgfpathlineto{\pgfqpoint{2.234513in}{2.208568in}}%
\pgfpathlineto{\pgfqpoint{2.280059in}{2.251489in}}%
\pgfpathlineto{\pgfqpoint{2.325606in}{2.293729in}}%
\pgfpathlineto{\pgfqpoint{2.371152in}{2.335304in}}%
\pgfpathlineto{\pgfqpoint{2.416699in}{2.376229in}}%
\pgfpathlineto{\pgfqpoint{2.462245in}{2.416520in}}%
\pgfpathlineto{\pgfqpoint{2.507791in}{2.456190in}}%
\pgfpathlineto{\pgfqpoint{2.553338in}{2.495256in}}%
\pgfpathlineto{\pgfqpoint{2.598884in}{2.533729in}}%
\pgfpathlineto{\pgfqpoint{2.644430in}{2.571624in}}%
\pgfpathlineto{\pgfqpoint{2.689977in}{2.608953in}}%
\pgfpathlineto{\pgfqpoint{2.735523in}{2.645729in}}%
\pgfpathlineto{\pgfqpoint{2.781070in}{2.681964in}}%
\pgfpathlineto{\pgfqpoint{2.826616in}{2.717670in}}%
\pgfpathlineto{\pgfqpoint{2.872162in}{2.752859in}}%
\pgfpathlineto{\pgfqpoint{2.917709in}{2.787542in}}%
\pgfpathlineto{\pgfqpoint{2.963255in}{2.821729in}}%
\pgfpathlineto{\pgfqpoint{3.008801in}{2.855431in}}%
\pgfpathlineto{\pgfqpoint{3.054348in}{2.888658in}}%
\pgfpathlineto{\pgfqpoint{3.099894in}{2.921421in}}%
\pgfpathlineto{\pgfqpoint{3.145440in}{2.953729in}}%
\pgfpathlineto{\pgfqpoint{3.190987in}{2.985591in}}%
\pgfpathlineto{\pgfqpoint{3.236533in}{3.017017in}}%
\pgfpathlineto{\pgfqpoint{3.282080in}{3.048015in}}%
\pgfpathlineto{\pgfqpoint{3.327626in}{3.078594in}}%
\pgfpathlineto{\pgfqpoint{3.373172in}{3.108762in}}%
\pgfpathlineto{\pgfqpoint{3.418719in}{3.138529in}}%
\pgfpathlineto{\pgfqpoint{3.464265in}{3.167901in}}%
\pgfpathlineto{\pgfqpoint{3.509811in}{3.196887in}}%
\pgfpathlineto{\pgfqpoint{3.555358in}{3.225494in}}%
\pgfpathlineto{\pgfqpoint{3.600904in}{3.253729in}}%
\pgfpathlineto{\pgfqpoint{3.646451in}{3.281600in}}%
\pgfpathlineto{\pgfqpoint{3.691997in}{3.309114in}}%
\pgfpathlineto{\pgfqpoint{3.737543in}{3.336277in}}%
\pgfpathlineto{\pgfqpoint{3.783090in}{3.363096in}}%
\pgfpathlineto{\pgfqpoint{3.828636in}{3.389578in}}%
\pgfpathlineto{\pgfqpoint{3.874182in}{3.415729in}}%
\pgfpathlineto{\pgfqpoint{3.919729in}{3.441555in}}%
\pgfpathlineto{\pgfqpoint{3.965275in}{3.467062in}}%
\pgfpathlineto{\pgfqpoint{4.010822in}{3.492256in}}%
\pgfpathlineto{\pgfqpoint{4.056368in}{3.517144in}}%
\pgfpathlineto{\pgfqpoint{4.101914in}{3.541729in}}%
\pgfpathlineto{\pgfqpoint{4.147461in}{3.566018in}}%
\pgfpathlineto{\pgfqpoint{4.193007in}{3.590016in}}%
\pgfpathlineto{\pgfqpoint{4.238553in}{3.613729in}}%
\pgfpathlineto{\pgfqpoint{4.284100in}{3.637161in}}%
\pgfpathlineto{\pgfqpoint{4.329646in}{3.660317in}}%
\pgfpathlineto{\pgfqpoint{4.375193in}{3.683203in}}%
\pgfpathlineto{\pgfqpoint{4.420739in}{3.705822in}}%
\pgfpathlineto{\pgfqpoint{4.466285in}{3.728180in}}%
\pgfpathlineto{\pgfqpoint{4.511832in}{3.750281in}}%
\pgfpathlineto{\pgfqpoint{4.557378in}{3.772129in}}%
\pgfpathlineto{\pgfqpoint{4.602924in}{3.793729in}}%
\pgfpathlineto{\pgfqpoint{4.648471in}{3.815085in}}%
\pgfpathlineto{\pgfqpoint{4.694017in}{3.836201in}}%
\pgfpathlineto{\pgfqpoint{4.739564in}{3.857081in}}%
\pgfpathlineto{\pgfqpoint{4.785110in}{3.877729in}}%
\pgfpathlineto{\pgfqpoint{4.830656in}{3.898149in}}%
\pgfpathlineto{\pgfqpoint{4.876203in}{3.918344in}}%
\pgfpathlineto{\pgfqpoint{4.921749in}{3.938319in}}%
\pgfpathlineto{\pgfqpoint{4.967295in}{3.958077in}}%
\pgfpathlineto{\pgfqpoint{5.012842in}{3.977621in}}%
\pgfpathlineto{\pgfqpoint{5.058388in}{3.996955in}}%
\pgfpathlineto{\pgfqpoint{5.103935in}{4.016082in}}%
\pgfpathlineto{\pgfqpoint{5.149481in}{4.035005in}}%
\pgfpathlineto{\pgfqpoint{5.195027in}{4.053729in}}%
\pgfpathlineto{\pgfqpoint{5.240574in}{4.072255in}}%
\pgfpathlineto{\pgfqpoint{5.286120in}{4.090588in}}%
\pgfpathlineto{\pgfqpoint{5.331666in}{4.108729in}}%
\pgfpathlineto{\pgfqpoint{5.377213in}{4.126682in}}%
\pgfpathlineto{\pgfqpoint{5.422759in}{4.144451in}}%
\pgfpathlineto{\pgfqpoint{5.468306in}{4.162037in}}%
\pgfpathlineto{\pgfqpoint{5.513852in}{4.179443in}}%
\pgfpathlineto{\pgfqpoint{5.559398in}{4.196673in}}%
\pgfpathlineto{\pgfqpoint{5.604945in}{4.213729in}}%
\pgfusepath{stroke}%
\end{pgfscope}%
\begin{pgfscope}%
\pgfsetrectcap%
\pgfsetmiterjoin%
\pgfsetlinewidth{0.803000pt}%
\definecolor{currentstroke}{rgb}{0.000000,0.000000,0.000000}%
\pgfsetstrokecolor{currentstroke}%
\pgfsetdash{}{0pt}%
\pgfpathmoveto{\pgfqpoint{0.870399in}{0.685729in}}%
\pgfpathlineto{\pgfqpoint{0.870399in}{4.381729in}}%
\pgfusepath{stroke}%
\end{pgfscope}%
\begin{pgfscope}%
\pgfsetrectcap%
\pgfsetmiterjoin%
\pgfsetlinewidth{0.803000pt}%
\definecolor{currentstroke}{rgb}{0.000000,0.000000,0.000000}%
\pgfsetstrokecolor{currentstroke}%
\pgfsetdash{}{0pt}%
\pgfpathmoveto{\pgfqpoint{5.830399in}{0.685729in}}%
\pgfpathlineto{\pgfqpoint{5.830399in}{4.381729in}}%
\pgfusepath{stroke}%
\end{pgfscope}%
\begin{pgfscope}%
\pgfsetrectcap%
\pgfsetmiterjoin%
\pgfsetlinewidth{0.803000pt}%
\definecolor{currentstroke}{rgb}{0.000000,0.000000,0.000000}%
\pgfsetstrokecolor{currentstroke}%
\pgfsetdash{}{0pt}%
\pgfpathmoveto{\pgfqpoint{0.870399in}{0.685729in}}%
\pgfpathlineto{\pgfqpoint{5.830399in}{0.685729in}}%
\pgfusepath{stroke}%
\end{pgfscope}%
\begin{pgfscope}%
\pgfsetrectcap%
\pgfsetmiterjoin%
\pgfsetlinewidth{0.803000pt}%
\definecolor{currentstroke}{rgb}{0.000000,0.000000,0.000000}%
\pgfsetstrokecolor{currentstroke}%
\pgfsetdash{}{0pt}%
\pgfpathmoveto{\pgfqpoint{0.870399in}{4.381729in}}%
\pgfpathlineto{\pgfqpoint{5.830399in}{4.381729in}}%
\pgfusepath{stroke}%
\end{pgfscope}%
\begin{pgfscope}%
\pgfsetbuttcap%
\pgfsetmiterjoin%
\pgfsetlinewidth{0.000000pt}%
\definecolor{currentstroke}{rgb}{0.800000,0.800000,0.800000}%
\pgfsetstrokecolor{currentstroke}%
\pgfsetstrokeopacity{0.000000}%
\pgfsetdash{}{0pt}%
\pgfpathmoveto{\pgfqpoint{1.025955in}{2.894490in}}%
\pgfpathlineto{\pgfqpoint{4.086675in}{2.894490in}}%
\pgfpathquadraticcurveto{\pgfqpoint{4.131120in}{2.894490in}}{\pgfqpoint{4.131120in}{2.938934in}}%
\pgfpathlineto{\pgfqpoint{4.131120in}{4.226173in}}%
\pgfpathquadraticcurveto{\pgfqpoint{4.131120in}{4.270618in}}{\pgfqpoint{4.086675in}{4.270618in}}%
\pgfpathlineto{\pgfqpoint{1.025955in}{4.270618in}}%
\pgfpathquadraticcurveto{\pgfqpoint{0.981510in}{4.270618in}}{\pgfqpoint{0.981510in}{4.226173in}}%
\pgfpathlineto{\pgfqpoint{0.981510in}{2.938934in}}%
\pgfpathquadraticcurveto{\pgfqpoint{0.981510in}{2.894490in}}{\pgfqpoint{1.025955in}{2.894490in}}%
\pgfpathlineto{\pgfqpoint{1.025955in}{2.894490in}}%
\pgfpathclose%
\pgfusepath{}%
\end{pgfscope}%
\begin{pgfscope}%
\pgfsetrectcap%
\pgfsetroundjoin%
\pgfsetlinewidth{1.505625pt}%
\definecolor{currentstroke}{rgb}{1.000000,0.341176,0.200000}%
\pgfsetstrokecolor{currentstroke}%
\pgfsetdash{}{0pt}%
\pgfpathmoveto{\pgfqpoint{1.070399in}{4.089476in}}%
\pgfpathlineto{\pgfqpoint{1.292621in}{4.089476in}}%
\pgfpathlineto{\pgfqpoint{1.514844in}{4.089476in}}%
\pgfusepath{stroke}%
\end{pgfscope}%
\begin{pgfscope}%
\definecolor{textcolor}{rgb}{0.000000,0.000000,0.000000}%
\pgfsetstrokecolor{textcolor}%
\pgfsetfillcolor{textcolor}%
\pgftext[x=1.692621in,y=4.011698in,left,base]{\color{textcolor}\ttfamily\fontsize{16.000000}{19.200000}\selectfont Corollaries \ref{cor:single:big_gamma}, \ref{cor:single:small_gamma}}%
\end{pgfscope}%
\begin{pgfscope}%
\pgfsetbuttcap%
\pgfsetroundjoin%
\pgfsetlinewidth{1.505625pt}%
\definecolor{currentstroke}{rgb}{0.501961,0.501961,0.501961}%
\pgfsetstrokecolor{currentstroke}%
\pgfsetstrokeopacity{0.500000}%
\pgfsetdash{{5.550000pt}{2.400000pt}}{0.000000pt}%
\pgfpathmoveto{\pgfqpoint{1.070399in}{3.762111in}}%
\pgfpathlineto{\pgfqpoint{1.292621in}{3.762111in}}%
\pgfpathlineto{\pgfqpoint{1.514844in}{3.762111in}}%
\pgfusepath{stroke}%
\end{pgfscope}%
\begin{pgfscope}%
\definecolor{textcolor}{rgb}{0.000000,0.000000,0.000000}%
\pgfsetstrokecolor{textcolor}%
\pgfsetfillcolor{textcolor}%
\pgftext[x=1.692621in,y=3.684333in,left,base]{\color{textcolor}\ttfamily\fontsize{16.000000}{19.200000}\selectfont Corollary \ref{cor:single:big_gamma}}%
\end{pgfscope}%
\begin{pgfscope}%
\pgfsetbuttcap%
\pgfsetroundjoin%
\pgfsetlinewidth{1.505625pt}%
\definecolor{currentstroke}{rgb}{0.000000,0.556863,0.619608}%
\pgfsetstrokecolor{currentstroke}%
\pgfsetdash{{9.600000pt}{2.400000pt}{1.500000pt}{2.400000pt}}{0.000000pt}%
\pgfpathmoveto{\pgfqpoint{1.070399in}{3.434746in}}%
\pgfpathlineto{\pgfqpoint{1.292621in}{3.434746in}}%
\pgfpathlineto{\pgfqpoint{1.514844in}{3.434746in}}%
\pgfusepath{stroke}%
\end{pgfscope}%
\begin{pgfscope}%
\definecolor{textcolor}{rgb}{0.000000,0.000000,0.000000}%
\pgfsetstrokecolor{textcolor}%
\pgfsetfillcolor{textcolor}%
\pgftext[x=1.692621in,y=3.356968in,left,base]{\color{textcolor}\ttfamily\fontsize{16.000000}{19.200000}\selectfont Corollary \ref{cor:single:general_gamma}}%
\end{pgfscope}%
\begin{pgfscope}%
\pgfsetbuttcap%
\pgfsetroundjoin%
\pgfsetlinewidth{1.505625pt}%
\definecolor{currentstroke}{rgb}{0.619608,0.000000,0.239216}%
\pgfsetstrokecolor{currentstroke}%
\pgfsetdash{{1.500000pt}{2.475000pt}}{0.000000pt}%
\pgfpathmoveto{\pgfqpoint{1.070399in}{3.107380in}}%
\pgfpathlineto{\pgfqpoint{1.292621in}{3.107380in}}%
\pgfpathlineto{\pgfqpoint{1.514844in}{3.107380in}}%
\pgfusepath{stroke}%
\end{pgfscope}%
\begin{pgfscope}%
\definecolor{textcolor}{rgb}{0.000000,0.000000,0.000000}%
\pgfsetstrokecolor{textcolor}%
\pgfsetfillcolor{textcolor}%
\pgftext[x=1.692621in,y=3.029603in,left,base]{\color{textcolor}\ttfamily\fontsize{16.000000}{19.200000}\selectfont Corollary \ref{cor:imp:markov}}%
\end{pgfscope}%
\end{pgfpicture}%
\makeatother%
\endgroup%

%% file: figures/dependence_graph.tex
\begingroup
\newcommand{\BELOW}{.3cm}
\newcommand{\RIGHT}{1cm}
\begin{tikzpicture}[->, >=stealth', line width=.9pt]
    \small
    \node[draw](s1){$S[1]$};
    \node[draw](req1)[below=\BELOW of s1]{\footnotesize $\req[1]$};

    \node[draw](s2)[right=\RIGHT of s1]{\footnotesize $S[2]$};
    \node[draw](req2)[below=\BELOW of s2]{\footnotesize $\req[2]$};
    
    \node[circle](s3)[right=\RIGHT of s2]{\footnotesize $\dots$};

    \node[draw](s4)[right=\RIGHT of s3]{\footnotesize $S[t-1]$};
    \node[draw](req4)[below=\BELOW of s4]{\footnotesize $\req[t-1]$};

    \node[draw](s5)[right=\RIGHT of s4]{\footnotesize $S[t]$};
    \node[draw](req5)[below=\BELOW of s5]{\footnotesize $\req[t]$};
    \node[draw](type5)[below right=\BELOW and \RIGHT of s5]{\footnotesize $V[t]$};
    
    \node[draw](blk)[below right=\BELOW and \RIGHT of req2]{\footnotesize $\blk[t]$};

    \path(s1) edge (s2);
    \path(s2) edge (s3);
    \path(s3) edge (s4);
    \path(s4) edge (s5);

    \path(s1) edge (req1);
    \path(s2) edge (req2);
    \path(s4) edge (req4);

    \path(s5) edge (type5);
    \path(s5) edge (req5);
    \path(type5) edge (req5);

    \path(req1) edge (blk);
    \path(req2) edge (blk);
    \path(req4) edge (blk);
\end{tikzpicture}
\endgroup

%% file: body/6.mult.tex
\section{Reusable Public Resources}
\label{sec:mult}

In this section, we study the reusable resources model, where an agent might need to use the resource for multiple rounds. 
We focus only on values independent across rounds.
We get results similar to the ones in \cref{thm:single:guar} for $\g = 1$ but with some caveats.
We touch on these when we present the \dmmf mechanism and the adjustments we must make for reusable resources.

\subsection{Single Agent model}
We note that the model of this section is identical to the one in \cite{DBLP:conf/sigecom/BanerjeeFT23}.
In every round, in addition to the agent's value in that round, there is also a duration for how long she wants to use the resource.
Every round $t$, the agent samples a type $(V[t], K[t])$ where $K[t]$ is the length of the demand and $V[t]$ is the agent's value \textit{per-round} if she gets the item for the desired duration.
In other words, if the agent is allocated the item for rounds $t, t+1, \ldots, t+K[t]-1$, then she receives total utility $V[t] K[t]$.
In contrast, if the agent is not allocated the item for \textit{all} of these rounds, then she cannot gain any utility from her round $t$ demand.

Similar to when agents have single-round durations, we assume that $(V[t], K[t])$ is sampled from some distribution $\calF$, independently of other agents and rounds (for reusable resources, we focus only on independence across rounds).
Note that this means if the agent has $K[t] > 1$ but is not allocated the item in round $t$, then in round $t+1$, the earlier demand is lost, and she draws a new demand $(V[t], K[t])$.
Finally, we assume that there is a $\kmax$ known to the mechanism such that the maximum desired duration of every agent is at most $\kmax$.
Once an agent is allocated the item for the rounds that correspond to her desired duration, we assume the item is unavailable for reallocation; this includes the agent who was allocated the item.
This commitment to not interrupt allocation models situations with large setup costs.

This model captures many practical examples, especially when research teams are competing for the use of a scientific instrument.
For example, teams of biologists might want to use a gene sequencer to run experiments whose durations have varying lengths.
After allocating the sequencer for a certain experiment, we cannot reallocate until that experiment is finished; otherwise, the samples and chemicals used are wasted.
We refer the reader to \cite{DBLP:conf/sigecom/BanerjeeFT23} for further examples that motivate this model.

\subsection{Dynamic Max-Min Fair Mechanism}
We now present the \dmmf mechanism for reusable resources.
First, note that \cref{algo:algo} as is does not lead to meaningful guarantees:
when the adversary can make reservations for multiple rounds that cannot be interrupted, the agent cannot get the resource unless she requests it before the adversary blocks it for a long period.
If the adversary's demands are really long, it is unlikely that the agent has positive value in the few rounds between two different requests by the adversary.
This means the agent either has to win rounds for which she has no value or win very infrequently.

To circumvent the issue above, the mechanism of this section limits when agents get the item for multiple rounds.
In particular, if the number of rounds the agent has won is above a certain threshold, then the agent cannot make any more multi-round demands.
We quantify that threshold with the total number of rounds $T$ and a parameter $r \ge 1$: if in round $t$ agent $i$ is requesting the item for multiple rounds and her total allocation would exceed $T \a_i / r$ if she wins, then she is not considered for this round.
We emphasize that this happens even if agent $i$ is the only agent who is requesting the good that round.
Unfortunately, because this algorithm depends on the horizon $T$, our guarantees hold only for that final round.
In \cref{thm:mult:imp}, we showcase the necessity of using the parameter $r$.
The algorithm's formal description is in \cref{algo:mult}.

\begin{algorithm}[t]
\DontPrintSemicolon
\caption{Dynamic Max-Min Fair Mechanism for reusable resources}
\label{algo:mult}
\KwIn{Horizon $T$, agents $[n]$, agents' fair shares $\{\a_i\}_{i\in[n]}$, parameter $r \ge 1$} 
\textbf{Initialize} agent allocations $A_i[0] = 0 \; \forall \, i\in [n]$\;
\While{$t \le T$}
{
    Let $\mathcal N_t =$ requesting agents and $\{d_i[t]\}_{i\in\mathcal N}$ their desired durations\;
    Let $\mathcal N'_t = \{i\in\mathcal N : d_i[t] = 1 \textrm{ or } A_i[t-1] + d_i[t] \le T\a_i/r \}$%
    \tcp*{Reject multi-round demands of \phantom{..} \hfill agents who have already won many rounds}

    \uIf{$|\mathcal N'_t| = 0$}{
        Set $i^* = \varnothing$, $A_i[t] = A_i[t-1] \; \forall \, i\in [n]$, and
        update $t \leftarrow t + 1$%
        \tcp*{No winner}
    }\Else{
        Allocate to $i^* = \argmin_{i \in \mathcal N_t}\frac{A_i[t-1] + d_i[t]}{\a_i}$ for $d_{i^*}[t]$ rounds%
        \tcp*{Minimum allocation}

        Set $A_i[t+d_{i^*}[t]-1] = A_i[t-1] + d_i[t]\One{i = i^*} \;\forall \, i\in [n]$ and
        update $t \leftarrow t + d_{i^*}[t]$\;
    }
}
\end{algorithm}

\subsection{Ideal Utility}
We now extend our definition of $\b$-ideal utility for reusable resources.
For $\b = \a_i$ this defines the ideal utility of agent $i$ in this setting, identical to the one in~\cite{DBLP:conf/sigecom/BanerjeeFT23}.

The idea of the $\b$-ideal utility of agent $i$ remains the same: it is the maximum expected per-round utility agent $i$ can get with no competition if she wins only a $\b$ fraction of the rounds in expectation.
While the idea is the same, the formal definition is more complicated than the one in single-round demands, \cref{def:b_ideal}.
We present the optimization problem that calculates it; note that this is identical to the program given in \cite{DBLP:conf/sigecom/BanerjeeFT23} for computing the ideal utility but with a general parameter $\b$ instead of the fair share $\a_i$.

\begin{equation}
\label{eq:mult:ideal}
\begin{split}
    \vstar_i(\b) \;=
    \quad
    \max_{\rho_\b:\R_+\times \N \to [0, 1]} \qquad &
    \frac{\Ex{ V K \rho_\b(V, K) }}{1 - \Ex{ \rho_\b(V, K) } + \Ex{ K \rho_\b(V, K) }}
    \\
    \textrm{such that} \qquad
    & \frac{\Ex{ K \rho_\b(V, K) }}{1 - \Ex{ \rho_\b(V, K) } + \Ex{ K \rho_\b(V, K) }} \le \b
\end{split}
\end{equation}
All expectations are taken over $(V, K) \sim \calF_i$.
\cite{DBLP:conf/sigecom/BanerjeeFT23} prove that the above problem can be turned into a linear program, and therefore, we can calculate $\rstar_\b$ in polynomial time.

We offer a short explanation of the above optimization problem in \cref{ssec:app:mult:explain_ideal}, as well as a short proof that $\vstar(\b)$ is concave in $\b$ (\cref{lem:app:mult:concave}).
We refer the reader to \cite{DBLP:conf/sigecom/BanerjeeFT23} for a more detailed explanation.

\subsection{Robust guarantees for reusable resources}
In this section, we present our robust guarantee for the \dmmf mechanism when agents have multi-round demands. 
Our guarantee is with respect to the $\b$-ideal utility for any $\b$, like our previous results.
Unlike \cref{thm:single:guar}, the bound also depends on the parameter $r$.
For a specific $\b$, the mechanism can tune $r$ (as a function of $\b$ and $\a$) to maximize the agent's guarantee.
In that case, we get the same bound as \cref{thm:single:guar} for $\g = 1$.
Tuning $r$ according to $\a$ and $\b$ is possible when agents have identical distributions and fair shares that are known to the mechanism.
In that case, the mechanism can tune $r$ to optimize the guarantee of all agents simultaneously, making the agents' utility guarantees in the last round as good as the ones in the much simpler single-round demand setting.
Alternately, the mechanism can set $r = 2$ to get the guarantee of half of the ideal utility, for all agents simultaneously.
This matches the guarantee in \cite{DBLP:conf/sigecom/BanerjeeFT23}, which the authors show is tight for \textit{any mechanism} under arbitrary distributions.

\begin{theorem}\label{thm:mult:guar}
    If in the \dmmf mechanism the agent uses a $\b$-aggressive policy for any $\b$, then for any $r \ge 1$ and behavior of the other agents it holds
    \begin{equation}\label{eq:mult:31}
        \sum_{t \in [T]} \Ex{U[t]}
        \ge
        \min\left\{
            \frac{\a}{\b r},
            1 - \frac{1 - \a}{r}
        \right\}
        \vstar(\b) T
        - \order{ \kmax\sqrt T }
        .
    \end{equation}
    If we set $r = \frac{\a + \b - \a \b}{\b}$ then we get
    $
        \sum_t \Ex{U[t]}
        \ge
        \frac{\a}{\a + \b - \a \b}
        \vstar(\b) T
        - \order*{ \kmax\sqrt T }
    $.
\end{theorem}

We note that the above guarantee is tight for the \dmmf mechanism in a strong sense.
In \cref{thm:mult:imp} we show that for every $\b$ there is a value distribution for which the agent cannot guarantee expected utility more than $\min\big\{ \frac{\a}{\b r}, 1 - \frac{1-\a}{r} \big\} \vstar(\b)T$, for any $r \ge 1$.
We show this assuming that $\kmax$ is a large constant.
The tightness in our results shows that using the parameter $r$ is necessary to get any meaningful guarantees.
Otherwise, setting $r = 1$ (a small value that is better than not using $r$) would imply that the agent can guarantee at most $\a \vstar(\b) T$ utility.

The proof of this theorem is more involved than the one in \cref{thm:single:guar}.
We first generalize \cref{lem:single:dmmf} since the adversary can now block the agent by also making multi-round demands, even if \dmmf would favor the agent over the adversary.
This new type of blocking is bounded due to the parameter $r$, which is part of the bound in \cref{lem:mult:dmmf}, the generalization of \cref{lem:single:dmmf}.
The rest of the proof involves showing that the agent gains enough value from every win and that she is not limited by the restriction due to the parameter $r$.
We present the full proof in \cref{sec:app:mult}.

%% file: body/7.imposs.tex
\section{Impossibility Results}
\label{sec:imp}

In this section, we prove upper bounds on the agent's utility.
In \cref{ssec:imp:ind} we upper bound the fraction of ideal utility that \emph{any mechanism} can guarantee, for single-round demands.
Then, in \cref{ssec:imp:dep,ssec:imp:mult}, we focus only on the \dmmf mechanism.
First, for single-round demands, we prove utility upper bounds for the case when agents have values dependent across time; we show that our guarantees in \cref{cor:single:big_gamma,cor:single:general_gamma,cor:single:pi_min} are tight up to constant factors.
Second, we turn our attention to the reusable resource setting and prove that our bound in \cref{thm:mult:guar} is tight for every $r \ge 1$ and $\b \in (0, 1]$.
We include the full proofs of the subsequent theorems in \cref{sec:app:imp}.

\subsection{Impossibility result for independent values}
\label{ssec:imp:ind}

In this section, we upper-bound the achievable guarantees of any mechanism for single-round demands.
Specifically, we show that when agents are symmetric with values sampled from a Bernoulli distribution, then no mechanism can guarantee every agent more than a $1 - \frac{1}{e} \approx .63$ fraction of her ideal utility.
This implies that no mechanism can get guarantees much better than our bound in \cref{cor:indep:worst}.

\begin{theorem}\label{thm:imp:ind}
    When agents have single-round demands, even if their values are independent across time, no mechanism can guarantee every agent, for any value distribution, more than a $1 - \frac 1 e$ fraction of ideal utility in expectation, as $n \to \infty$.
\end{theorem}

\subsection{Impossibility results for dependent values}
\label{ssec:imp:dep}

In this section, we show upper bounds for the achievable utility of the agent when faced with adversarial competition.
Our upper bound is parametrized by $\g$, the decorrelation parameter we used in \cref{sec:single}.
We first provide a general bound for every $\g$ and any fair share $\a$ of the agent.
We prove this upper bound by assuming that the adversary has very limited information: she only knows which past rounds the agent has won.
We subsequently use that to prove bounds in two cases.
First, if $\a$ is much smaller than $\g$, our theorem implies that the agent cannot achieve more than a $\frac{\g}{1 + \g}$ fraction of her ideal utility.
This makes our guarantee of $\approx\g/4$ in \cref{cor:single:general_gamma} tight, up to a constant of $4$ and additive constant factors (see also~\cref{fig:bounds_comp}).
Second, if $\g$ is much smaller than $\a$, our theorem implies that the agent cannot achieve more than a $\frac{\min_s \pi(s)}{1 - \a}$ fraction of her ideal utility.
This makes our guarantee of $\frac{\min_s \pi(s)}{1 - \a + \min_s \pi(s)}$ in \cref{cor:single:pi_min} approximately tight when $\min_s \pi(s)$ and $\a$ are small.

\begin{theorem}\label{thm:imp:dep}
    Assume that the adversary knows only which past rounds in which the agent won the item.
    Then for any $\a, \g$ there is an example where the agent's Markov chain satisfies $p(s',s) \ge \g \pi(s)$ for all $s,s'\in \calS$ and the agent cannot guarantee a total expected utility of more than
    \begin{equation}\label{eq:imp:thm1}
        \frac{\g}{(1-\a)\big( 1 + \g - (1 - \g)^{\frac{1-\a}{\a}} \big)}
        \vstar(\a) T
        +
        O(1)
    \end{equation}
    in the \dmmf mechanism.
    Additionally, if $\a \le 1/2$, in that example it holds that $\min_s \pi(s) = \a$.
\end{theorem}

\begin{corollary}\label{cor:imp:markov}
    The term in front of $\vstar(\a)T$ in \eqref{eq:imp:thm1} is at most
    \begin{itemize}[leftmargin=20pt]
        \item $\frac{\g}{1 + \g}$ as $\a\to 0$.
        \item $\frac{\a}{1 - \a} = \frac{\min_s\pi_s}{1 - \a}$ as $\g \to 0$, if $\a \le 1/2$.
    \end{itemize}
\end{corollary}

We give intuition about the proof for the case $\g \gg \a$.
We consider a Markov chain with two states, state $1$ and state $2$.
The agent's value in state $1$ is always $1$ and her value in state $2$ is always $0$.
We set up the transition probabilities to make each state as ``sticky'' as possible, given the constraint $p(s', s) \ge \g \pi(s)$.
This implies that when the agent stays in state $1$ for roughly $1 + 1 / \g$ rounds in expectation.
The adversary can take advantage of this structure: whenever the agent wins a round in state $1$, because of the properties of the \dmmf mechanism, the adversary can win the next $\frac{1 - \a}{\a}$ rounds.
This makes the agent miss out on any utility during these rounds, the first $1/\g$ of which are state $1$ rounds in expectation.
This makes the agent miss out on roughly a $\frac{1/\g}{1 + 1/\g}$ fraction of her positive-valued rounds, no matter what strategy she follows.

\subsection{Impossibility results for reusable resources}
\label{ssec:imp:mult}

We now show an upper bound on the amount of utility the agent can guarantee in the multi-round demand setting.
Our result is an upper bound for the expected utility of the agent in the \dmmf mechanism relative to $\vstar(\b) T$, for any $\b\in (0, 1]$ and any $r \ge 1$ used by the mechanism.
We note that our result holds under any strategy used by the agent, not necessarily a $\b$-aggressive strategy.

Our bound is most interesting when $\kmax$ is a large constant.
More specifically, if $\frac{\kmax - 1}{\kmax}$ is about $1$, then our result implies that the agent cannot guarantee more than $(1 - \frac{1-\a}{r}) \vstar(\b) T$ utility.
This makes our guarantee of
    $\min\big\{
        \frac{\a}{r\b} ,
        1 - \frac{1-\a}{r}
    \big\} \vstar(\b) T$
in \cref{thm:mult:guar} essentially tight, for any $r,\b$ (up to an additive error of $O(\sqrt T)$):
Because an agent with only multi-round demands cannot win more than $T \a / r$ rounds, by the definition of $\b$-ideal utility, the largest expected utility she can guarantee is $T \vstar(\a/r)$ which by concavity is at most $T \vstar(\b)\min\{1, \frac{\a}{r \b}\}$.
Combining this with the following result for large $\kmax$ proves the tightness.

\begin{theorem}\label{thm:mult:imp}
    In the \dmmf mechanism, for every $\b \in (0,1]$, there is an agent value distribution such that the agent cannot get total expected utility more than
    \begin{equation*}
        \left( 1 - \frac{1 - \a}{r}\frac{\kmax-1}{\kmax} \right)
        \vstar(\b) T
        + \vstar(\b) (\kmax - 1)
        .
    \end{equation*}
\end{theorem}

The proof considers the agent having single-round demands with Bernoulli values with mean $\b$ and the adversary requesting $\kmax$ rounds every round.
Unless the agent has high utility when the adversary makes the request, she has to either wait $\kmax$ rounds before requesting the item again or request the item in a zero-valued round.
This essentially means that for large $\kmax$, the agent misses out on any utility until the adversary wins $T \frac{1 - \a}{r}$ rounds, after which she cannot make any more multi-round demands.
This leaves only $T (1 - \frac{1 - \a}{r})$ rounds for the agent where she can only get $\vstar(\b) = \b$ utility out of each in expectation.

%% file: body/8.conclusion.tex
\section{Conclusion}

We perform an extensive study of \dmmf: a \emph{simple} mechanism for public resource sharing that, we show, is $(i)$ \emph{individually robust}: each agent receives a guaranteed utility, irrespective of how others behave; $(ii)$ \emph{globally efficient}: agents get close to optimal utility under mild assumptions on value distributions, and $(iii)$ \emph{universal}: guarantees hold in many settings, including correlated demands. 

We believe our work takes a significant step towards reinforcing the success of \dmmf in practice. Moreover, we hope our techniques for getting distribution-dependent robustness guarantees and handling correlated values can prove useful in other dynamic mechanism design problems.

%% file: appendixes/5.single.tex
\section{Deferred proofs from Section \ref{sec:single}}
\label{sec:app:markov}

\subsection{Proofs of Corollaries in Section \ref{ssec:single:indep}}
\label{ssec:app:single:indep}

We start by proving \textbf{\cref{cor:indep:density}}.
Let $F:[\underline v, \overline v] \to [0,1]$ be the CDF and $F^{-1}:[0,1] \to [\underline v, \overline v]$ be its inverse. We are going to show that for any $\beta$
\begin{equation}\label{eq:app:v_der}
    \frac{d}{d\b}\vstar(\b)
    =
    F^{-1}(1-\b)
\end{equation}
which implies
\begin{equation*}
    \frac{d^2}{d^2\b}\vstar(\b)
    =
    - \frac{1}{f(F^{-1}(1-\b))}
    \in
    \left[ -\frac{1}{\l_1}, -\frac{1}{\l_2} \right]
\end{equation*}
which in turn implies that $-\vstar$ is $\frac{1}{\l_2}$-strictly convex and $\frac{1}{\l_1}$-smooth. To show \eqref{eq:app:v_der} we use the fact that $\vstar(\b)$ is the expectation of the random variable of the value $V$, condition on it being on the top $\b$-th percentile:
\begin{equation*}
    \vstar(\b)
    =
    \Ex{V \One{V \ge F^{-1}(1-\b)}}
    =
    \int_{F^{-1}(1-\b)}^{\overline v} x f(x) dx
\end{equation*}
which implies that
\begin{alignat*}{3}
    \Line{
        \frac{d}{d\b}\vstar(\b)
    }{=}{
        \frac{d}{d\b}\int_{F^{-1}(1-\b)}^{\overline v} x f(x) dx
    }{}\\
    \Line{}{=}{
        F^{-1'}(1-\b)
        F^{-1}(1-\b)
        f\left(F^{-1}(1-\b)\right)
    }{}\\
    \Line{}{=}{
        - F^{-1}(1-\b)
        \frac{d}{d\b} \left( F\big(F^{-1}(1-\b)\big) \right)
    }{}\\
    \Line{}{=}{
        F^{-1}(1-\b)
    }{}
\end{alignat*}

We now use the $\frac{1}{\l_1}$-smoothness of $-\vstar$ on $\b, \a$:
\begin{align}\label{eq:app:1}
    & - \vstar(\b)
    \le
    - \vstar(\a)
    - F^{-1}(1 - \a) (\b - \a)
    + \frac{1}{2\l_1} (\b - \a)^2
    \nonumber\\
    \iff &
    \vstar(\b) - \vstar(\a)
    \ge
    F^{-1}(1 - \a) (\b - \a)
    - \frac{1}{2\l_1} (\b - \a)^2
\end{align}
and the $\frac{1}{\l_1}$-smoothness on $\a, 0$:
\begin{align}\label{eq:app:2}
    & - \vstar(0)
    \le
    - \vstar(\a)
    - F^{-1}(1 - \a) (0 - \a)
    + \frac{1}{2\l_1} (0 - \a)^2
    \nonumber\\
    \iff &
    \vstar(\a)
    \le
    F^{-1}(1 - \a) \a
    + \frac{1}{2\l_1} \a^2
\end{align}

Fix any $\b \ge \a$ and take:
\begin{alignat*}{3}
    \Line{
        \frac{\vstar(\b) - \vstar(\a)}{\vstar(\a)}
    }{\ge}{
        \frac{
            F^{-1}(1 - \a) (\b - \a) - \frac{1}{2\l_1} (\b - \a)^2
        }{
           F^{-1}(1 - \a) \a + \frac{1}{2\l_1} \a^2 
        }
    }{\text{by \eqref{eq:app:1}, \eqref{eq:app:2}}}
    \\
    \Line{}{\ge}{
        \frac{
            \frac{1 - \a}{\l_2} (\b - \a) - \frac{1}{2\l_1} (\b - \a)^2
        }{
           \frac{1 - \a}{\l_2} \a + \frac{1}{2\l_1} \a^2
        }
    }{F^{-1}(1 - \a) \ge \frac{1 - \a}{\l_2} + \underline v \ge \frac{1 - \a}{\l_2}}
    \\
    \Line{}{=}{
        \frac{
            (\b - \a)\big( 2 (1 - \a)\l_1 - \l_2(\b - \a) \big)
        }{
           \a\big( 2(1 - \a)\l_1 + \a \l_2 \big)
        }
    }{}
\end{alignat*}

The above makes the guarantee of \cref{thm:single:guar}
\begin{alignat*}{3}
    \Line{
        \frac{\a}{\a + \b - \a\b}\left( 1 + \frac{\vstar(\b) - \vstar(\a)}{\vstar(\a)} \right)
    }{\ge}{
        \frac{\a}{\a + \b - \a\b}
        \left( 1 + \frac{
            (\b - \a)\big( 2 (1 - \a)\l_1 - \l_2(\b - \a) \big)
        }{
           \a\big( 2(1 - \a)\l_1 + \a \l_2 \big)
        }
        \right)
    }{}
    \\
    \Line{}{=}{
        \frac{
            \b
            \big( 2 \a (\l_2 - \l_1) + 2\l_1 - b \l_2 \big)
        }{
            \big( \b + \a (1 - \b) \big)
            \big( 2 (1 - \a) \l_1 + \a \l_2 \big)
        }
    }{}
\end{alignat*}

Because $\a \le \min\left\{\frac{2\l_1}{\l_2},\frac{\l_2}{2\l_1}\right\}$ we can set $\b = \sqrt{2\a \l_1 / \l_2} \in [\a, 1]$ to get the following bound
\begin{equation*}
    \frac{\sqrt{2} \sqrt{\a \l_1} \left(\sqrt{2} \sqrt{\a \l_1 \l_2}+2 (\alpha -1) \l_1-2 \a \l_2\right)}{\left(-\sqrt{2} \sqrt{\a \l_1}+\sqrt{2} \a \sqrt{\a \l_1}-\a \sqrt{\l_2}\right) \left(2 (\alpha -1) \l_1-\a \l_2\right)}
\end{equation*}

Set $\l = \l_2 / \l_1 \ge 1$, in which case the above becomes
\begin{alignat*}{3}
    \Line{
        \frac{
            \sqrt{2} \left(2  - \sqrt{2\a \l} + 2 \a (\l-1) \right)
        }{
            (2 + \a (\l-2)) \left(\sqrt{2} + \sqrt{\a \l} - \sqrt{2} \alpha \right)
        }
    }{\ge}{
        \frac{
            \sqrt{2} \left(2  - \sqrt{2\a \l}\right)
        }{
            (2 + \a \l) \left(\sqrt{2} + \sqrt{\a \l} \right)
        }
    }{\a \ge 0}
    \\
    \Line{}{=}{
        \frac{
            4 - 4 \sqrt{2\a\l} + 2\a\l
        }{
            4 - \a^2 \l^2
        }
    }{\substack{\text{multiply numerator and}\\\text{denominator with } 2 - \a\l}}
    \\
    \Line{}{=}{
        1 - \frac{
            4 \sqrt{2\a\l} - 2\a\l - \a^2 \l^2
        }{
            4 - \a^2 \l^2
        }
    }{}
    \\
    \Line{}{=}{
        1 - \sqrt{\a \l} \frac{
            4 \sqrt{2} - 2\sqrt{\a\l} - (\a \l)^{3/2}
        }{
            4 - \a^2 \l^2
        }
    }{}
    \\
    \Line{}{\ge}{
        1 - \sqrt{\a \l} \sqrt{2}
    }{\substack{\text{fraction decreasing}\\\text{in } \a\l \text{ if } 0\le \a\l \le 2}}
\end{alignat*}
which proves the result.

We now prove \textbf{\cref{cor:indep:bernoulli}}.
It is not hard to see that for Bernoulli values with mean $p$ it holds $\vstar(\b) = \min\{p, \b\}$ for any $\b \in [0, 1]$. This makes the guarantee in \cref{thm:single:guar} for $\g = 1$ and $\b = p$
\begin{equation*}
    \frac{\sum_t \Ex{U[t]}}{\vstar(\a) T}
    +
    O\left(\frac 1 T\right)
    \ge
    \frac{\a}{\a + p - \a p} \frac{\vstar(p)}{\vstar(\a)}
    =
    \frac{\a}{\a + p - \a p} \frac{p}{\min\{\a, p\}}
    =
    \frac{\max\{\a, p\}}{\a + p - \a p}
\end{equation*}

For the second part of the guarantee in the corollary, letting $\l = \max\left\{ \frac{\a}{p}, \frac{p}{\a} \right\} = \frac{\max\{\a, p\}}{\min\{\a, p\}}$ we get
\begin{equation*}
    \frac{\max\{\a, p\}}{\a + p - \a p}
    =
    \frac{\frac{\max\{\a, p\}}{\min\{\a, p\}}}{\frac{\a + p - \a p}{\min\{\a, p\}}}
    =
    \frac{\l}{1 + \l - \max\{\a, p\}}
    \ge
    \frac{\l}{1 + \l - \a}
    =
    1 - \frac{1 - \a}{1 + \l - \a}
    \ge
    1 - \frac{1}{1 + \l}
\end{equation*}
where in the last inequality we use the fact that $\a \ge 0$.

\subsection{Proofs of Corollaries in Section \ref{ssec:single:correl}}
\label{ssec:app:correlatedguarantees}

We start by showing that
\begin{equation*}
    \frac{\g + \a(1 - \g)}{1 + (1 - \a)\g}
    \ge
    \frac{\g}{1 + \g}
    +
    \frac{\a \g}{(1 + \g)^2}
\end{equation*}
as stated in \textbf{\cref{cor:single:big_gamma}}. We notice that
\begin{alignat*}{3}
    \Line{
        \frac{\g + \a(1 - \g)}{1 + (1 - \a)\g}
        -
        \frac{\g}{1 + \g}
    }{=}{
        \frac{
            \g + \a(1 - \g)
            +
            \g^2 + \a(1 - \g)\g
            -
            \g
            -
            (1 - \a)\g^2
        }{(1 + \g)(1 + (1 - \a) \g)}
    }{}
    \\
    \Line{}{=}{
        \frac{\a}{(1 + \g)(1 + (1 - \a) \g)}
    }{}
    \\
    \Line{}{\ge}{
        \frac{\a \g}{(1 + \g)^2}
    }{\a \ge 0}
\end{alignat*}
which proves what we want.

Now we prove the equivalence and implication that was stated in \cref{eq:single:small_gamma_condition}:
\begin{equation*}
    \sqrt{1-\g} + 1 - \g
    >
    \frac{1}{1-\a}
    \iff
    \g
    <
    \frac{1}{2}\left(
        \sqrt{\frac{5-\a}{1-\a}}
        -
        \frac{1 + \a}{1 - \a}
    \right)
    \implies
    \g
    <
    \frac{\sqrt 5 - 1}{2}
\end{equation*}
For the first part let $\sqrt{1-\g} = x$ and $\frac{1}{1 - \a} = y$. The inequality in this case is equivalent with
\begin{equation*}
    x^2 + x > y
\end{equation*}
which is equivalent with
\begin{equation*}
    x + \frac{1}{2} > \sqrt{y + \frac{1}{4}}
    \text{ or }
    x + \frac{1}{2} < - \sqrt{y + \frac{1}{4}}
\end{equation*}

Since $x > 0$ we keep only the first case which, by substituting $x = \sqrt{1 - \g}$, becomes
\begin{equation*}
    \sqrt{1 - \g} > \sqrt{y + \frac{1}{4}} - \frac{1}{2}
\end{equation*}
or equivalently since $\sqrt{y + \frac{1}{4}} \ge \frac{1}{2}$
\begin{equation*}
    1 - \g
    >
    y + \frac{1}{4} 
    - \sqrt{y + \frac{1}{4}} + \frac{1}{4}
\end{equation*}
or
\begin{equation*}
    \g
    <
    \frac{1}{2}
    - y 
    + \sqrt{y + \frac{1}{4}}
    =
    \frac{1 - \a - 2}{2(1 - \a)}
    + \sqrt{\frac{4 + 1 - \a}{4(1 - \a)}}
    =
    -\frac{1 + \a}{2(1 - \a)}
    + \sqrt{\frac{5 - \a}{4(1 - \a)}}
\end{equation*}
which proves the equivalence.
To show that that above makes $\g$ at most $\frac{\sqrt 5 - 1}{2}$ we simply notice the function
\begin{equation*}
    - \frac{1 + \a}{2(1 - \a)}
    + \sqrt{\frac{5 - \a}{4(1 - \a)}}
\end{equation*}
is decreasing in $\a$ which means that is maximized when $\a = 0$ where its value is $\frac{\sqrt 5 - 1}{2}$. We note that the derivative of the function above is
\begin{equation*}
    \frac{
        \sqrt{\frac{1 - \a}{5 - \a}} - 1
    }{(1-\a)^2}
\end{equation*}
which is negative since $\frac{1 - \a}{5 - \a} < 1$.

We now prove \textbf{\cref{cor:single:small_gamma}}.
We use $\b = \frac{\a}{(1-\a)(\sqrt{1-\g} + 1 - \g)} \le \a$ (the inequality follows from \cref{eq:single:small_gamma_condition}) in \cref{thm:single:guar} and the fact that $\vstar(\b) \ge \frac{\b}{\a} \vstar(\a)$ which follows from the concavity of $\vstar$ and the fact that $\b \le \a$.
This gives the following factor of the ideal utility:
\begin{alignat*}{3}
    \Line{
        \g
        \frac{\a - (1-\a)\b(1-\g)}{\a + (1-\a)\b\g}
        \frac{\b}{\a}
    }{=}{
        \g
        \frac{\a - \frac{\a}{\sqrt{1-\g} + 1 - \g}(1-\g)}{\a + \frac{\a}{\sqrt{1-\g} + 1 - \g}\g}
        \frac{\frac{\a}{(1-\a)(\sqrt{1-\g} + 1 - \g)}}{\a}
    }{}
    \\
    \Line{}{=}{
        \g
        \frac{1 - \frac{1}{\sqrt{1-\g} + 1 - \g}(1-\g)}{1 + \frac{1}{\sqrt{1-\g} + 1 - \g}\g}
        \frac{1}{\sqrt{1-\g} + 1 - \g}
        \frac{1}{1-\a}
    }{}
    \\
    \Line{}{=}{
        \g
        \frac{
            1 - \frac{\sqrt{1-\g} - (1 - \g)}{\g}
        }{
            \sqrt{1-\g} + 1 - \g + \g
        }
        \frac{1}{1-\a}
        =\;
        \frac{
            \g - (\sqrt{1-\g} - (1 - \g))
        }{
            \sqrt{1-\g} + 1
        }
        \frac{1}{1-\a}
    }{}
    \\
    \Line{}{=}{
        \frac{
            1 - \sqrt{1-\g}
        }{
            \sqrt{1-\g} + 1
        }
        \frac{1}{1-\a}
    }{}
\end{alignat*}
which proves the first part of the guarantee. For the second part, show that
\begin{equation*}
    \frac{1 - \sqrt{1 - \g}}{1 + \sqrt{1 + \g}}
    >
    \frac{\g}{4}
\end{equation*}
as needed in \cref{cor:single:small_gamma}. We notice that
\begin{equation*}
    \frac{1 - \sqrt{1 - \g}}{1 + \sqrt{1 + \g}}
    =
    \frac{(1 - \sqrt{1 - \g})(1 + \sqrt{1 - \g})}{(1 + \sqrt{1 + \g})(1 + \sqrt{1 - \g})}
    =
    \frac{\g}{(1 + \sqrt{1 + \g})(1 + \sqrt{1 - \g})}
    \ge
    \frac{\g}{4}
\end{equation*}
where in the inequality we use the fact that $\g \ge 0$ on the denominator $(1 + \sqrt{1 + \g})(1 + \sqrt{1 - \g})$ since it is a decreasing function: its derivative is
\begin{equation*}
    \frac{\sqrt{1 - \g} - \sqrt{1 + \g} - 2\g}{2\sqrt{1 - \g^2}}
    \;\overset{\g \ge 0}{\le}\;
    \frac{\sqrt{1 - 0} - \sqrt{1 + 0} - 2\, 0}{2\sqrt{1 - \g^2}}
    \;=\;
    0
\end{equation*}

Finally, we prove the bound in \textbf{\cref{cor:single:general_gamma}}. Using $\b = \a/2$ in \cref{thm:single:guar} we get
\begin{alignat*}{3}
    \Line{
        \g
        \frac{\a - (1-\a)\b(1-\g)}{\a + (1-\a)\b\g}
        \frac{\b}{\a}
    }{=}{
        \g
        \frac{\a - (1-\a)\frac{\a}{2}(1-\g)}{\a + (1-\a)\frac{\a}{2}\g}
        \frac{1}{2}
        =\;
        \frac{\g}{2}
        \frac{2 - (1-\a)(1-\g)}{2 + (1-\a)\g}
    }{}
    \\
    \Line{}{\ge}{
        \frac{\g}{2}
        \frac{2 - (1-\g)}{2 + \g}
        =\;
        \g
        \frac{1 + \g}{4 + 2\g}
    }{\a \ge 0}
\end{alignat*}
which proves the corollary.

\subsection{Deferred proof of Lemma \ref{lem:single:dmmf}}
\label{ssec:app:lemma_proof}

Here we include the full proof of \cref{lem:single:dmmf} that is required for the proof of \cref{thm:single:guar}.

\begin{proof}[Proof of \cref{lem:single:dmmf}]
    Let $\Blk_j[t] = 1$ indicate that $\Blk[t] = 1$ because of agent $j \ne i$. Let $T_j$ be the last round before $T$ where $\Blk_j[t] = 1$. Then
    \begin{alignat*}{3}
        \Line{
            \sum_{t=1}^T \Blk[t]
        }{=}{
            \sum_{j\ne i} \sum_{t=1}^T \Blk_j[t]
            =
            \sum_{j\ne i} \sum_{t=1}^{T_j} \Blk_j[t]
        }{}
        \\
        \Line{}{\le}{
            \sum_{j\ne i} A_j[T_j]
        }{\blk_j[t] = 1 \implies j \text{ wins } t}
        \\
        \Line{}{\le}{
            \sum_{j\ne i} \frac{\a_j}{\a_i} (A_i[T_j] + 1)
        }{\blk_j[T_j] = 1 \implies j \text{ wins over $i$}}
        \\
        \Line{}{\le}{
            (A_i[T] + 1) \sum_{j\ne i} \frac{\a_j}{\a_i}
        }{A_i[T_j] \le A_i[T]}
        \\
        \Line{}{=}{
            (A_i[T] + 1)\frac{1-\a_i}{\a_i}
        }{\sum\nolimits_{j \in [n]} \a_j = 1}
    \end{alignat*}

    The proof is completed by noticing that $\sum_{t\in [T]} W[t] = A_i[T]$.
\end{proof}

%% file: appendixes/6.mult.tex
\section{Deferred Text of Section \ref{sec:mult}}
\label{sec:app:mult}

\subsection{Explanation of the definition of ideal utility for multi-round demands}
\label{ssec:app:mult:explain_ideal}

This section explains the definition of the $\b$-ideal utility for multi-round demands. First, we re-state the formal definition, \eqref{eq:mult:ideal}:
\begin{equation}
\label{eq:mult:ideal2}
\begin{split}
    \vstar_i(\b) \;=
    \quad
    \max_{\rho_\b:\R_+\times \N \to [0, 1]} \qquad &
    \frac{\Ex{ V K \rho_\b(V, K) }}{1 - \Ex{ \rho_\b(V, K) } + \Ex{ K \rho_\b(V, K) }}
    \\
    \textrm{such that} \qquad
    & \frac{\Ex{ K \rho_\b(V, K) }}{1 - \Ex{ \rho_\b(V, K) } + \Ex{ K \rho_\b(V, K) }} \le \b
\end{split}
\end{equation}

To understand the above optimization problem, it is easiest to first define $\req(V, K)$ as a Bernoulli random variable with mean $\rho_\b(V, K)$.
If we let $q = \Ex{\rho_\b(V, K)} = \Ex{\req(V, K)}$ then each fraction in \eqref{eq:mult:ideal2} becomes
\begin{equation}\label{eq:mult:51}
    \frac{\Ex{ V K \rho_\b(V, K) }}{1 - \Ex{ \rho_\b(V, K) } + \Ex{ K \rho_\b(V, K) }}
    =
    \frac{\Ex{ V K \;|\; \req(V, K) = 1 }}{\frac{1}{q} - 1 + \Ex{ K \rho_\b(V, K) \;|\; \req(V, K) = 1 }}
\end{equation}
and similarly
\begin{equation}\label{eq:mult:52}
    \frac{\Ex{ K \rho_\b(V, K) }}{1 - \Ex{ \rho_\b(V, K) } + \Ex{ K \rho_\b(V, K) }}
    =
    \frac{\Ex{ K \;|\; \req(V, K) = 1 }}{\frac{1}{q} - 1 + \Ex{ K \rho_\b(V, K) \;|\; \req(V, K) = 1 }}
\end{equation}

The denominator in the r.h.s. of \eqref{eq:mult:51} and \eqref{eq:mult:52} can be broken into two parts: $\frac{1}{q} - 1$ is the expected number of rounds before the agent asks for the item.
$\Ex{ K \rho_\b(V, K) | \req(V, K) = 1 }$ is the expected number of rounds the agent holds the item for a single request. 
These two quantities together define the expected number of rounds between the agent's requests.
Given that the numerators of \eqref{eq:mult:51} and \eqref{eq:mult:52} are the expected utility and duration of an agent's demand, respectively, each fraction defines the expected average of each quantity: \eqref{eq:mult:51} is the expected per-round utility and \eqref{eq:mult:52} is the probability that the agent holds the item in any single round.

We now briefly prove that the $\b$-ideal utility $\vstar(\b)$ is concave even in the case of reusable resources.

\begin{lemma}\label{lem:app:mult:concave}
    For every valuation $\calF$ on $(V, K)$ the $\b$-ideal utility $\vstar(\b)$ is concave.
\end{lemma}

\begin{proof}
    We prove the lemma when in the distribution $\calF$ the values $(V, K)$ can take finitely many values.
    The lemma for general distributions follows by taking an arbitrarily small finite cover.

    Assume that there are $m$ different values of $(V, K)$, denoted $(v_j, k_j)$ for $j \in [m]$.
    $(v_j, k_j)$ appears with probability $p_i$.
    In this case we can transform \cref{eq:mult:ideal2} into an LP as in \cite{DBLP:conf/sigecom/BanerjeeFT23}:
    \begin{equation*}
    \begin{split}
        \vstar(\b)
        =
        \max_{\{f_j\}_{j\in [m]}} \qquad
        & \sum_{j \in [m]} v_j k_j f_j
        \\
        \textrm{such that} \qquad
        & \sum_{j \in [m]} k_j f_j \le \b
        \\
        & 0 \le f_j \le p_j \left( 1 - \sum_{j' \in [m]} (k_{j'} - 1) f_{j'} \right)\qquad\forall\,j \in [m]
    \end{split}
    \end{equation*}
    where $f_j$ denotes the frequency with which $(v_j, k_j)$ is requested.
    We can re-write the above using the vector $\vec f$ and some constant (that only depend on $k_j, v_j, p_j$) matrices $A\in\R^{m\times m}$, $\vec a \in \R^m$, $\vec c \in \R^m$:
    \begin{equation*}
    \begin{split}
        \vstar(\b)
        =
        \max_{\vec f \in \R^m} \qquad
        & \vec c ^\top \vec f
        \\
        \textrm{such that} \qquad
        & \vec a^\top \vec f \le \b
        \\
        & A f \le \vec b
    \end{split}
    \end{equation*}
    where for simplicity we hide the $\vec f \ge 0$ constraint. Taking the dual, we get
    \begin{equation*}
    \begin{split}
        \vstar(\b)
        =
        \min_{g_0\in\R, \vec g\in\R^m} \qquad
        & \b g_0 + \vec b^\top \vec g
        \\
        \textrm{such that} \qquad
        & A^\top \vec g + g_0 \vec a \ge \vec c
    \end{split}
    \end{equation*}
    equivalently, by letting $G \subset \R^{m+1}$ be the set of extreme points of the polytope of the constraint $A^\top \vec g + g_0 \vec a \ge \vec c$, we get
    \begin{equation*}
        \vstar(\b)
        =
        \min_{(g_0,\vec g) \in G}\left( 
            \b g_0 + \vec b^\top \vec g
        \right)
    \end{equation*}

    The above proves that $\vstar(\b)$ is concave since we can write it as the minimum of concave functions.
\end{proof}

\subsection{Deferred proof of Theorem \ref{thm:mult:guar}}
\label{ssec:app:mult:proof}

This section presents the proof of \cref{thm:mult:guar}.

In addition to the rounds where the agent gets blocked by the adversary because the fairness criterion does not favor her, we also need to account for the rounds where the agent gets blocked because the adversary holds the item from some previous round.
To this end, we define $\Blk[t] = 1$ to indicate that the adversary has the item in a round $t$ and that the agent could not have won the item if she had requested it in that round.
This includes rounds when the agent $i$ is not favored by the fairness criterion and rounds when the adversary had won a multi-round demand on a previous round.

We first present the generalization of \cref{lem:single:dmmf} for reusable resources.

Using the above definition for $\blk$ we show a weaker version of \cref{lem:single:dmmf}. We bound the number of rounds the agent gets blocked by the maximum of the number of rounds the agent has won (as in \cref{lem:single:dmmf}) and the upper bound for multiple round demands of the adversary, $(1-\a)T/r$.

\begin{lemma}\label{lem:mult:dmmf}
    Fix the parameter $r \ge 1$. For any strategy followed by the agent, with probability $1$,
    \begin{equation*}
        \sum_{t=1}^T \Blk[t]
        \le
        \max\left\{
            \frac{1-\a}{\a} \left( \kmax + \sum_{t=1}^T W[t] K[t] \right)
            ,
            (1 - \a) T / r
        \right\}
    \end{equation*}
\end{lemma}

The proof is similar to the one in \cref{lem:single:dmmf}.
If $\sum_{t=1}^T \Blk[t] > (1 - \a) T / r$ (otherwise the statement holds), we can reason that the last time the agent gets blocked is due to a single round demand.
This allows us to use the fairness criterion to prove that the number of times the adversary has won the item cannot be much more than the agent's.

\begin{proof}[Proof of \cref{lem:mult:dmmf}]
    Fix an agent $j\ne i$. Let $\Blk_j[t] = 1$ indicate that $\Blk[t] = 1$ because of agent $j$. We are going to prove that
    \begin{equation}\label{eq:mult:1}
        \sum_{t=1}^T \Blk_j[t]
        \le
        \frac{\a_j}{\a_i} \max\left\{
            \kmax + A_i[T]
            ,
            \a_i T/r
        \right\}
    \end{equation}
    and the theorem follows by summing over $j\ne i$.

    If $\sum_{t=1}^T \Blk_j[t] \le \a_j T / r$ then \eqref{eq:mult:1} is obviously true.
    Assume now that $\sum_{t=1}^T \Blk_j[t] > \a_j T/r$ and let $T_j$ be the last round where $\Blk_j[t] = 1$.
    This implies that $A_j[T_j] \ge \sum_{t=1}^T \Blk_j[t] > \a_j T/r$
    Because of the restriction on the set $\mathcal N'$, this implies that agent $j$ won the good in round $T_j$ by making a single round duration request.
    This implies that in $T_j$ agent $i$ got blocked by $j$ because the fairness criterion favored $j$ over $i$:
    \begin{equation*}
        \frac{A_j[T_j]}{a_j}
        \le
        \frac{A_i[T_j] + d_i[T_j]}{a_i}
    \end{equation*}

    This \eqref{eq:mult:1} by noticing that $A_j[T_j] \ge \sum_{t=1}^T \Blk_j[t]$ and $A_i[T_j] + d_i[T_j] \le A_i[T] + \kmax$.
\end{proof}

We now prove the lemma that relates to the agent's expected utility and the expected duration of her request. 

\begin{lemma}\label{lem:mult:bpb}
    Fix $\b \in (0, 1]$.
    Assume that the agent requests every round with probability $\rho_\b^*(V[t], K[t])$, i.e., $\req[t]$ is a Bernoulli random variable with mean $\rho_\b^*(V[t], K[t])$.
    Then, for every round $t$, it holds that
    \begin{equation*}
        \Ex{U[t]}
        =
        \frac{\vstar(\b)}{\b'} \Ex{K[t] W[t]}
    \end{equation*}
    where $\b' = \frac{\Ex{ K[t] \req[t] }}{1 - q + \Ex{ K[t] \req[t] }}$ is the expected frequency of rounds the agent holds the policy under $\rho_\b^*$.
\end{lemma}

\begin{proof}
    We have that
    \begin{alignat*}{3}
        \Line{
            \Ex{U[t]}
        }{=}{
            \Ex{V[t] K[t] W[t]}
        }{}
        \\
        \Line{}{=}{
            \Ex{V[t] K[t] \big| W[t] = 1}
            \Ex{W[t]}
        }{}
        \\
        \Line{}{=}{
            \Ex{V[t] K[t] \big| \req[t] = 1}
            \Ex{W[t]}
        }{}
    \end{alignat*}
    where the last inequality holds because $V[t]$ and $K[t]$ do not directly depend on the probability that the agent wins but rather on her requesting the item.
    Similarly, we show that
    \begin{alignat*}{3}
        \Line{
            \Ex{K[t] W[t]}
        }{=}{
            \Ex{K[t] \big| W[t] = 1}
            \Ex{W[t]}
        }{}
        \\
        \Line{}{=}{
            \Ex{K[t] \big| \req[t] = 1}
            \Ex{W[t]}
        }{}
    \end{alignat*}

    Combining the two above equations and the definition of ideal utility (see the \eqref{eq:mult:ideal2}, \eqref{eq:mult:51}, and \eqref{eq:mult:52}), we get the lemma.
\end{proof}

We now prove the main result of \cref{sec:mult}, \cref{thm:mult:guar}.

\begin{proof}[Proof of \cref{thm:mult:guar}]
    Let $W'[t] = 1$ indicate that the agent wins round $t$ by requesting the item if she is never affected by the limit on multi-round demands of \cref{algo:mult} and if she is able to make multi-round demands even if there are not enough rounds left.
    Because for $t < T - \kmax$ it holds $W'[t] = W[t]$ until the round when the agent's allocation reaches $a T / r - \kmax$, it holds that
    \begin{equation}\label{eq:mult:61}
        \sum_{t = 1}^T K[t]W[t]
        \ge
        \min\left\{\a T/r,\sum_{t = 1}^T K[t]W'[t]\right\} - 2\kmax    
        .
    \end{equation}
    
    Let $P[t] = 1$ indicate that the agent is blocked by her own demand in some previous round: $P[t] = 1$ if for some round $t' < t$ the agent ``won'' round $t'$ ($W'[t'] = 1$) and her duration covered round $t$ ($K[t'] > t - t'$).
    It holds that
    \begin{equation*}
        W'[t]
        =
        \req[t](1 - \blk[t])(1 - P[t])
        =
        \req[t](1 - \blk[t] - P[t])
    \end{equation*}
    where in the last inequality we used the fact that it cannot hold that $\blk[t] = P[t] = 1$. Since $\req[t]$ and $K[t]$ are independent from $\blk[t]$ and $P[t]$ it holds that
    \begin{equation*}
        \Ex{K[t] W'[t]}
        =
        \Ex{K[t]\req[t]}(1 - \Ex{\blk[t]} - \Ex{P[t]})
        =
        \frac{(1-q)\b}{1-\b}
        (1 - \Ex{\blk[t]} - \Ex{P[t]})
    \end{equation*}
    where $q = \Ex{\req[t]}$ and $\b' = \frac{\Ex{ K[t] \req[t] }}{1 - q + \Ex{ K[t] \req[t] }}$. Note that $\b' \le \b$. Now we show a high probability bound: for every $\d > 0$,
    \begin{equation*}
        \Pr{
            \sum_t K[t]W'[t]
            \le
            \frac{(1-q)\b'}{1-\b'} \sum_t(1 - \blk[t] - P[t])
            - \kmax\sqrt{2T \log(1/\d)}
        }
        \le
        \d
    \end{equation*}
    We show this using the standard Azuma-Hoeffding inequality.
    Now we use in the above equation the fact that $\sum_t P[t] = \sum_t W'[t](K[t] - 1)$ and re-arrange to get 
    \begin{equation}\label{eq:mult:2}
        \Pr{
            \sum_t W'[t]\left( K[t] + \frac{(1 - q)\b'}{1 - \b'}(K[t] - 1) \right)
            \le
            \frac{(1 - q)\b'}{1 - \b'}
            \left(T - \sum_t\blk[t]\right)
            - \kmax\sqrt{2T \log(1/\d)}
        }
        \le
        \d
    \end{equation}

    Now we examine the expectation of the l.h.s. of the inequality of the probability in \eqref{eq:mult:2} for a round $t$ which can be re-written as
    \begin{equation*}
        \Ex{W'[t]\left( K[t] + \frac{(1 - q)\b'}{1 - \b'}(K[t] - 1) \right)}
        =
        q
        \Ex{W'[t]\left( K[t] + \frac{(1 - q)\b'}{1 - \b'}(K[t] - 1) \right) \bigg| \req[t] = 1}
    \end{equation*}

    In the above we use the conditional independence of $W'[t]$ and $K[t]$ given the event $\req[t] = 1$ along with the fact that $\Ex{K[t] | \req[t] = 1} = \frac{(1 - q)\b'}{q(1 - \b')}$ to get
    \begin{alignat*}{3}
        \Line{
            \Ex{W'[t]\left( K[t] + \frac{(1 - q)\b'}{1 - \b'}(K[t] - 1) \right)}
        }{=}{
            q
            \Ex{W'[t] \big| \req[t] = 1}
            \frac{(1-q)^2 \b'}{q (1-\b')^2}
        }{}
        \\
        \Line{}{=}{
            q
            \Ex{W'[t] \big| \req[t] = 1}
            \Ex{K[t] \big| \req[t] = 1}
            \frac{1-q}{1-\b'}
        }{}
        \\
        \Line{}{=}{
            \Ex{W'[t] K[t]}
            \frac{1-q}{1-\b'}
        }{}
    \end{alignat*}

    Using the above, we get another high probability bound: for every $\d > 0$
    \begin{equation}\label{eq:mult:3}
        \Pr{
            \frac{1-q}{1-\b'} \sum_t W'[t] K[t]
            \le
            \sum_t \Ex{W'[t]\left( K[t] + \frac{(1 - q)\b'}{1 - \b'}(K[t] - 1) \right)}
            - \kmax\sqrt{2T \log(1/\d)}
        }
        \le
        \d
    \end{equation}

    Combining with a union bound \eqref{eq:mult:2} and \eqref{eq:mult:3} we get that for all $\d > 0$
    \begin{equation*}
        \Pr{
            \frac{1-q}{1-\b'} \sum_t W'[t] K[t]
            \le
            \frac{(1 - q)\b'}{1 - \b'}
            \left(T - \sum_t\blk[t]\right)
            - \kmax\sqrt{8 T \log(1/\d)}
        }
        \le
        2\d
    \end{equation*}
    or equivalently that
    \begin{equation*}
        \Pr{
            \sum_t W'[t] K[t]
            \le
            \b'
            \left(T - \sum_t\blk[t]\right)
            - \kmax \frac{1-\b'}{1-q} \sqrt{8T \log(1/\d)}
        }
        \le
        2\d
    \end{equation*}

    We now use \eqref{eq:mult:61} in the above to get that with probability at least $1 - 2\d$ it holds that 
    \begin{equation*}
        \sum_t W'[t] K[t]
        \ge
        \min\left\{
            T \frac{\a}{r} ,
            \b'
            \left(T - \sum_t\blk[t]\right)
        \right\}
        - \kmax \left( \frac{1-\b'}{1-q} \sqrt{8T \log(1/\d)} + 2 \right)
    \end{equation*}
    
    \cref{lem:mult:dmmf} in the above implies that with probability at least $1 - 2\d$
    \begin{align*}
        \sum_t W[t] K[t]
        & \ge
        \min\left\{
            T \frac{\a}{r} ,
            \b' T - \b' \frac{1-\a}{\a} \sum_t W[t] K[t],
            \b' T \left( 1 -\frac{1-\a}{r}\right)
        \right\}
        \\
        & \qquad
        - \kmax \left( \frac{1-\b'}{1-q} \sqrt{8T \log(1/\d)} + 2 + \b' \frac{1-\a}{\a}\right)
    \end{align*}

    Each of the three terms in the minimum provides a bound for $\sum_t W[t] K[t]$. More specifically, the second one implies that
    \begin{equation*}
        \sum_t W[t] K[t]
        \ge
        \b' T - \b' \frac{1-\a}{\a} \sum_t W[t] K[t]
        - \kmax \left( \frac{1-\b'}{1-q} \sqrt{8T \log(1/\d)} + 2 + \b' \frac{1-\a}{\a}\right)
    \end{equation*}

    Solving the above using $\sum_t W[t] K[t]$ and substituting in the original bound we get
    \begin{align*}
        \sum_t W[t] K[t]
        & \ge
        T \min\left\{
            \frac{\a}{r} ,
            \frac{\b'}{\a + \b' - \a\b'},
            \b' \left( 1 -\frac{1-\a}{r}\right)
        \right\}
        \\
        & \qquad
        - \kmax \left( \frac{1-\b'}{1-q} \sqrt{8T \log(1/\d)} + 2 + \b' \frac{1-\a}{\a}\right)
    \end{align*}

    We can simplify the above if we use the fact that $q \le b'$ and that the middle of the three terms in the minimum can never be strictly less by both the other two
    \begin{align*}
        \sum_t W[t] K[t]
        \ge
        T \min\left\{
            \frac{\a}{r} ,
            \b' \left( 1 -\frac{1-\a}{r}\right)
        \right\}
        - \kmax \left( \sqrt{8T \log(1/\d)} + 2 + \b' \frac{1-\a}{\a}\right)
    \end{align*}

    Now, taking the expectation of the above, we get
    \begin{align*}
        \sum_t\Ex{W[t] K[t]}
        \ge
        T \min\left\{
            \frac{\a}{r} ,
            \b' \left( 1 -\frac{1-\a}{r}\right)
        \right\}
        - \kmax O\left( \sqrt T \right)
    \end{align*}
    
    Using \cref{lem:mult:bpb} we get
    \begin{align*}
        \sum_t\Ex{U[t]}
        \ge
        \vstar(\b) T \min\left\{
            \frac{\a}{\b' r} ,
            \left( 1 -\frac{1-\a}{r}\right)
        \right\}
        - \kmax O\left( \sqrt T \right)
    \end{align*}
    which proves the lemma using the fact that $\b' \le \b$.
\end{proof}

%% file: appendixes/7.imposs.tex
\section{Deferred proofs of Section \ref{sec:imp}}
\label{sec:app:imp}

In this section, we include the deferred proofs of the theorems of \cref{sec:imp}.

\begin{proof}[Proof of \cref{thm:imp:ind}]
    Consider $n$ agents with equal fair shares, $\a_i = 1/n$ for all $i \in [n]$.
    All the agents have independent values across time, sampled by the same Bernoulli distribution: for all $i,t$, independent of other rounds and agents
    \begin{equation*}
        V_i[t] =
        \begin{cases}
            1, & \text{ with probability } \frac{1}{n}
            \\
            0, & \text{ with probability } 1 - \frac{1}{n}
        \end{cases}
    \end{equation*}

    It is not hard to notice that for any allocation (and therefore mechanism) in any round $t$
    \begin{equation*}
        \Ex{\sum_i U_i[t]}
        \le
        1 - \qty(1 - \frac{1}{n})^n
    \end{equation*}
    since the probability that every agent will have a zero value in round $t$ is $\qty(1 - \frac{1}{n})^n$.
    Summing over $t \in [T]$ the above gives us
    \begin{equation*}
        \sum_i \sum_t \Ex{U_i[t]}
        \le
        \qty( 1 - \qty(1 - \frac{1}{n})^n ) T
    \end{equation*}
    which in turn implies
    \begin{equation*}
        \min_i \sum_t \Ex{U_i[t]}
        \le
        \frac{1}{n} \qty( 1 - \qty(1 - \frac{1}{n})^n ) T
    \end{equation*}

    Since the ideal utility of every agent is $\vstar_i(1/n) = 1/n$, this proves that there must exist an agent $i$ such that
    \begin{equation*}
        \sum_t \Ex{U_i[t]}
        \le
        \qty( 1 - \qty(1 - \frac{1}{n})^n ) \vstar_i(1/n)T
    \end{equation*}

    Since $1 - \qty(1 - \frac{1}{n})^n \to 1 - \frac 1 e$ as $n \to \infty$ we get the theorem.
\end{proof}

\begin{proof}[Proof of \cref{thm:imp:dep}]
    The agent's Markov chain has $2$ states, with the following (right stochastic) transition matrix


    \begin{figure}[ht]
        \centering
        \input{figures/mc_a.tex}
        \caption{Markov chain for the proof of \cref{thm:imp:dep}.}
    \end{figure}

    Regarding the agent's value, it holds that
    \begin{equation*}
        V[t] =
        \begin{cases}
            1, & \text{ if } s[t] = 1
            \\
            0, & \text{ if } s[t] = 2
        \end{cases}
    \end{equation*}

    It is not hard to calculate that the stationary distribution is $\pi = (\a, 1-\a)$ which makes $\vstar(\a) = \a$. Let $\zeta = 1 - \g$. 
    For any integer $k \ge 1$, we can calculate
    \begin{equation*}
        P^k =
        \begin{pmatrix}
            \a + \zeta^k(1 - \a) & (1 - \zeta^k)(1 - \a) \\
            (1 - \zeta^k)\a & 1 - (1 - \zeta^k)\a
        \end{pmatrix}
    \end{equation*}

    Let $\ell = \lfloor \frac{1-\a}{\a} \rfloor \ge \frac{1-\a}{\a} - 1$. 
    The adversary will request the next $\ell$ rounds whenever the agent wins a round.
    By the properties of the \dmmf mechanism, the adversary will win the item in each of those rounds, regardless of the agent's strategy.

    Let $t$ be a round where the agent gets positive utility, i.e., $V[t] = 1$ (and therefore $S[t] = 1$), and she wins. Let $t + L$ be the first round after $t$ that the agent has positive value and the adversary is not winning that round, i.e., $L = \min\{\tau > \ell: s_{t + \tau} = 1\}$. We notice that 
    \begin{alignat*}{3}
        \Line{
            \Ex{L}
        }{=}{
            (\ell + 1)\Pr{L = \ell + 1}
            +
            \sum_{k = \ell + 2}^\infty k\Pr{L = k}
        }{}
        \\
        \Line{}{=}{
            (\ell + 1)\Pr{L = \ell + 1}
            +
            \sum_{k = \ell + 2}^\infty k
                \Pr{s_{t + \ell + 1} = \ldots = s_{t+k-1} = 2}
                \Pr{s_{t+k} = 1 | s_{t+k-1} = 2} 
        }{}
        \\
        \Line{}{=}{
            (\ell + 1) \left( \a + \zeta^{\ell + 1}(1-\a) \right)
            +
            \sum_{k = \ell + 2}^\infty k
                (1-\zeta^{\ell+1}) (1 - \a) \left( 1 - (1-\zeta)\a \right)^{k - \ell - 2}
                (1-\zeta)\a
        }{}
        \\
        \Line{}{=}{
            \frac{1 - \a \zeta \left(\ell + 1 - \zeta^\ell \right) + \a\ell - \zeta^{\ell + 1}}
            {\a (1 - \zeta)}
        }{}
        \\
        \Line{}{\ge}{
            \frac{(1 - \a) \left(\zeta(2 - \zeta) - \zeta^{1/\a} \right)}
            {\a  \zeta (1 - \zeta)}
        }{}
    \end{alignat*}
    where in the last inequality we used that $\ell \ge \frac{1-\a}{\a} - 1$.
    This entails that the frequency with which the agent can win rounds with positive value (and therefore her expected per-round utility) is at most
    \begin{equation*}
        \frac{1}{\Ex{L}}
        \le
        \frac{\zeta(1-\zeta)}{(1-\a)\left(\zeta(2-\zeta) - \zeta^{1/\a}\right)} \a
    \end{equation*}
    
    Substituting $\zeta = 1 - \g$ and using the fact that this argument is not correct for the last time the agent wins (which introduces the additive $O(1)$ error) proves the main part of the theorem.

    For the last part of the theorem, it is easy to see that if $\a \le 1/2$, then $\min_s \pi(s) = \a$.
\end{proof}

\begin{proof}[Proof of \cref{cor:imp:markov}]
    We want to bound
    \begin{equation}\label{eq:imp:21}
        \frac{\g}{(1-\a)\left( 1 + \g - (1 - \g)^{\frac{1 - \a}{\a}} \right)}
    \end{equation}

    For part (1), if we take $\a \to 0$ then, since $1 - \g < 1$, it holds that $(1 - \g)^{\frac{1-\a}{\a}} \to 0$. This proves the desired bound.

    For part (2), we have that
    \begin{equation*}
        \lim_{\g\to 0}
        \frac{\g}{1 + \g - (1 - \g)^{\frac{1 - \a}{\a}}}
        =
        \lim_{\g\to 0}
        \frac{1}{ 1 + \frac{1 - \a}{\a}(1 - \g)^{\frac{1 - \a}{\a}-1}}
        =
        \frac{1}{1 + \frac{1 - \a}{\a}}
        =
        \a
    \end{equation*}
    where in the second equality we used L'Hopital's rule. This proves the desired bound
\end{proof}


\begin{proof}[Proof of \cref{thm:mult:imp}]
    Fix a $\b\in (0, 1]$.
    The agent has $(V[t],K[t]) = (1, 1)$ with probability $\b \in (0, 1]$; otherwise her value is $0$. It is not hard to see that $\vstar(\b) = \b$.
    
    We will assume that the adversary asks for $\kmax$ rounds every round.
    Let $\blk[t] = 1$ denote that the adversary holds the item from a previous round, which means the agent cannot win the item in round $t$ even if she requests it.
    
    We notice that the agent does not have anything to gain by requesting multiple rounds: requesting $k$ rounds on round $t$ and winning can be simulated by requesting the item for a single round on rounds $t,t+1,\ldots,t+k-1$, since the agent will satisfy the fairness constrain on all the rounds. Thus, from now on, we assume that the agent only makes single-round demands.
    Using this, we have:
    \begin{alignat*}{3}
        \Line{
            U[t]
        }{=}{
            V[t] \One{\text{agent wins}}
        }{}
        \\
        \Line{}{=}{
            V[t]
            \One{\text{agent requests}}
            (1 - \blk[t])
            \One{\text{\dmmf criterion favors agent}}
        }{}
        \\
        \Line{}{\le}{
            V[t] (1 - \blk[t])
        }{}
    \end{alignat*}

    Since $V[t]$ is independent from $\blk[t]$ we have that
    \begin{equation*}
        \Ex{U[t]}
        \le
        \b \Ex{1 - \blk[t]}
        =
        \vstar(\b) \Ex{1 - \blk[t]}
    \end{equation*}

    We are going to prove that $\sum_t \blk[t] \ge \frac{\kmax - 1}{\kmax} \frac{1-\a}{r} T - (\kmax - 1)$. The above inequality implies the theorem by summing over all $t$.

    To prove that $\sum_t \blk[t] \ge \frac{\kmax - 1}{\kmax} \frac{1-\a}{r} T - (\kmax - 1)$, first let $A[t]$ and $A'[t]$ be the total allocations of the agent and adversary, respectively. We notice that
    \begin{equation*}
        \sum_{t = 1}^T \blk[t]
        =
        \frac{\kmax - 1}{\kmax} A'[T]
    \end{equation*}

    Now proving that $\sum_t \blk[t] \ge \frac{\kmax - 1}{\kmax} \frac{1-\a}{r} T - (\kmax - 1)$ is equivalent with proving $A'[T] + \kmax \ge \frac{1-\a}{r}T$.

    Toward a contradiction assume that $A'[T] + \kmax < \frac{1-\a}{r}T$.
    This implies that the item gets allocated every round up to round $T' := T - \kmax + 1$: since the adversary's demands never go above her limit, $\frac{1-\a}{r}T$, the adversary requests the item every round (except possible the last $\kmax - 1$ rounds), meaning that it gets allocated to either the agent or the adversary.
    This implies that $A[T'] + A'[T'] \ge T'$.
    Now let $t$ be the final round in which the agent won the item while the adversary also made a request. From the fairness criterion of the \dmmf mechanism, we have that
    \begin{equation*}
        \frac{A[t]}{\a}
        \le
        \frac{A'[t] + \kmax}{1 - \a}
    \end{equation*}
    
    Since $t$ is the final round when the agent won by competing with the adversary and the adversary asks the item in every round except the final $\kmax - 1$ ones, it holds that $t \le T'$ making
    $A[T'] = A[t]$ and $A'[T'] \ge A'[t]$. This makes the above
    \begin{equation*}
        \frac{A[T']}{\a}
        \le
        \frac{A'[T'] + \kmax}{1 - \a}
    \end{equation*}

    Using the fact that $A[T'] + A'[T'] \ge T'$ the above becomes
    \begin{equation*}
        \frac{T' - A'[T']}{\a}
        \le
        \frac{A'[T'] + \kmax}{1 - \a}
    \end{equation*}
    which implies
    \begin{equation*}
        A'[T']
        \ge
        T (1 - \a) - \a \kmax
    \end{equation*}

    This leads to a contradiction since $A'[T'] = A'[T]$, $\a < 1$ and we had assumed that $A'[T] + \kmax \le T \frac{1 - \a}{r}$ and $r \ge 1$.
\end{proof}

%% file: figures/mc_a.tex
\begin{tikzpicture}[->, >=stealth', node distance=2cm, line width=0.6pt]
    \node[circle, draw](one){1};
    \node[circle, draw](two)[right of=one]{2};
    \path(one) edge[bend left] node[above]{$\g(1-\a)$} (two);
    \path(one) edge[loop left] node[left]{$1 - \g(1-\a)$} (one);
    \path(two) edge[bend left] node[below]{$\g\a$} (one);
    \path(two) edge[loop right] node[right]{$1 - \g\a$} (two);

    \node[](p)[left=3cm of one]{$
        P =
        \begin{pmatrix}
            1 - \g(1-\a) & \g(1-\a) \\
            \g\a & 1 - \g\a
        \end{pmatrix}
    $};
\end{tikzpicture}

%% file: ref.bib
@article{balseiro2015repeated,
  title     = {Repeated auctions with budgets in ad exchanges: Approximations and design},
  author    = {Balseiro, Santiago R and Besbes, Omar and Weintraub, Gabriel Y},
  journal   = {Management Science},
  volume    = {61},
  number    = {4},
  pages     = {864--884},
  year      = {2015},
  publisher = {INFORMS}
}

@inproceedings{banerjee2022online,
  author    = {Siddhartha Banerjee and
               Vasilis Gkatzelis and
               Artur Gorokh and
               Billy Jin},
  editor    = {Joseph (Seffi) Naor and
               Niv Buchbinder},
  title     = {Online Nash Social Welfare Maximization with Predictions},
  booktitle = {Proceedings of the 2022 {ACM-SIAM} Symposium on Discrete Algorithms,
               {SODA} 2022, Virtual Conference / Alexandria, VA, USA, January 9 -
               12, 2022},
  pages     = {1--19},
  publisher = {{SIAM}},
  year      = {2022},
  url       = {https://doi.org/10.1137/1.9781611977073.1},
  doi       = {10.1137/1.9781611977073.1},
  bibsource = {dblp computer science bibliography, https://dblp.org},
  address   = {Alexandria, VA, USA}
}

@article{banerjee2023online,
  title     = {Online Fair Allocation of Perishable Resources},
  author    = {Banerjee, Siddhartha and Hssaine, Chamsi and Sinclair, Sean R},
  journal   = {ACM SIGMETRICS},
  volume    = {51},
  number    = {1},
  pages     = {55--56},
  year      = {2023},
  publisher = {ACM}
}

@inproceedings{barman2022universal,
  author    = {Siddharth Barman and
               Arindam Khan and
               Arnab Maiti},
  title     = {Universal and Tight Online Algorithms for Generalized-Mean Welfare},
  booktitle = {Thirty-Sixth {AAAI} Conference on Artificial Intelligence, {AAAI}
               2022, Thirty-Fourth Conference on Innovative Applications of Artificial
               Intelligence, {IAAI} 2022, The Twelveth Symposium on Educational Advances
               in Artificial Intelligence, {EAAI} 2022 Virtual Event, February 22
               - March 1, 2022},
  pages     = {4793--4800},
  publisher = {{AAAI} Press},
  year      = {2022},
  url       = {https://doi.org/10.1609/aaai.v36i5.20406},
  doi       = {10.1609/AAAI.V36I5.20406},
  bibsource = {dblp computer science bibliography, https://dblp.org},
  address   = {Virtual Event}
}

@article{bergemann2010dynamic,
  issn      = {00129682, 14680262},
  url       = {http://www.jstor.org/stable/40664492},
  author    = {Dirk Bergemann and Juuso Välimäki},
  journal   = {Econometrica},
  number    = {2},
  pages     = {771--789},
  publisher = {[Wiley, Econometric Society]},
  title     = {The Dynamic Pivot Mechanism},
  urldate   = {2023-09-29},
  volume    = {78},
  year      = {2010}
}

@inproceedings{bonald2001impact,
  author    = {Thomas Bonald and
               Laurent Massouli{\'{e}}},
  editor    = {Mary K. Vernon},
  title     = {Impact of fairness on Internet performance},
  booktitle = {Proceedings of the Joint International Conference on Measurements
               and Modeling of Computer Systems, SIGMETRICS/Performance 2001, June
               16-20, 2001, Cambridge, MA, {USA}},
  pages     = {82--91},
  publisher = {{ACM}},
  year      = {2001},
  url       = {https://doi.org/10.1145/378420.378438},
  doi       = {10.1145/378420.378438},
  bibsource = {dblp computer science bibliography, https://dblp.org},
  address   = {Cambridge, MA, USA}
}

@article{bonald2006queueing,
  title     = {A queueing analysis of max-min fairness, proportional fairness and balanced fairness},
  author    = {Bonald, Thomas and Massouli{\'e}, Laurent and Proutiere, Alexandre and Virtamo, Jorma},
  journal   = {Queueing systems},
  volume    = {53},
  pages     = {65--84},
  year      = {2006},
  publisher = {Springer}
}

@article{budish2017course,
  title     = {Course match: A large-scale implementation of approximate competitive equilibrium from equal incomes for combinatorial allocation},
  author    = {Budish, Eric and Cachon, G{\'e}rard P and Kessler, Judd B and Othman, Abraham},
  journal   = {Operations Research},
  volume    = {65},
  number    = {2},
  pages     = {314--336},
  year      = {2017},
  publisher = {INFORMS}
}

@inproceedings{cavallo2014incentive,
  author    = {Ruggiero Cavallo},
  editor    = {Ana L. C. Bazzan and
               Michael N. Huhns and
               Alessio Lomuscio and
               Paul Scerri},
  title     = {Incentive compatible two-tiered resource allocation without money},
  booktitle = {International conference on Autonomous Agents and Multi-Agent Systems,
               {AAMAS} '14, Paris, France, May 5-9, 2014},
  pages     = {1313--1320},
  publisher = {{IFAAMAS/ACM}},
  year      = {2014},
  url       = {http://dl.acm.org/citation.cfm?id=2617457},
  bibsource = {dblp computer science bibliography, https://dblp.org},
  address   = {Paris, France}
}

@misc{dawson2013reserving,
  title     = {Reserving services within a cloud computing environment},
  author    = {Dawson, Christopher J and DiLuoffo, Vincenzo V and Kendzierski, Michael D and Seaman, James W},
  year      = {2013},
  publisher = {Google Patents},
  note      = {US Patent 8,615,584}
}

@book{DBLP:books/daglib/0070442,
  author    = {Drew Fudenberg and
               Jean Tirole},
  title     = {Game theory {(3.} pr.)},
  publisher = {{MIT} Press},
  year      = {1991},
  isbn      = {978-0-262-06141-4},
  bibsource = {dblp computer science bibliography, https://dblp.org}
}

@inproceedings{DBLP:conf/atal/KashPS13,
  author    = {Ian A. Kash and
               Ariel D. Procaccia and
               Nisarg Shah},
  editor    = {Maria L. Gini and
               Onn Shehory and
               Takayuki Ito and
               Catholijn M. Jonker},
  title     = {No agent left behind: dynamic fair division of multiple resources},
  booktitle = {International conference on Autonomous Agents and Multi-Agent Systems,
               {AAMAS} '13, Saint Paul, MN, USA, May 6-10, 2013},
  pages     = {351--358},
  publisher = {{IFAAMAS}},
  year      = {2013},
  url       = {http://dl.acm.org/citation.cfm?id=2484977},
  bibsource = {dblp computer science bibliography, https://dblp.org},
  address   = {Saint Paul, MN, USA}
}

@inproceedings{DBLP:conf/nsdi/GhodsiZHKSS10,
  author    = {Ali Ghodsi and
               Matei Zaharia and
               Benjamin Hindman and
               Andy Konwinski and
               Scott Shenker and
               Ion Stoica},
  editor    = {David G. Andersen and
               Sylvia Ratnasamy},
  title     = {Dominant Resource Fairness: Fair Allocation of Multiple Resource Types},
  booktitle = {Proceedings of the 8th {USENIX} Symposium on Networked Systems Design
               and Implementation, {NSDI} 2011, Boston, MA, USA, March 30 - April
               1, 2011},
  publisher = {{USENIX} Association},
  year      = {2011},
  address   = {Boston, MA, USA}
}

@inproceedings{DBLP:conf/osdi/GrandlKRAK16,
  author    = {Robert Grandl and
               Srikanth Kandula and
               Sriram Rao and
               Aditya Akella and
               Janardhan Kulkarni},
  title     = {{GRAPHENE:} Packing and Dependency-Aware Scheduling for Data-Parallel
               Clusters},
  booktitle = {12th {USENIX} Symposium on Operating Systems Design and Implementation,
               {OSDI} 2016, Savannah, GA, USA, November 2-4, 2016},
  pages     = {81--97},
  publisher = {{USENIX} Association},
  year      = {2016},
  url       = {https://www.usenix.org/conference/osdi16/technical-sessions/presentation/grandl\_graphene},
  bibsource = {dblp computer science bibliography, https://dblp.org},
  address   = {Savannah, GA, USA}
}

@inproceedings{DBLP:conf/osdi/ShueFS12,
  author    = {David Shue and
               Michael J. Freedman and
               Anees Shaikh},
  editor    = {Chandu Thekkath and
               Amin Vahdat},
  title     = {Performance Isolation and Fairness for Multi-Tenant Cloud Storage},
  booktitle = {10th {USENIX} Symposium on Operating Systems Design and Implementation,
               {OSDI} 2012, Hollywood, CA, USA, October 8-10, 2012},
  pages     = {349--362},
  publisher = {{USENIX} Association},
  year      = {2012},
  url       = {https://www.usenix.org/conference/osdi12/technical-sessions/presentation/shue},
  bibsource = {dblp computer science bibliography, https://dblp.org},
  address   = {Hollywood, CA, USA}
}

@inproceedings{DBLP:conf/osdi/VuppalapatiF0CK23,
  author    = {Midhul Vuppalapati and
               Giannis Fikioris and
               Rachit Agarwal and
               Asaf Cidon and
               Anurag Khandelwal and
               {\'{E}}va Tardos},
  editor    = {Roxana Geambasu and
               Ed Nightingale},
  title     = {Karma: Resource Allocation for Dynamic Demands},
  booktitle = {17th {USENIX} Symposium on Operating Systems Design and Implementation,
               {OSDI} 2023, Boston, MA, USA, July 10-12, 2023},
  pages     = {645--662},
  publisher = {{USENIX} Association},
  year      = {2023},
  url       = {https://www.usenix.org/conference/osdi23/presentation/vuppalapati},
  bibsource = {dblp computer science bibliography, https://dblp.org},
  address   = {Boston, MA, USA}
}

@inproceedings{DBLP:conf/sigecom/BabaioffEF21,
  author    = {Moshe Babaioff and
               Tomer Ezra and
               Uriel Feige},
  editor    = {P{\'{e}}ter Bir{\'{o}} and
               Shuchi Chawla and
               Federico Echenique},
  title     = {Fair-Share Allocations for Agents with Arbitrary Entitlements},
  booktitle = {{EC} '21: The 22nd {ACM} Conference on Economics and Computation,
               Budapest, Hungary, July 18-23, 2021},
  pages     = {127},
  publisher = {{ACM}},
  year      = {2021},
  url       = {https://doi.org/10.1145/3465456.3467559},
  doi       = {10.1145/3465456.3467559},
  bibsource = {dblp computer science bibliography, https://dblp.org},
  address   = {Budapest, Hungary}
}

@inproceedings{DBLP:conf/sigecom/BanerjeeFT23,
  author    = {Siddhartha Banerjee and
               Giannis Fikioris and
               {\'{E}}va Tardos},
  editor    = {Kevin Leyton{-}Brown and
               Jason D. Hartline and
               Larry Samuelson},
  title     = {Robust Pseudo-Markets for Reusable Public Resources},
  booktitle = {Proceedings of the 24th {ACM} Conference on Economics and Computation,
               {EC} 2023, London, United Kingdom, July 9-12, 2023},
  pages     = {241},
  publisher = {{ACM}},
  year      = {2023},
  url       = {https://doi.org/10.1145/3580507.3597723},
  doi       = {10.1145/3580507.3597723},
  bibsource = {dblp computer science bibliography, https://dblp.org},
  address   = {London, United Kingdom}
}

@inproceedings{DBLP:conf/sigecom/GaitondeT20,
  author    = {Jason Gaitonde and
               {\'{E}}va Tardos},
  editor    = {P{\'{e}}ter Bir{\'{o}} and
               Jason D. Hartline and
               Michael Ostrovsky and
               Ariel D. Procaccia},
  title     = {Stability and Learning in Strategic Queuing Systems},
  booktitle = {{EC} '20: The 21st {ACM} Conference on Economics and Computation,
               Virtual Event, Hungary, July 13-17, 2020},
  pages     = {319--347},
  publisher = {{ACM}},
  year      = {2020},
  url       = {https://doi.org/10.1145/3391403.3399491},
  doi       = {10.1145/3391403.3399491},
  bibsource = {dblp computer science bibliography, https://dblp.org},
  address   = {Virtual Event, Hungary}
}

@inproceedings{DBLP:conf/sigecom/GaitondeT21,
  author    = {Jason Gaitonde and
               {\'{E}}va Tardos},
  editor    = {P{\'{e}}ter Bir{\'{o}} and
               Shuchi Chawla and
               Federico Echenique},
  title     = {Virtues of Patience in Strategic Queuing Systems},
  booktitle = {{EC} '21: The 22nd {ACM} Conference on Economics and Computation,
               Budapest, Hungary, July 18-23, 2021},
  pages     = {520--540},
  publisher = {{ACM}},
  year      = {2021},
  url       = {https://doi.org/10.1145/3465456.3467640},
  doi       = {10.1145/3465456.3467640},
  bibsource = {dblp computer science bibliography, https://dblp.org},
  address   = {Budapest, Hungary}
}

@inproceedings{DBLP:conf/sigecom/GorokhBI21,
  author    = {Artur Gorokh and
               Siddhartha Banerjee and
               Krishnamurthy Iyer},
  editor    = {P{\'{e}}ter Bir{\'{o}} and
               Shuchi Chawla and
               Federico Echenique},
  title     = {The Remarkable Robustness of the Repeated Fisher Market},
  booktitle = {{EC} '21: The 22nd {ACM} Conference on Economics and Computation,
               Budapest, Hungary, July 18-23, 2021},
  pages     = {562},
  publisher = {{ACM}},
  year      = {2021},
  url       = {https://doi.org/10.1145/3465456.3467560},
  doi       = {10.1145/3465456.3467560},
  bibsource = {dblp computer science bibliography, https://dblp.org},
  address   = {Budapest, Hungary}
}

@article{DBLP:journals/teco/ParkesP015,
  author       = {David C. Parkes and
                  Ariel D. Procaccia and
                  Nisarg Shah},
  title        = {Beyond Dominant Resource Fairness: Extensions, Limitations, and Indivisibilities},
  journal      = {{ACM} Trans. Economics and Comput.},
  volume       = {3},
  number       = {1},
  pages        = {3:1--3:22},
  year         = {2015},
  url          = {https://doi.org/10.1145/2739040},
  doi          = {10.1145/2739040},
  bibsource    = {dblp computer science bibliography, https://dblp.org}
}

@inproceedings{DBLP:conf/sigmetrics/FreemanZCL18,
  author    = {Rupert Freeman and
               Seyed Majid Zahedi and
               Vincent Conitzer and
               Benjamin C. Lee},
  editor    = {Konstantinos Psounis and
               Aditya Akella and
               Adam Wierman},
  title     = {Dynamic Proportional Sharing: {A} Game-Theoretic Approach},
  booktitle = {Abstracts of the 2018 {ACM} International Conference on Measurement
               and Modeling of Computer Systems, {SIGMETRICS} 2018, Irvine, CA, USA,
               June 18-22, 2018},
  pages     = {33--35},
  publisher = {{ACM}},
  year      = {2018},
  url       = {https://doi.org/10.1145/3219617.3219631},
  doi       = {10.1145/3219617.3219631},
  bibsource = {dblp computer science bibliography, https://dblp.org},
  address   = {Irvine, CA, USA}
}

@inproceedings{DBLP:conf/wine/BabaioffEF22,
  author    = {Moshe Babaioff and
               Tomer Ezra and
               Uriel Feige},
  editor    = {Kristoffer Arnsfelt Hansen and
               Tracy Xiao Liu and
               Azarakhsh Malekian},
  title     = {On Best-of-Both-Worlds Fair-Share Allocations},
  booktitle = {Web and Internet Economics - 18th International Conference, {WINE}
               2022, Troy, NY, USA, December 12-15, 2022, Proceedings},
  series    = {Lecture Notes in Computer Science},
  volume    = {13778},
  pages     = {237--255},
  publisher = {Springer},
  year      = {2022},
  url       = {https://doi.org/10.1007/978-3-031-22832-2\_14},
  doi       = {10.1007/978-3-031-22832-2\_14},
  bibsource = {dblp computer science bibliography, https://dblp.org},
  address   = {Troy, NY}
}

@misc{DBLP:journals/corr/FikiorisAT21,
  author     = {Giannis Fikioris and
                Rachit Agarwal and
                {\'{E}}va Tardos},
  title      = {Incentives in Resource Allocation under Dynamic Demands},
  journal    = {CoRR},
  volume     = {abs/2109.12401},
  year       = {2021},
  url        = {https://arxiv.org/abs/2109.12401},
  eprinttype = {arXiv},
  eprint     = {2109.12401},
  bibsource  = {dblp computer science bibliography, https://dblp.org}
}

@inproceedings{devanur2015perfect,
  author    = {Nikhil R. Devanur and
               Yuval Peres and
               Balasubramanian Sivan},
  editor    = {Piotr Indyk},
  title     = {Perfect Bayesian Equilibria in Repeated Sales},
  booktitle = {Proceedings of the Twenty-Sixth Annual {ACM-SIAM} Symposium on Discrete
               Algorithms, {SODA} 2015, San Diego, CA, USA, January 4-6, 2015},
  pages     = {983--1002},
  publisher = {{SIAM}},
  year      = {2015},
  url       = {https://doi.org/10.1137/1.9781611973730.67},
  doi       = {10.1137/1.9781611973730.67},
  bibsource = {dblp computer science bibliography, https://dblp.org},
  address   = {San Diego, CA, USA}
}

@article{elokda2023self,
  title   = {A Self-Contained Karma Economy for the Dynamic Allocation of Common Resources},
  author  = {Elokda, Ezzat and Bolognani, Saverio and Censi, Andrea and D{\"o}rfler, Florian and Frazzoli, Emilio},
  journal = {Dynamic Games and Applications},
  pages   = {1--33},
  year    = {2023},
  month   = {4},
  day     = {25},
  issn    = {2153-0793},
  doi     = {10.1007/s13235-023-00503-0},
  url     = {https://doi.org/10.1007/s13235-023-00503-0},
  volume  = {13}
}

@inproceedings{feldman2016price,
  author    = {Michal Feldman and
               Nicole Immorlica and
               Brendan Lucier and
               Tim Roughgarden and
               Vasilis Syrgkanis},
  editor    = {Daniel Wichs and
               Yishay Mansour},
  title     = {The price of anarchy in large games},
  booktitle = {Proceedings of the 48th Annual {ACM} {SIGACT} Symposium on Theory
               of Computing, {STOC} 2016, Cambridge, MA, USA, June 18-21, 2016},
  pages     = {963--976},
  publisher = {{ACM}},
  year      = {2016},
  url       = {https://doi.org/10.1145/2897518.2897580},
  doi       = {10.1145/2897518.2897580},
  bibsource = {dblp computer science bibliography, https://dblp.org},
  address   = {Cambridge, MA, USA}
}

@inproceedings{freeman2017fair,
  author    = {Rupert Freeman and
               Seyed Majid Zahedi and
               Vincent Conitzer},
  editor    = {Carles Sierra},
  title     = {Fair and Efficient Social Choice in Dynamic Settings},
  booktitle = {Proceedings of the Twenty-Sixth International Joint Conference on
               Artificial Intelligence, {IJCAI} 2017, Melbourne, Australia, August
               19-25, 2017},
  pages     = {4580--4587},
  publisher = {ijcai.org},
  year      = {2017},
  url       = {https://doi.org/10.24963/ijcai.2017/639},
  doi       = {10.24963/IJCAI.2017/639},
  bibsource = {dblp computer science bibliography, https://dblp.org},
  address   = {Melbourne, Australia}
}

@article{gershkov2009dynamic,
  title     = {Dynamic revenue maximization with heterogeneous objects: A mechanism design approach},
  author    = {Gershkov, Alex and Moldovanu, Benny},
  journal   = {American economic Journal: microeconomics},
  volume    = {1},
  number    = {2},
  pages     = {168--198},
  year      = {2009},
  publisher = {American Economic Association}
}

@inproceedings{gkatzelis2023,
  title     = {Proportionally Fair Online Allocation of Public Goods with Predictions},
  author    = {Banerjee, Siddhartha and Gkatzelis, Vasilis and Hossain, Safwan and Jin, Billy and Micha, Evi and Shah, Nisarg},
  booktitle = {Proceedings of the Thirty-Second International Joint Conference on
               Artificial Intelligence, {IJCAI-23}},
  publisher = {International Joint Conferences on Artificial Intelligence Organization},
  pages     = {20--28},
  year      = {2023},
  month     = {8},
  address   = {Macao, S.A.R}
}

@inproceedings{gorokh2017,
  author    = {Artur Gorokh and
               Siddhartha Banerjee and
               Krishnamurthy Iyer},
  editor    = {Constantinos Daskalakis and
               Moshe Babaioff and
               Herv{\'{e}} Moulin},
  title     = {From Monetary to Non-Monetary Mechanism Design via Artificial Currencies},
  booktitle = {Proceedings of the 2017 {ACM} Conference on Economics and Computation,
               {EC} '17, Cambridge, MA, USA, June 26-30, 2017},
  pages     = {563--564},
  publisher = {{ACM}},
  year      = {2017},
  url       = {https://doi.org/10.1145/3033274.3085140},
  doi       = {10.1145/3033274.3085140},
  bibsource = {dblp computer science bibliography, https://dblp.org},
  address   = {Cambridge, MA, USA}
}

@misc{gummadi2012repeated,
  title      = {Optimal bidding strategies and equilibria in dynamic auctions with budget constraints},
  author     = {Gummadi, Ramki and Key, Peter and Proutiere, Alexandre},
  url        = {https://papers.ssrn.com/abstract_id=2066175},
  eprinttype = {SSRN},
  eprint     = {2066175},
  year       = {2013}
}

@inproceedings{guo2010,
  author    = {Mingyu Guo and
               Vincent Conitzer},
  editor    = {Wiebe van der Hoek and
               Gal A. Kaminka and
               Yves Lesp{\'{e}}rance and
               Michael Luck and
               Sandip Sen},
  title     = {Strategy-proof allocation of multiple items between two agents without
               payments or priors},
  booktitle = {9th International Conference on Autonomous Agents and Multiagent Systems
               {(AAMAS} 2010), Toronto, Canada, May 10-14, 2010, Volume 1-3},
  pages     = {881--888},
  publisher = {{IFAAMAS}},
  year      = {2010},
  url       = {https://dl.acm.org/citation.cfm?id=1838324},
  bibsource = {dblp computer science bibliography, https://dblp.org},
  address   = {Toronto, Canada}
}

@inproceedings{Iyer2014,
  author    = {Krishnamurthy Iyer and
               Ramesh Johari and
               Mukund Sundararajan},
  editor    = {Yoav Shoham and
               Yan Chen and
               Tim Roughgarden},
  title     = {Mean field equilibria of dynamic auctions with learning},
  booktitle = {Proceedings 12th {ACM} Conference on Electronic Commerce (EC-2011),
               San Jose, CA, USA, June 5-9, 2011},
  pages     = {339--340},
  publisher = {{ACM}},
  year      = {2011},
  url       = {https://doi.org/10.1145/1993574.1993631},
  doi       = {10.1145/1993574.1993631},
  bibsource = {dblp computer science bibliography, https://dblp.org},
  address   = {San Jose, CA, USA}
}

@article{jackson,
  title     = {Overcoming incentive constraints by linking decisions},
  author    = {Jackson, Matthew O and Sonnenschein, Hugo F},
  journal   = {Econometrica},
  volume    = {75},
  number    = {1},
  pages     = {241--257},
  year      = {2007},
  publisher = {Wiley Online Library}
}

@inproceedings{kawase2022online,
  author    = {Yasushi Kawase and
               Hanna Sumita},
  editor    = {Panagiotis Kanellopoulos and
               Maria Kyropoulou and
               Alexandros A. Voudouris},
  title     = {Online Max-min Fair Allocation},
  booktitle = {Algorithmic Game Theory - 15th International Symposium, {SAGT} 2022,
               Colchester, UK, September 12-15, 2022, Proceedings},
  series    = {Lecture Notes in Computer Science},
  volume    = {13584},
  pages     = {526--543},
  publisher = {Springer},
  year      = {2022},
  url       = {https://doi.org/10.1007/978-3-031-15714-1\_30},
  doi       = {10.1007/978-3-031-15714-1\_30},
  bibsource = {dblp computer science bibliography, https://dblp.org},
  address   = {Colchester, UK}
}

@article{kelly1998rate,
  title     = {Rate control for communication networks: shadow prices, proportional fairness and stability},
  author    = {Kelly, Frank P and Maulloo, Aman K and Tan, David Kim Hong},
  journal   = {Journal of the Operational Research society},
  volume    = {49},
  pages     = {237--252},
  year      = {1998},
  publisher = {Springer}
}

@inproceedings{leme2012sequential,
  author    = {Renato Paes Leme and
               Vasilis Syrgkanis and
               {\'{E}}va Tardos},
  editor    = {Yuval Rabani},
  title     = {Sequential auctions and externalities},
  booktitle = {Proceedings of the Twenty-Third Annual {ACM-SIAM} Symposium on Discrete
               Algorithms, {SODA} 2012, Kyoto, Japan, January 17-19, 2012},
  pages     = {869--886},
  publisher = {{SIAM}},
  year      = {2012},
  url       = {https://doi.org/10.1137/1.9781611973099.70},
  doi       = {10.1137/1.9781611973099.70},
  bibsource = {dblp computer science bibliography, https://dblp.org},
  address   = {Kyoto, Japan}
}

@inproceedings{nekipelov2015econometrics,
  author    = {Denis Nekipelov and
               Vasilis Syrgkanis and
               {\'{E}}va Tardos},
  editor    = {Tim Roughgarden and
               Michal Feldman and
               Michael Schwarz},
  title     = {Econometrics for Learning Agents},
  booktitle = {Proceedings of the Sixteenth {ACM} Conference on Economics and Computation,
               {EC} '15, Portland, OR, USA, June 15-19, 2015},
  pages     = {1--18},
  publisher = {{ACM}},
  year      = {2015},
  url       = {https://doi.org/10.1145/2764468.2764522},
  doi       = {10.1145/2764468.2764522},
  bibsource = {dblp computer science bibliography, https://dblp.org},
  address   = {Portland, OR, USA}
}

@article{prendergast2022allocation,
  title     = {The allocation of food to food banks},
  author    = {Prendergast, Canice},
  journal   = {Journal of Political Economy},
  volume    = {130},
  number    = {8},
  pages     = {1993--2017},
  year      = {2022},
  publisher = {The University of Chicago Press Chicago, IL}
}

@article{sinclair2022sequential,
  title     = {Sequential fair allocation: Achieving the optimal envy-efficiency tradeoff curve},
  author    = {Sinclair, Sean R and Banerjee, Siddhartha and Yu, Christina Lee},
  journal   = {ACM SIGMETRICS},
  volume    = {50},
  number    = {1},
  pages     = {95--96},
  year      = {2022},
  publisher = {ACM}
}

@misc{vasudevan2016customizable,
  title     = {Customizable model for throttling and prioritizing orders in a cloud environment},
  author    = {Vasudevan, Ramesh and Prathipati, Anjani Kalyan and Seetharam, Pradeep and Arun, Gopalan},
  year      = {2016},
  publisher = {Google Patents},
  note      = {US Patent 9,253,113}
}

@inproceedings{walsh2014allocation,
  author    = {Toby Walsh},
  editor    = {Carsten Lutz and
               Michael Thielscher},
  title     = {Allocation in Practice},
  booktitle = {{KI} 2014: Advances in Artificial Intelligence - 37th Annual German
               Conference on AI, Stuttgart, Germany, September 22-26, 2014. Proceedings},
  series    = {Lecture Notes in Computer Science},
  volume    = {8736},
  pages     = {13--24},
  publisher = {Springer},
  year      = {2014},
  url       = {https://doi.org/10.1007/978-3-319-11206-0\_2},
  doi       = {10.1007/978-3-319-11206-0\_2},
  bibsource = {dblp computer science bibliography, https://dblp.org},
  address   = {Stuttgart, Germany}
}

@article{yin2022optimal,
  title   = {Optimal Efficiency-Envy Trade-Off via Optimal Transport},
  author  = {Yin, Steven and Kroer, Christian},
  journal = {Advances in Neural Information Processing Systems (NeurIPS)},
  volume  = {35},
  pages   = {25644--25654},
  year    = {2022}
}

@article{DBLP:journals/pomacs/ImMMP20,
  author       = {Sungjin Im and
                  Benjamin Moseley and
                  Kamesh Munagala and
                  Kirk Pruhs},
  title        = {Dynamic Weighted Fairness with Minimal Disruptions},
  journal      = {Proc. {ACM} Meas. Anal. Comput. Syst.},
  volume       = {4},
  number       = {1},
  pages        = {19:1--19:18},
  year         = {2020},
  url          = {https://doi.org/10.1145/3379485},
  doi          = {10.1145/3379485},
  bibsource    = {dblp computer science bibliography, https://dblp.org}
}

@article{DBLP:journals/mor/VardiPF22,
  author       = {Shai Vardi and
                  Alexandros Psomas and
                  Eric J. Friedman},
  title        = {Dynamic Fair Resource Division},
  journal      = {Math. Oper. Res.},
  volume       = {47},
  number       = {2},
  pages        = {945--968},
  year         = {2022},
  url          = {https://doi.org/10.1287/moor.2021.1155},
  doi          = {10.1287/MOOR.2021.1155},
  bibsource    = {dblp computer science bibliography, https://dblp.org}
}

@inproceedings{DBLP:conf/aaai/TamuzVZ18,
  author       = {Omer Tamuz and
                  Shai Vardi and
                  Juba Ziani},
  editor       = {Sheila A. McIlraith and
                  Kilian Q. Weinberger},
  title        = {Non-Exploitable Protocols for Repeated Cake Cutting},
  booktitle    = {Proceedings of the Thirty-Second {AAAI} Conference on Artificial Intelligence,
                  (AAAI-18), the 30th innovative Applications of Artificial Intelligence
                  (IAAI-18), and the 8th {AAAI} Symposium on Educational Advances in
                  Artificial Intelligence (EAAI-18), New Orleans, Louisiana, USA, February
                  2-7, 2018},
  pages        = {1226--1233},
  publisher    = {{AAAI} Press},
  year         = {2018},
  url          = {https://doi.org/10.1609/aaai.v32i1.11472},
  doi          = {10.1609/AAAI.V32I1.11472},
  bibsource    = {dblp computer science bibliography, https://dblp.org}
}

@inproceedings{DBLP:conf/nips/BranzeiHPSW24,
  author       = {Simina Br{\^{a}}nzei and
                  MohammadTaghi Hajiaghayi and
                  Reed Phillips and
                  Suho Shin and
                  Kun Wang},
  editor       = {Amir Globersons and
                  Lester Mackey and
                  Danielle Belgrave and
                  Angela Fan and
                  Ulrich Paquet and
                  Jakub M. Tomczak and
                  Cheng Zhang},
  title        = {Dueling over Dessert, Mastering the Art of Repeated Cake Cutting},
  booktitle    = {Advances in Neural Information Processing Systems 38: Annual Conference
                  on Neural Information Processing Systems 2024, NeurIPS 2024, Vancouver,
                  BC, Canada, December 10 - 15, 2024},
  year         = {2024},
  url          = {http://papers.nips.cc/paper\_files/paper/2024/hash/b125999bde7e80910cbdbd323087df8f-Abstract-Conference.html},
  bibsource    = {dblp computer science bibliography, https://dblp.org}
}
